\documentclass[review]{elsarticle}

\usepackage{lineno,hyperref}
\modulolinenumbers[5]

\usepackage{float}
\usepackage{epstopdf}
\usepackage{graphicx}
\usepackage{caption}
\usepackage{subcaption}
\usepackage{amsmath,amsfonts}
\usepackage{amsthm}
\usepackage{setspace}
\usepackage{tikz} 					
\usetikzlibrary{fit,positioning}		
\usepackage{algorithm,algorithmic}
\definecolor{orange}{rgb}{1,.5,0}

\newtheorem{theorem}{Theorem}[section] 

\journal{Computational Statistics \& Data Analysis}

\begin{document}

\begin{frontmatter}

\title{Transdimensional Approximate Bayesian Computation for Inference on Invasive Species Models with Latent Variables of Unknown Dimension}

\author{Oksana A. Chkrebtii\fnref{myemail}}
\address{Department of Statistics, The Ohio State University, Columbus, OH, USA}
\fntext[myemail]{email: oksana@stat.osu.edu}

\author{Erin K. Cameron\fnref{currentaddress}}
\address{Department of Biological Sciences, University of Alberta, Edmonton, AB, Canada}
\fntext[currentaddress]{Current address: Metapopulation Research Group, Department of Biological and Environmental Sciences, 
PO Box 65 (Viikinkaari 1), 00014 University of Helsinki, Finland}

\author{David A. Campbell}
\address{Department of Statistics \& Actuarial Science, Simon Fraser University, Surrey, Canada}

\author{Erin M. Bayne}
\address{Department of Biological Sciences, University of Alberta, Edmonton, AB, Canada}


\begin{abstract}
Accurate information on patterns of introduction and spread of non-native species is essential for making predictions and management decisions. In many cases, estimating unknown rates of introduction and spread from observed data requires evaluating intractable variable-dimensional integrals. In general, inference on the large class of models containing latent variables of large or variable dimension precludes exact sampling techniques. Approximate Bayesian computation (ABC) methods provide an alternative to exact sampling but rely on inefficient conditional simulation of the latent variables.  To accomplish this task efficiently, a new transdimensional Monte Carlo sampler is developed for approximate Bayesian model inference and used to estimate rates of introduction and spread for the non-native earthworm species \textit{Dendrobaena octaedra} (Savigny) along roads in the boreal forest of northern Alberta.  Using low and high estimates of introduction and spread rates, the extent of earthworm invasions  in northeastern Alberta was simulated to project the proportion of suitable habitat  invaded in the year following data collection. 
\end{abstract}


\end{frontmatter}

\linenumbers


\section{Introduction}

Biological invasions are occurring at unprecedented rates worldwide \citep{Ricciardi2007} but often remain undetected until invading species have spread extensively.  Detailed records documenting the time of initial introduction and subsequent changes in distribution are therefore not available for many invasions \citep{Fang2005}.  This information is critical for projection of future expansion and for development of appropriate management strategies \citep{Abbott2006}.  However, in cases where long-term temporal data is unavailable, current distribution patterns of non-native species can be used to infer how the invasion process may have occurred \citep{Fang2005,McIntireFajardo2009}.

Invasions that occur below-ground, such as earthworm invasions, can be particularly difficult to track over time and may initially proceed unnoticed.   Non-native earthworms are currently spreading in many forests across North America \citep{FrelichEtAl2006, HendrixEtAl2008} but have limited ability to disperse actively \citep{MarinissenBosch1992}.  Therefore, passive jump dispersal via abandonment of bait by anglers and transport along roads via vehicle traffic is thought to be important in their spread \citep{GundaleEtAl2005, FrelichEtAl2006, CameronEtAl2007}.  In northern hardwood and boreal forests, which are devoid of native earthworm species, earthworm invasions are causing significant changes to nutrient cycling and soil structure \citep{BohlenEtAl2004, FrelichEtAl2006,HaleEtAl2006}.  These changes have led to cascading effects on songbird \citep{LossBlair2011} and plant communities \citep{ScheuParkinson1994, HaleEtAl2006, NuzzoEtAl2009}.   
Because earthworms can affect other species directly, as well as indirectly via changes in the physical environment, their invasions may cause substantial changes over a large area of the Canadian boreal forest in the future. Accurate estimates of earthworm introduction and spread rates in this region are thus critical for the development of appropriate management strategies.

For many invasive species, spread occurs via a combination of diffusive spread around invaded sites and jump dispersal to new locations (i.e., stratified diffusion; \cite{ShigesadaEtAl1995}).  Because even rare long-distance jump dispersal events result in faster spread than would be expected with diffusive spread alone \citep{ShigesadaEtAl1995}, estimates of both introduction and spread rates are typically needed to predict the future spatial extent of an invasive species. Introductions of invasive species can be described using point process models, while models for diffusive spread following an introduction are application-dependent and can be complex (e.g., \cite{IllianEtAl2009}). As a result, data likelihoods for such models are not often available in closed form.

Approximate Bayesian computation (ABC, \cite{DiggleGratton1984,BeaumontEtAl2002,MarjoramEtAl2003}) has proven to be a useful approach for approximate inference on intractable likelihood problems, including distinguishing among introduction scenarios and invasion routes of non-native species (e.g., \cite{MillerEtAl2005, LombaertEtAl2010}).  It relies on repeated model simulation in the absence of an explicit likelihood function.  Proposed parameter values are accepted or rejected based on the distance between low-dimensional summaries of a conditional model realization and the observed data.  

ABC algorithms are generally computationally intensive, and can result in unreliable estimates when too few proposals are accepted (e.g., \cite{FearnheadPrangle2012}, lemma 1). Existing implementations typically construct a Markov chain with a dependent proposal mechanism defined on a fixed probability space \citep{MarjoramEtAl2003}. However, in many models, the dimensionality of the parameter space varies \citep{Green1995,RichardsonGreen1997}.  
In this paper we develop a transdimensional ABC algorithm that allows efficient exploration of parameter subspaces of variable dimension.  This approach is applied to estimate introduction and spread rates of non-native earthworms in the boreal forest. 

The paper is organized as follows.  
Section \ref{sec:methods} describes a general point-process model for introduction of an invasive species and its subsequent spread.  As the likelihood function is an intractable integral of variable dimension, a basic ABC algorithm for obtaining a sample from an approximate posterior distribution associated with this model is described.  The novel transdimensional ABC approach is then presented as an efficient alternative to existing sampling methods.
Section \ref{sec:wormstudy} introduces the motivating problem of estimating the rates of introduction and spread of the earthworm species \textit{Dendrobaena octaedra} (Savigny) in northern Alberta by combining information from two datasets.  The hierarchical model of the population dynamics is described, and algorithmic implementation details are provided.  
Inference results are summarized in Section \ref{sec:results}, and are then utilized in a spatio-temporal simulation model to project the extent of invasion in a sample region in Section \ref{sec:map}.  Concluding remarks are presented in Section \ref{sec:conclusion}.


\section{Approximate Bayesian Computation}\label{sec:methods}

Let us consider a model for the invasion of a non-native species over a given region with the vector of unknown model parameters, $\boldsymbol\theta$, such as rates of introduction, birth, predation, or spread.  A stochastic mechanism generates $k\in \mathbb{N}$ introduction events on the spatio-temporal horizon $\mathcal{H}\subset\mathbb{R}^d$.  Denote by $\mbox{x}^{(i)}\in \mathcal{H}$ the location of $i$th, $1\leq i \leq k$, latent introduction event, and define the vector concatenation,
\[
{\bf x}_k = [\mbox{x}^{(1)}, \cdots , \mbox{x}^{(k)}]^\top \in \mathcal{H}^{k} \subset \mathbb{R}^{dk}.
\]

\noindent  The spatio-temporal spread resulting from the introductions ${\bf x}_k$ follows a deterministic or stochastic model, generating the data $Y\in\mathcal{Y}$. 
The dependence structure of model components on the parameters $\boldsymbol\theta$ is illustrated in  Figure \ref{fig:general}.

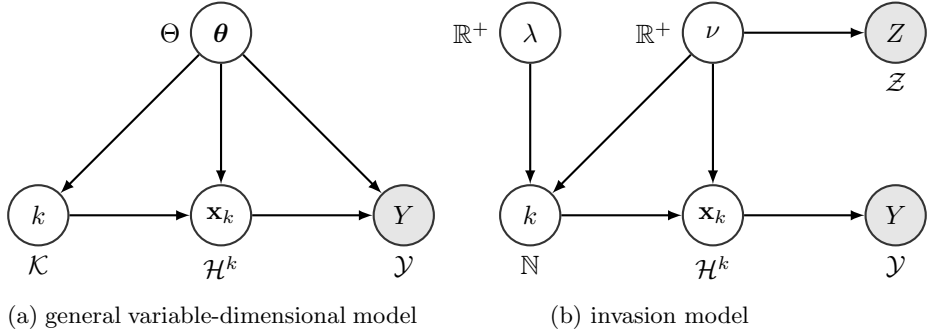
\begin{figure}
        \centering
        \begin{subfigure}[b]{0.45\textwidth}
                \centering
               \begin{tikzpicture}
			\tikzstyle{main}=[circle, minimum size = 8mm, thick, draw =black!80, node distance = 16mm]
			\tikzstyle{connect}=[-latex, thick]
			\tikzstyle{box}=[rectangle, draw=black!100]
 		 	\node[main, fill = white!100]  (k) [label=below:$\mathcal{K}$] {$k$};
  			\node[main] (x) [right=of k,label=below:$\mathcal{H}^k$] {${\bf x}_k$};
 		 	\node[main] (theta) [above=of x,label=left:$\Theta$] {$\boldsymbol\theta$};
 			\node[main, fill = black!10] (Y) [right=of x,label=below:$\mathcal{Y}$] {$Y$};
  			\path 	(theta) edge [connect] (k)
					(k) edge [connect] (x)
        				(theta) edge [connect] (x)
        				(x) edge [connect] (Y)
        				(theta) edge [connect] (Y);
		\end{tikzpicture}
                \caption{general variable-dimensional model}
                \label{fig:general}
        \end{subfigure}
        ~ 
        \begin{subfigure}[b]{0.45\textwidth}
                \centering
                \begin{tikzpicture}
			\tikzstyle{main}=[circle, minimum size = 8mm, thick, draw =black!80, node distance = 16mm]
			\tikzstyle{connect}=[-latex, thick]
			\tikzstyle{box}=[rectangle, draw=black!100]
  			\node[main, fill = white!100]  (k) [label=below:$\mathbb{N}$] {$k$};
  			\node[main] (x) [right=of k,label=below:$\mathcal{H}^k$] {${\bf x}_k$};
 	 		\node[main] (v) [above=of x,label=left:$\mathbb{R}^+$] {$\nu$};
  			\node[main] (lambda) [above=of k,label=left:$\mathbb{R}^+$] {$\lambda$};
  			\node[main, fill = black!10] (Y) [right=of x,label=below:$\mathcal{Y}$] {$Y$};
   			\node[main, fill = black!10] (Z) [above=of Y,label=below:$\mathcal{Z}$] {$Z$};
   			\path 	(v) edge [connect] (k)
					(k) edge [connect] (x)
        				(v) edge [connect] (x)
        				(x) edge [connect] (Y)
        				(lambda) edge [connect] (k)
					(v) edge [connect] (Z);
		\end{tikzpicture}
                \caption{invasion model}
                \label{fig:worms}
        \end{subfigure}
        \caption{Directed acyclic graph diagram for (a) a general variable-dimensional invasion model, (b) the specific invasion model studied in our analysis;  arrows denote conditional dependence of parameters; labels next to each node denote the parameter spaces. }\label{fig:DAG}
\end{figure}

Exact inference for model parameters $\boldsymbol\theta$ is based on the posterior probability density,
\begin{eqnarray}
\pi(\boldsymbol\theta \mid Y) \; \propto \;  p(Y \mid \boldsymbol\theta) \; \pi(\boldsymbol\theta), 
\label{eqn:posterior}
\end{eqnarray}

\noindent where the first factor on the right hand side is the likelihood of the data given $\boldsymbol\theta$, and the second is its prior. The model parameters impact the data indirectly through the number and configuration of introduction locations, ${\bf x}_k$, which are not themselves of interest.  
Therefore the likelihood  in expression (\ref{eqn:posterior}) is obtained by integration with respect to the number and location of the latent introduction events:

\begin{equation}
p(Y \mid \boldsymbol\theta)  = \sum_{k=0}^\infty \int_{\mathcal{H}^k} p(Y \mid {\bf x}_k,\boldsymbol\theta) \;  p({\bf x}_k \mid k,\boldsymbol\theta) \; p(k \mid \boldsymbol\theta) \; \mbox{d}{\bf x}_k.\label{eq:int4Like}
\end{equation}

Estimation for this class of models poses analytical and computational challenges.  In practice, evaluating $p(Y \mid {\bf x}_k,\boldsymbol\theta)$ is often infeasible, especially under complex models for interaction and spread, or when the data is only partially observed in space.  Furthermore, the domain of integration in (\ref{eq:int4Like}) 
 changes dimension with the dimension of the latent variable, leading to further difficulties in obtaining a closed form likelihood.  
For the motivating invasive species model in Section \ref{sec:wormstudy}, even under simple deterministic linear spread in one dimension, evaluating the likelihood  (\ref{eq:int4Like}) presents a geometric and combinatorial challenge whose complexity renders exact inference practically limited.  As a result, approximate simulation-based inference, such as ABC, often becomes the only tool available to approach such problems.




\subsection{Approximate Bayesian inference for latent variable models}

Assume that the model can be forward-simulated to generate a synthetic dataset $D\in\mathcal{Y}$ from $p(D \mid {\bf x}_k,\boldsymbol\theta)$.  Under deterministic spread, $p(D \mid {\bf x}_k,\boldsymbol\theta)$ is a point mass density around a data simulation function.  

ABC replaces (\ref{eqn:posterior}) with an approximation of the $D$-augmented posterior density marginalized over the synthetic data, $D$, and latent variables, $(k,{\bf x}_k)$.  The approximation is based on the use of a low-dimensional summary statistic function $s(\cdot)$, as follows:
\begin{align*}
\pi(\boldsymbol\theta  \mid Y) & \propto  
\sum_{k}^{\infty} \int_{\mathcal{H}^{k}}
\pi(\boldsymbol\theta) 
\; p(k \mid \boldsymbol\theta) 
 \; p({\bf x}_k \mid k, \boldsymbol\theta)
\; p(Y \mid {\bf x}_k,\boldsymbol\theta)
\; \mbox{d} {\bf x}_k  \\
&\propto  
\sum_{k=0}^{\infty}\int_{\mathcal{Y}} \int_{\mathcal{H}^{k}}
\pi(\boldsymbol\theta)  
\; p(k \mid \boldsymbol\theta) 
\;  p({\bf x}_k \mid k, \boldsymbol\theta)
\; p(D \mid {\bf x}_k,\boldsymbol\theta)
\; p(Y \mid s(Y),{\bf x}_k,\boldsymbol\theta)
 \; p(s(Y) \mid D) 
\; \mbox{d} {\bf x}_k
\; \mbox{d}  D\\
& \approx 
\sum_{k=0}^{\infty} \int_{\mathcal{Y}}\int_{\mathcal{H}^{k}}
\pi(\boldsymbol\theta) 
\; p(k \mid \boldsymbol\theta) 
\;  p({\bf x}_k \mid k, \boldsymbol\theta)
\; p(D \mid {\bf x}_k,\boldsymbol\theta)
\; p(s(Y) \mid D) 
\; \mbox{d} {\bf x}_k
\;  \mbox{d}  D\\
& \equiv \pi_{ABC}(\boldsymbol\theta \mid Y).
\end{align*}

 The accuracy of the ABC approximation depends on the degree of sufficiency of the data summary $s(Y)$ for $(\boldsymbol\theta, k, {\bf x}_k)$, and on the discrepancy  term $p(s(Y) \mid D)$ relating the summarized synthetic data to the observed data by a kernel $K_\epsilon$ with bandwidth $\epsilon\geq 0$: 
\begin{align*}p(s(Y)\mid D) &= K_\epsilon\left[s(Y),s(D)\right].\end{align*}

\noindent    
When $s(Y)$ is sufficient for $(\boldsymbol\theta, k, {\bf x}_k)$ and $\epsilon = 0$, the ABC posterior density is exact.  In other words, when $p(Y \mid s(Y),{\bf x}_k, \boldsymbol\theta) =p(Y \mid s(Y)) $ and $p(s(Y) \mid D)$ is a  point mass function centered at $s(D)$, then $\pi(\boldsymbol\theta \mid Y)= \pi_{ABC}(\boldsymbol\theta \mid Y)$.  
However, low-dimensional sufficient statistics cannot in general be obtained when the likelihood is unknown, so the data summaries employed for ABC are chosen subjectively, leading to an approximate posterior density \citep{FearnheadPrangle2012}.   

The bandwidth $\epsilon$ controls the tolerance for the discrepancy between the summarized real and synthetic data. Effectively, $\epsilon$ controls the trade-off between the dimension of the summary statistic and Monte Carlo error from too few data matches (e.g., \cite{FearnheadPrangle2012}, lemma 1).

Under the data simulation model, the ABC posterior for the general invasion model (Figure \ref{fig:general}) is:
\begin{align}
 \pi_{ABC} \left(\boldsymbol\theta  \mid  Y \right)  & \; \propto
\pi\left(\boldsymbol\theta\right)
\;
\int_{\mathcal{Y}} 
\underbrace
{\sum_{k=0}^\infty
\int_{\mathcal{H}^k}   
p\left(D  \mid  {\bf x}_k, \boldsymbol\theta \right)
\; p \left({\bf x}_k \mid k, \boldsymbol\theta \right)
\; p\left(k  \mid  \boldsymbol\theta\right)
\; \mbox{d} {\bf x}_k}
_{p(D \mid \boldsymbol\theta)} 
 \; p\left( s(Y) \mid D\right)
\; \mbox{d}D.
\label{eqn:ABCposterior}
\end{align}

\noindent Estimates of the mean, mode, or quantiles of $\pi_{ABC} \left(\boldsymbol\theta  \mid  Y \right)$ can now be obtained from a Monte Carlo sample.  One simple method to produce such a sample is rejection ABC  \citep{SissonFan2010}, shown in Algorithm \ref{alg:rejectionABC}. More efficient ABC-MCMC sampling strategies rely on dependent proposals for parameters. However, our introduction model involves the parameter  ${\bf x}_k$ whose dimension changes with $k$, and therefore we develop a transdimensional sampling approach to avoid further approximation (such as the approximation of  \cite{congdon2006}).

\begin{algorithm}
\caption{rejection ABC}
\label{alg:rejectionABC}
\begin{spacing}{1.5}
\begin{algorithmic} 
\STATE at iteration $\ell = 0$ initialize $(\boldsymbol\theta,k,{\bf x}_{k}, D)^{(\ell)}$
\FOR{iteration $\ell = 1$ {\bf to} $L$}
\STATE propose $\boldsymbol\theta^\star \sim \pi(\boldsymbol\theta^\star)$ 
\STATE propose $k^\star \sim \pi(k)$
\STATE conditionally simulate ${\bf x}_{k^\star}^\star\sim p({\bf x}_{k^\star}^\star\mid k^\star, \boldsymbol\theta^\star)$
\STATE generate $D^\star$ from $p(D \mid {\bf x}_{k^\star}^\star, \boldsymbol\theta^\star)$ 
\STATE with probability $ K_\epsilon\left[s(Y), s(D^\star))\right]$ set $(\boldsymbol\theta,k,{\bf x}_{k}, D)^{(\ell)}\leftarrow(\boldsymbol\theta^\star,k^\star,{\bf x}_{k^\star}^\star,D^\star)$, \\ otherwise set $(\boldsymbol\theta,k,{\bf x}_{k}, D)^{(\ell)}\leftarrow(\boldsymbol\theta,k,{\bf x}_{k}, D)$
\ENDFOR
\end{algorithmic}
\end{spacing}
\end{algorithm}

\subsection{Transdimensional Approximate Bayesian Computation}
\label{sec:transABC}

Addressing the problem of low acceptance rates and unreliable estimates from the approximation in (\ref{eqn:ABCposterior}), we develop an  efficient algorithm for obtaining samples.  Instead of relying on conditional simulation on the variable-dimensional model subspace,  a transdimensional approach \citep{Green1995} is adopted allowing proposals between probability spaces of different dimensions.

In order to construct transitions between all model spaces of different dimension, it is sufficient to define  pair-wise transitions between all model index pairs $(i,j) \in\mathcal{K}\times\mathcal{K}$, with associated model-specific parameter vectors ${\bf x}_i \in \mathbb{R}^{n_i}$ and ${\bf x}_j\in\mathbb{R}^{n_j}$.  
Model-specific vectors are augmented by auxiliary random variables, ${\bf u}_i\in\mathbb{R}^{m_i}$ and ${\bf u}_j\in\mathbb{R}^{m_j}$ respectively, chosen by convenience under the constraint $n_i + m_i = n_j + m_j$, which ensures that both model spaces have the same dimension.

The diffeomorphic transformation $\phi_{ij}: \mathbb{R}^{n_i}  \times \mathbb{R}^{m_i} \to \mathbb{R}^{n_j} \times  \mathbb{R}^{m_j}$ is defined corresponding to the mapping $\left({\bf x}_i, {\bf u}_i\right) \to \left({\bf x}_j, {\bf u}_j\right)$.  For a given transformation $\phi_{ij}$ with Jacobian $|J_{ij}| = | \partial \phi_{ij}\left({\bf x}_i, {\bf u}_i\right)/\partial ({\bf x}_i, {\bf u}_i) |$, a proposed move from model $i$ to $j$ is constructed by first drawing a vector ${\bf u}_i$ from the density $p_{i}({\bf u}_i)$, and then obtaining $({\bf x}_j,{\bf u}_j) = \phi_{ij}({\bf x}_i, {\bf u}_i)$.  The inverse mapping $\left({\bf x}_j, {\bf u}_j\right) \to \left({\bf x}_i, {\bf u}_i\right)$ can be accomplished by using the transformation $\phi_{ji}$.  For simplicity, one generally sets $m_j=0$ for all $(i,j)$ such that $n_j\geq n_i$.   

The resulting transdimensional Metropolis-Hastings random walk ABC sampler is described in Algorithm \ref{alg:rjmcmcABC} and produces a Markov chain whose stationary distribution has the desired ABC posterior density (\ref{eqn:ABCposterior}; proof is provided in the Appendix) .

\begin{algorithm}
\caption{transdimensional ABC}
\label{alg:rjmcmcABC}
\begin{spacing}{1.5}
\begin{algorithmic} 
\STATE at iteration $\ell=0$ initialize $(\boldsymbol\theta, k,{\bf x}_k, D)^{(\ell)}$ from an area of positive posterior probability (e.g. via rejection ABC)
\FOR{ $\ell = 1$ {\bf to} $L$}
\STATE propose model parameter $\boldsymbol\theta^\star \sim q(\boldsymbol\theta^\star \mid \boldsymbol\theta)$ and move type from $k$ to $k^\star \sim q(k^\star \mid k)$
\STATE sample ${\bf u}_{k} \sim p_{kk^\star}({\bf u}_{k})$ 
\STATE obtain $({\bf x}_{k^\star}^\star, {\bf u}_{k^\star}^\star) = \phi_{kk^\star}({\bf x}_k, {\bf u}_k)$
\STATE generate synthetic data $D^\star$ from  $p(D^\star  \mid  {\bf x}^\star_{k^\star},\boldsymbol\theta^\star )$
\STATE calculate:
\[
\alpha =
\frac{ p\left( s(Y) \mid D^\star\right)}
{ p\left( s(Y) \mid D\right)}
\times
\frac{ p({\bf x}^\star_{k^\star} \mid k^\star,\boldsymbol\theta^\star)
p(k^\star \mid \boldsymbol\theta^\star)\pi(\boldsymbol\theta^\star) }
{p({\bf x}_k \mid k,\boldsymbol\theta)
p(k \mid \boldsymbol\theta) 
\pi(\boldsymbol\theta)} 
\times
\frac{q(\boldsymbol\theta^\star \mid \boldsymbol\theta)  
q(k^\star \mid k)
p_{k^\star k}({\bf u}^\star_{k^\star})}
{q(\boldsymbol\theta \mid \boldsymbol\theta^\star)  
q(k \mid k^\star)
p_{kk^\star}({\bf u}_k)
} 
\;|J_{kk^\star}|
\]
\STATE  with probability $\min\{1,\alpha\}$ set $(\boldsymbol\theta, k,{\bf x}_k, D)^{(\ell)} \leftarrow(\boldsymbol\theta^\star,k^\star,{\bf x}_{k^\star}^\star,D^\star)$, \\ otherwise return $(\boldsymbol\theta, k,{\bf x}_k, D)^{(\ell)}\leftarrow(\boldsymbol\theta,k,{\bf x}_k,D)^{(\ell-1)}$ 
\ENDFOR
\end{algorithmic}
\end{spacing}
\end{algorithm}



\section{Motivating application}\label{sec:wormstudy}

\subsection{Field methods}

\begin{figure}
        \centering
        \includegraphics[trim= 0cm 4cm 0cm 0cm, width = 4in]{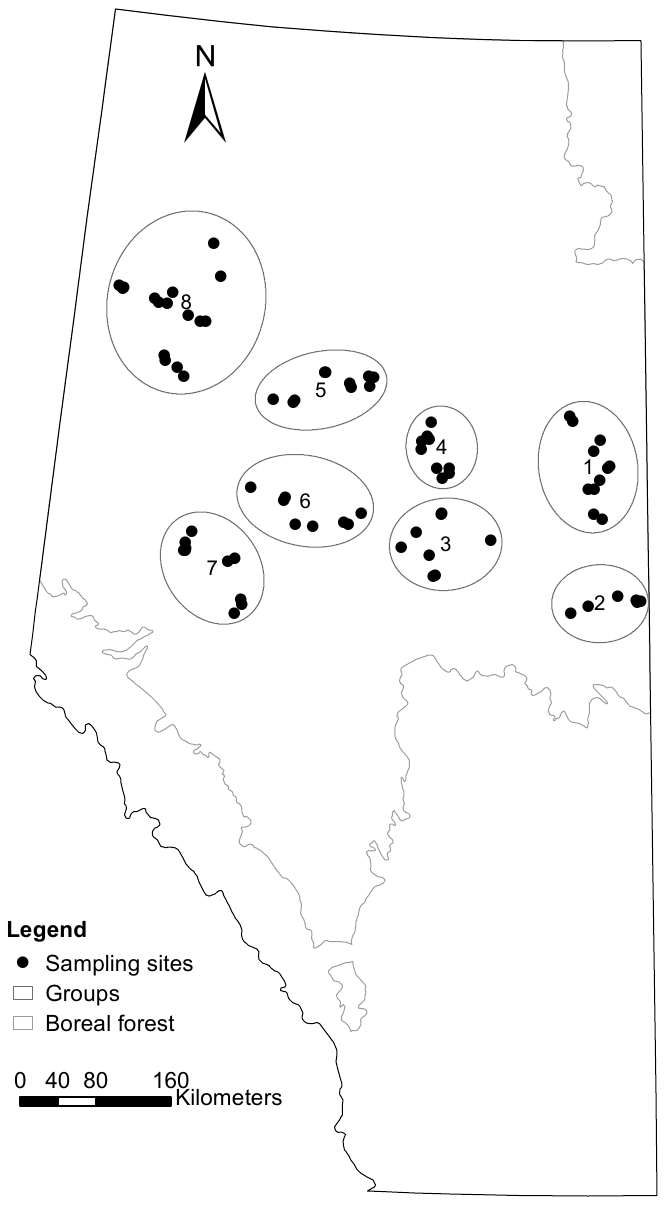}
        \caption{Survey locations (n=78), represented by black circles, within the boreal forest of northern Alberta. Sites were clustered into 8 groups according to spatial location, represented by ovals. }\label{fig:map}
\end{figure}

Data on earthworm occurrence were collected in the boreal forest of northern Alberta, Canada between 54.4$^\circ$N and 57.8$^\circ$N latitude and 110.1$^\circ$W and 119.8$^\circ$W longitude (Figure \ref{fig:map}; see \cite{CameronBayne2009} for further description).  We focused on the litter-dwelling earthworm \textit{Dendrobaena octaedra} (Savigny), which is introduced from Europe and is the most common earthworm species in northern Alberta \citep{CameronEtAl2007,CameronBayne2009}.  This species is not commonly used as bait and therefore transport by vehicles is the key mechanism involved in passive dispersal \citep{TiunovEtAl2006}.  Earthworms were sampled at roads ranging in age from 6 to 56 years old in 2006.  At each road, sampling occurred along a 50.25 m transect which ran parallel to the road.  Leaf litter was hand-sorted to determine earthworm occurrence in six 0.0625 m$^2$ (25 by 25 cm) quadrats with a 9.75m gap between quadrats on each transect.   Transects were 1-2 m into the forest from its edge, with alternate quadrats located $\sim$5 m farther into the forest interior.  We will hereafter refer to this as the spatial data $Y$.

To examine the spread rate within our study area, additional temporal data on earthworm occurrence were collected along transects perpendicular to roads where earthworms were already established (n = 26).  These transects were 500 m long, with sampling quadrats every 50 m.  Sampling occurred in 2006 \citep{CameronBayne2009} and again in 2012/13.  We subtracted the distance of the farthest quadrat with earthworms present in 2006 from the distance of the farthest occupied quadrat in 2012/13 and divided by 6 years.
We will refer to this as the temporal data $Z$.

\subsection{Model inference for earthworm invasions}
\label{sec:earthwormmodel}

We now propose a model for the invasion of earthworms along one spatial dimension over time.  Our framework is applicable for any model with a latent variable structure, requiring integration with respect to parameters of variable dimension, as long as data can be obtained by forward-simulation from the model.

The 78 observation transects were grouped into eight categories based on their spatial location on the landscape (Figure \ref{fig:map}).  Each of these groups contained young, intermediate-aged, and old roads. Observations $Y_{gr}\in\{0,1\}^6$  for road $r=1,\ldots,n_g$ in group $g = 1,\ldots,8$ consist of binary error-free measurements of presence or absence at the sampled 6 quadrats of the corresponding transect.   For notational clarity, we omit dependence of the introduction locations on $k$ and instead define ${\bf x}_{gr}$ to be the value of ${\bf x}_k $ associated with group $g$ and road $r$.  

For every road in the study, we model earthworm introductions and spread over the spatio-temporal horizon $\mathcal{H}_{gr}$, defined by the largest possible extent of activity that can affect the observed data under our model, as shown in Figure \ref{fig:simhorizon}.
We assume that the introduction rate is constant within each of the 8 selected road groups per unit area of $\mathcal{H}_{gr}$ by modelling the $k$ introductions according to a homogeneous space-time Poisson process with rate $\lambda_g  \times \mbox{Area of }\mathcal{H}_{gr} \in\mathbb{R}^+$, measured in introductions/(m$\times$yr) to account for different road ages.   

\begin{figure}
        \centering
        \includegraphics[trim = 2cm 0cm 0cm 2cm, width = 5in]{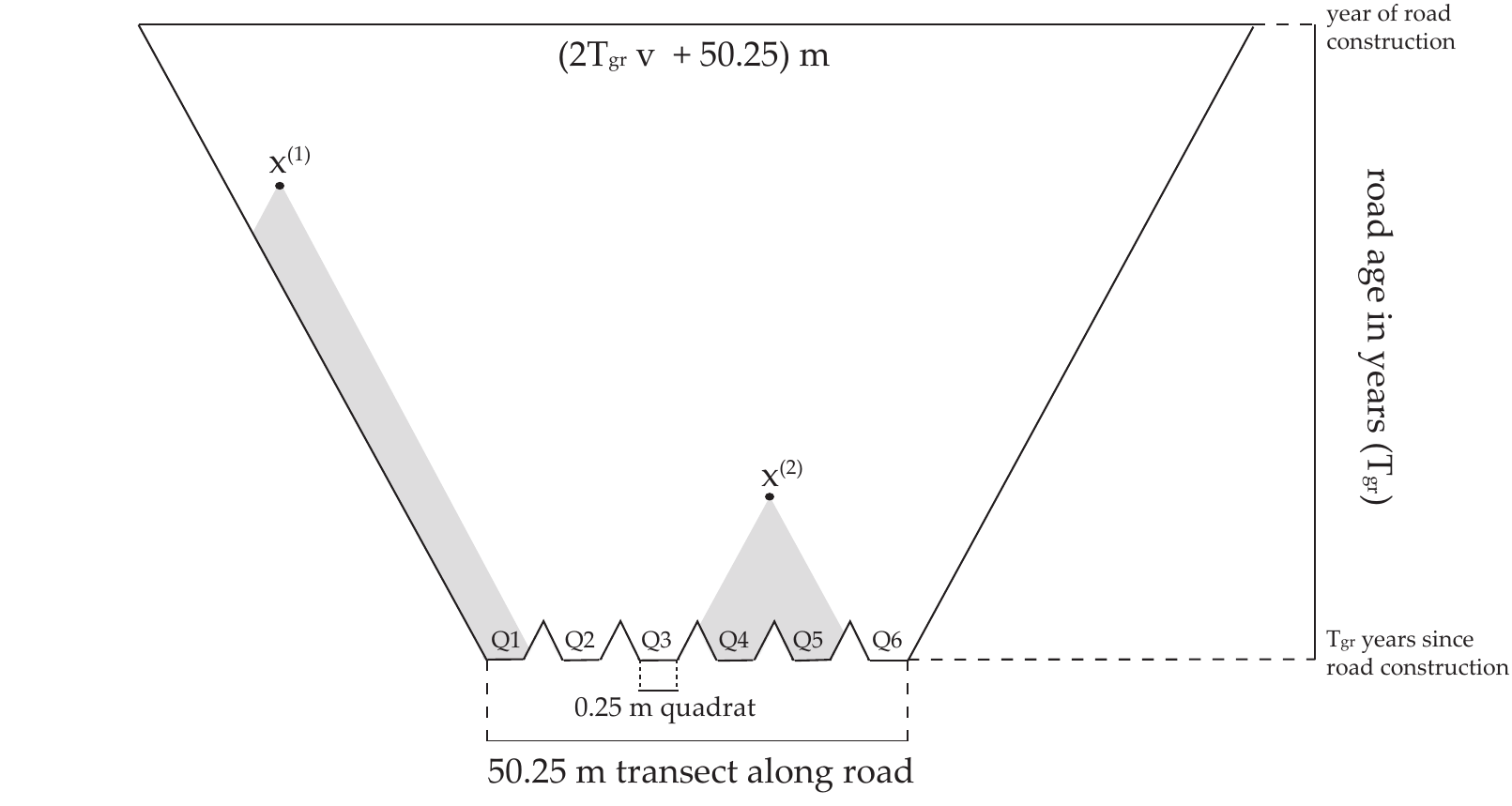}
        \caption{Simulation horizon, $\mathcal{H}_{gr}$, consisting of all spatio-temporal locations where an introduction event will affect the data according to our spread model. Two introductions and subsequent spread are shown shaded in grey.  The data generated by the illustrated invasion is $m\left({\bf x}_{gr}, \nu, T_{gr} \right) = [1, 0, 0, 1, 1, 0].$}\label{fig:simhorizon}
\end{figure}

For a $T_{gr}$-year-old road, earthworms are assumed to spread linearly in time from point-source locations, ${\bf x}_{gr}$, at an unknown constant rate,  $\nu\in\mathbb{R}^+$ (measured in m/yr), defining the geometry of the horizon $\mathcal{H}_{gr} = H(\nu,T_{gr}) \in \mathbb{R}^2$. If earthworms spread far enough from ${\bf x}_{gr}$ to overlap with a sampling quadrat, a presence indicator is recorded as shown in Figure \ref{fig:simhorizon}.  We denote the resulting presence or absence of worms at the measurement locations by $m({\bf x}_{gr},\nu,T_{gr})$. Since we assume that presence or absence of worms in a quadrat is measured without error, the data generating mechanism $D=m({\bf x}_{gr},\nu,T_{gr})$ is conditionally deterministic. For expositional clarity we use the point mass density, $ \delta_{m\left({\bf x}_{gr}, \nu, T_{gr} \right)}$, to highlight our use of a conditionally deterministic data generating mechanism.
The resulting hierarchical model, illustrated in Figure \ref{fig:worms}, is:

\begin{equation}
\begin{array}{r@{}l}
Y_{gr}  \mid  {\bf x}_{gr}, \nu & \; \sim \; \delta_{m\left({\bf x}_{gr}, \nu, T_{gr} \right)},\\
{\bf x}_{gr}  \mid  k_{gr}, \nu & \; \sim \; \mbox{Uniform}\big(H(\nu,T_{gr})\big), \\
 k_{gr}  \mid \lambda_g, \nu  & \; \sim \; \mbox{Poisson}\big(\lambda_g  \times \mbox{Area of }H(\nu,T_{gr}) \big),\\
Z \mid \nu& \; \sim \; \mbox{Normal}(\nu, \sigma^2),\\
\nu & \; \sim \; \mbox{Normal}(10,10^2),\\
\lambda_g & \; \sim \; \mbox{Exponential}(1),\\
\sigma^2 & \; \propto \; 1/\sigma^2.
\end{array}
\label{eqn:hierarchy}
\end{equation}

\noindent  
The prior on $\nu$ is based on related studies of earthworm spread rates \citep{Hale2008} but has an inflated variance to account for differences in soil characteristics and climate between studies, while
the prior on $\lambda_g$ reflects the expectation that introduction events should be relatively rare in the boreal forest.

Under this simple scenario, the space-time horizon $\mathcal{H}_{gr}$ can be subdivided into distinct regions whose shape depends on $\nu$ and road age, and the effect on the data of an introduction event occurring within each region can be computed exactly.  Therefore, we can obtain $p(Y_{gr}  \mid  {\bf x}_{gr}, \nu)$ in closed form using a geometric and combinatorial argument, and subsequently compute the likelihood $p(Y_{gr}  \mid  \lambda, \nu)$ through integration over these distinct regions, as shown in the Supplementary Materials.  In practice, such lengthy likelihood calculations are not feasible for general practitioners and their complexity introduces the potential for errors in implementation.  However, forward simulation is often fast and computationally accessible for practitioners.  As the difficulties with exact implementation increase with all but the simplest models, we illustrate transdimensional ABC as a practical alternative for this problem.
 
Our inference about introduction and spread is based on the marginal ABC posterior density:
\begin{align}
& \pi_{ABC} \left(\lambda_1,\ldots,\lambda_8,\nu, \sigma^2  \mid  Y_{1,1},\ldots Y_{8,n_g} , Z\right)  
 \; \propto \;
p\left(\sigma^2 \mid \nu, Z\right)\,
p\left(\nu \mid Z\right)\;
\label{eqn:wormsABCposterior}\\
& \quad
\prod_{g=1}^{8}\;\prod_{r=1}^{n_g}\;
\left[\pi\left(\lambda_g\right)\,
p\left( s(Y_{gr}) \mid D_{gr}\right) 
  \sum_{k_{gr}=0}^\infty
\int_{H(\nu,T_{gr})}   
 p \left({\bf x}_{gr} \mid k_{gr}, \nu \right)
 p\left(k_{gr}  \mid  \lambda_g, \nu\right)
\mbox{d}{\bf x}_{gr}\right]. \nonumber
\end{align}

\noindent
We use the transdimensional ABC sampler described in Algorithm \ref{alg:birthdeathABC}, which is a variant of Algorithm  \ref{alg:rjmcmcABC}, implemented with a birth-death proposal for the number of introductions for each road.  A birth-death proposal consists of either increasing by one (birth), decreasing by one (death), or maintaining unchanged the number of introduction events with some probability. This sampling algorithm mimics the data-generating process to produce synthetic data efficiently, which leads to the fast mixing desired in an ABC algorithm.
The birth-death proposal is a special case of the general transdimensional proposal with Jacobian of the transformation given by $|J|=1$, because $\phi_{ij}$ is taken to be the identity function of the auxiliary variables, which are sampled uniformly on $\mathcal{H}_{gr}$.

Due to the efficiency of this algorithm, we were able to use an error tolerance of $\epsilon=0$, with a point mass kernel distance metric to define the term,
\[ p\left( s(Y_{gr}) \mid D_{gr}\right) = \mathbb{I}\left\{s(Y_{gr}) = s(D_{gr})\right\}.\] 

\noindent This corresponds to accepting proposed parameters only when the resulting summarized simulated data matches the observed simulated data exactly.
For the summary statistic, $s$, we chose a 2-dimensional vector consisting of: (i) the number of consecutive occupied quadrats (strings of 1s in $Y_{gr}$), and (ii) the total number of occupied quadrats:
\begin{align*}
s( Y_{gr} ) =\left[ \; \mbox{number of strings of consecutive 1s} \; , \;  \Sigma_{i=1}^6 Y_{gr}^{(i)} \; \right]^\top.
\end{align*}

\noindent
Each consecutive sequence of occupied quadrats (strings of 1s) indicates that at least one introduction must have occurred. The length of each consecutive string of occupied quadrats can help distinguish between a recent introduction (e.g. 1 or 2 consecutive occurrences) and one or more old introductions (e.g. 6 consecutive occurrences).

\begin{algorithm}
\caption{birth-death ABC}
\label{alg:birthdeathABC}
\begin{spacing}{1.5}
\begin{algorithmic} 
\STATE at step $\ell=0$, initialize $(\lambda,\nu, k, {\bf x})^{(\ell)}$ from an area of positive posterior probability and generate $D$ from $m({\bf x},\nu,T)$
\FOR{$\ell = 1$ {\bf to} $L$}
\STATE sample $\sigma^{2(\ell)} \sim p(\sigma^2 \mid \nu^{(\ell-1)} \lambda^{(\ell-1)}, {\bf x}^{(\ell-1)}, k^{(\ell-1)},Y,Z)$ 
\STATE sample $\nu^{(\ell)} \sim p(\nu  \mid  \lambda^{(\ell-1)}, {\bf x}^{(\ell-1)}, k^{(\ell-1)},\sigma^{2(\ell)},Y,Z)$ 
\FOR{$g = 1$ {\bf to} $8$}
\STATE sample $\lambda^{(\ell)}_g \sim p(\lambda \mid \nu^{(\ell)}, {\bf x}_{g}^{(\ell-1)}, k^{(\ell-1)},\sigma^{2(\ell)},Y,Z)$ 
\FOR{$r = 1$ {\bf to} $n_g$}
\STATE sample ${\bf x}^{(\ell)}_{gr}$ from $p({\bf x}_{gr} \mid \nu^{(\ell)}, \lambda^{(\ell)}_g, k^{(\ell-1)}_{gr},\sigma^{2(\ell)},Y_{gr},Z)$ 
\STATE generate synthetic data $D_{gr}$ from $m({\bf x}_{gr},\nu^{(\ell)},T_{gr})$ 
\STATE propose $b \sim \mbox{Discrete Uniform}\{-1,0,1\}$ and set $k^\star_{gr}=k^{(\ell-1)}_{gr}+b$
\IF{$b = 1$}
\STATE sample $\mbox{x}\sim U\left[H(\nu^{(\ell)},T_{gr})\right]$
\STATE ${\bf x}^\star_{gr} = [{\bf x}^{(\ell)}_{gr}; \mbox{x}]$
\ELSIF{$b = -1$}
\STATE generate $d \sim \mbox{Uniform}\{1,\ldots,k_{gr}^{(\ell-1)}\}$   
\STATE ${\bf x}^\star_{gr} = ({\bf x}^{(\ell)}_{gr})_{-d} $
\ELSE
\STATE ${\bf x}^\star_{gr} = {\bf x}^{(\ell)}_{gr}$
\ENDIF 
\STATE generate synthetic data $D^\star_{gr}$ from $m({\bf x}^\star_{gr},\nu^{(\ell)},T_{gr} )$ and calculate:
\[
\alpha =
\frac{ p\left( s(Y) \mid D^\star\right)}
{ p\left( s(Y) \mid D\right)}
\times
\frac{ p({\bf x}^\star_{gr}  \mid k_{gr}^\star)p(k_{gr}^\star) }{p({\bf x}^{(\ell)}_{gr} \mid k^{(\ell)}_{gr}) p(k^{(\ell)}_{gr})}\times\frac{q(k_{gr}^\star\mid k_{gr}^{(\ell-1)})}{  q(k_{gr}^{(\ell-1)}\mid k_{gr}^\star)} 
\]
\IF{ $\min\{1,\alpha\} \geq \mbox{Uniform}(0,1)$}
\STATE set $(k_{gr}^{(\ell)},{\bf x}^{(\ell)}_{gr},D^{(\ell)}_{gr}) \leftarrow (k_{gr}^\star,{\bf x}^\star_{gr},D^\star_{gr})$ 
\ELSE 
\STATE set $(k_{gr}^{(\ell)},{\bf x}^{(\ell)}_{gr},D^{(\ell)}_{gr}) \leftarrow (k_{gr}^{(\ell-1)},{\bf x}^{(\ell-1)}_{gr},D^{(\ell-1)}_{gr})$ 
\ENDIF
\ENDFOR
\ENDFOR
\ENDFOR
\end{algorithmic}
\end{spacing}
\end{algorithm}


\section{Inference Results and Simulation Study}\label{sec:results}

\subsection{Inference on introduction and spread rates}

For this analysis we performed 250 000 iterations of Algorithm \ref{alg:birthdeathABC} and discarded the first 25 000 iterations as burn-in.   A Geweke convergence diagnostic was performed comparing the first 50 000 iterations (post burn-in) to the last \\50 000 iterations for each of $\lambda_g, g = 1,\ldots,8$ and $\nu$ assuming a 4\% autocorrelation.  The test statistic returned p-values in the range of 0.17 to 0.65 for each parameter.
Bivariate posterior density heat maps for the introduction and spread rates, $\nu$ and $\lambda_g$, are shown in Figure \ref{fig:bivariate_posterior}.  The marginal posterior densities of $\lambda_g, g=1,\ldots,8$ can be classified into three categories as shown in the top panel of Figure \ref{fig:posterior_introductions_all}.  Groups 4 and 2 have the lowest posterior median rates of introduction ($1.31\times 10^{-5}$ and $2.69\times 10^{-5}$ introductions/(m$\times$yr) respectively).  Groups 1, 3, 5, and 8 have posterior median introduction rates in the mid-range ($4.03\times 10^{-5}$, $4.12\times 10^{-5}$, $3.50\times 10^{-5}$, and $4.09\times 10^{-5}$ introductions/(m$\times$yr) respectively).  Finally, groups 6 and 7 have the highest posterior median introduction rates ($6.53\times 10^{-5}$, $7.38\times 10^{-5}$ introductions/(m$\times$yr) respectively).  The marginal posterior of $\nu$ is shown in the middle panel of Figure \ref{fig:posterior_introductions_all}, where the posterior median was estimated at 13.93 m/yr with a 95\% credible interval between 10.60 and 16.99 m/yr.   The bottom panel of Figure \ref{fig:posterior_introductions_all} shows the marginal density of spatio-temporal introduction locations, $\pi_{ABC}\left({\bf x}_{1,11} \mid  Y_{1,1},\ldots Y_{8,n_g} , Z\right)$, from a 26 year old road in group 1. The associated presence or absence data, $Y_{1,11}=[1,1,0,0,1,0]$, indicates an absence between two consecutive ``split'' observed presences. The split between subsequent presence measurements induces the distinct marginal density of introductions over the spatio-temporal horizon. This posterior non-uniformity suggests that employing any ABC sampling scheme that generates all introduction events, ${\bf x}_{gr}$, from the uniform prior will be extremely inefficient, and supports the use of a transdimensional approach.

\begin{figure}
 \centerline{\includegraphics[trim = 1cm 0cm 1cm 2cm,width=5in]{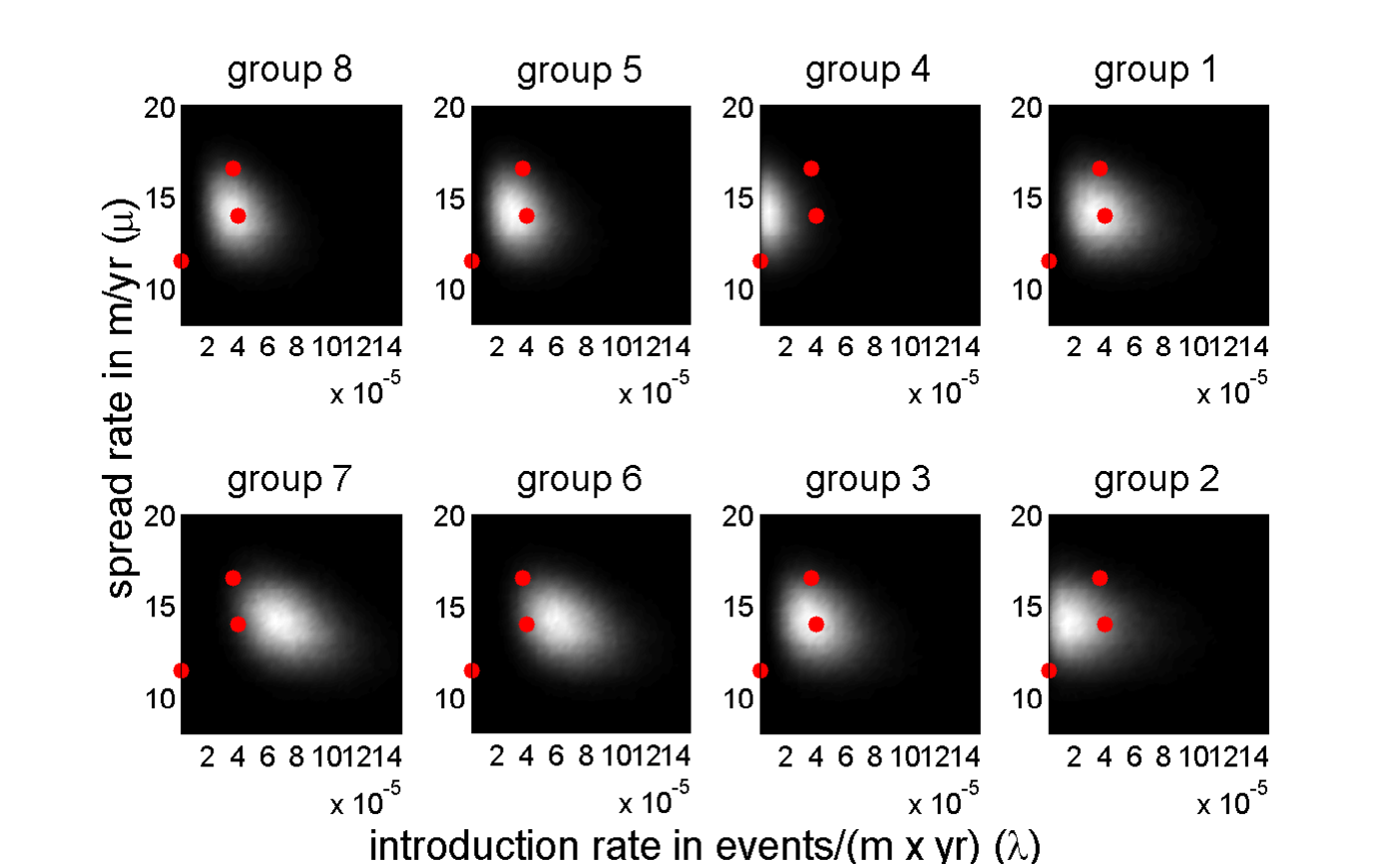}}
\caption{Marginal posterior densities of spread rate (vertical axis) and introduction rate (horizontal axis) by group of roads.  Red dots indicate the values of ($\lambda,\nu$) used for the predictive simulation described in Section \ref{sec:map}.}
\label{fig:bivariate_posterior}
\end{figure}

\begin{figure}
\centering
 \includegraphics[trim = 0cm 0.75cm 0cm 0cm,width = 2.7in]{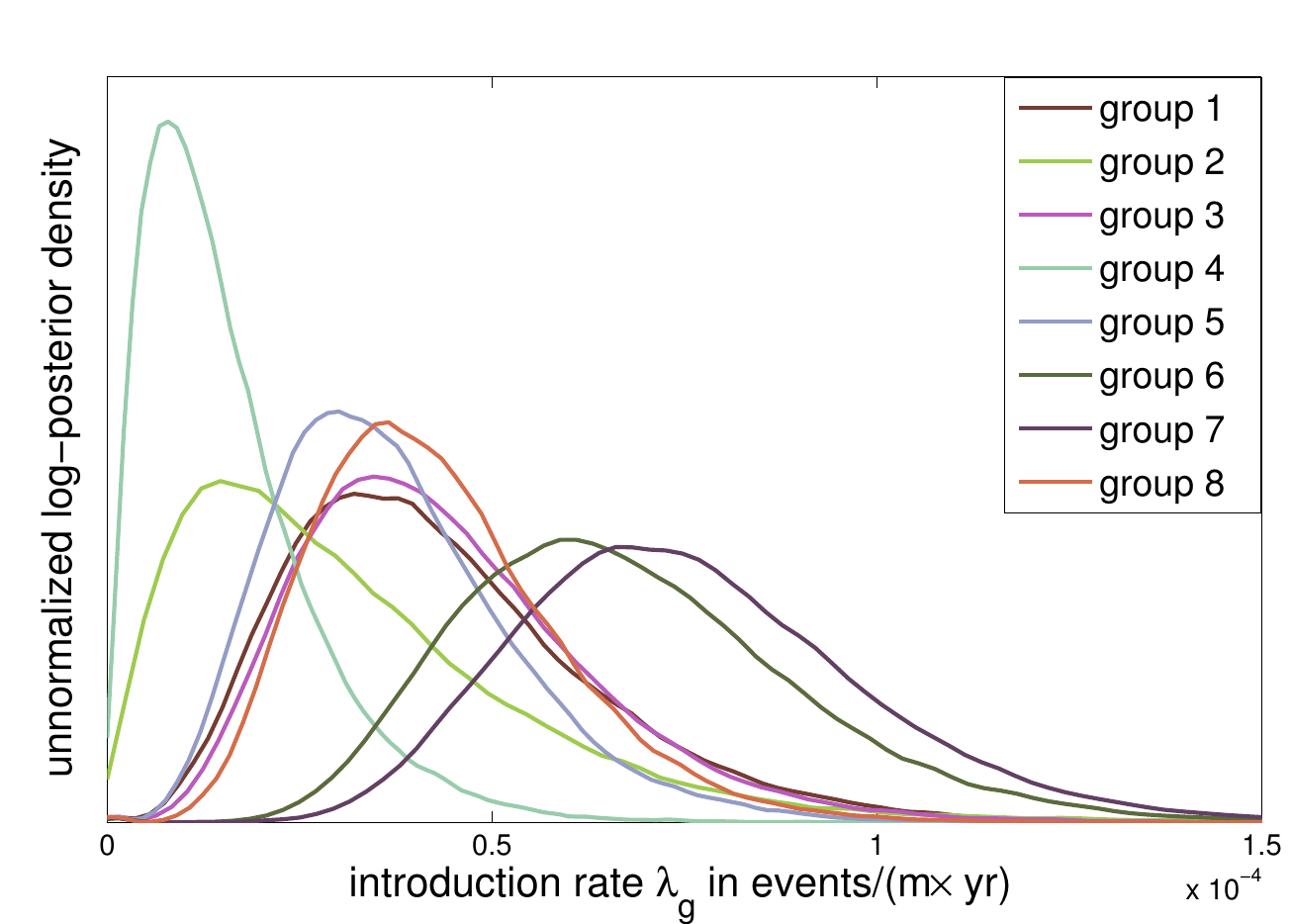}\\
\vspace{0.25cm}
\includegraphics[width = 2.7in]{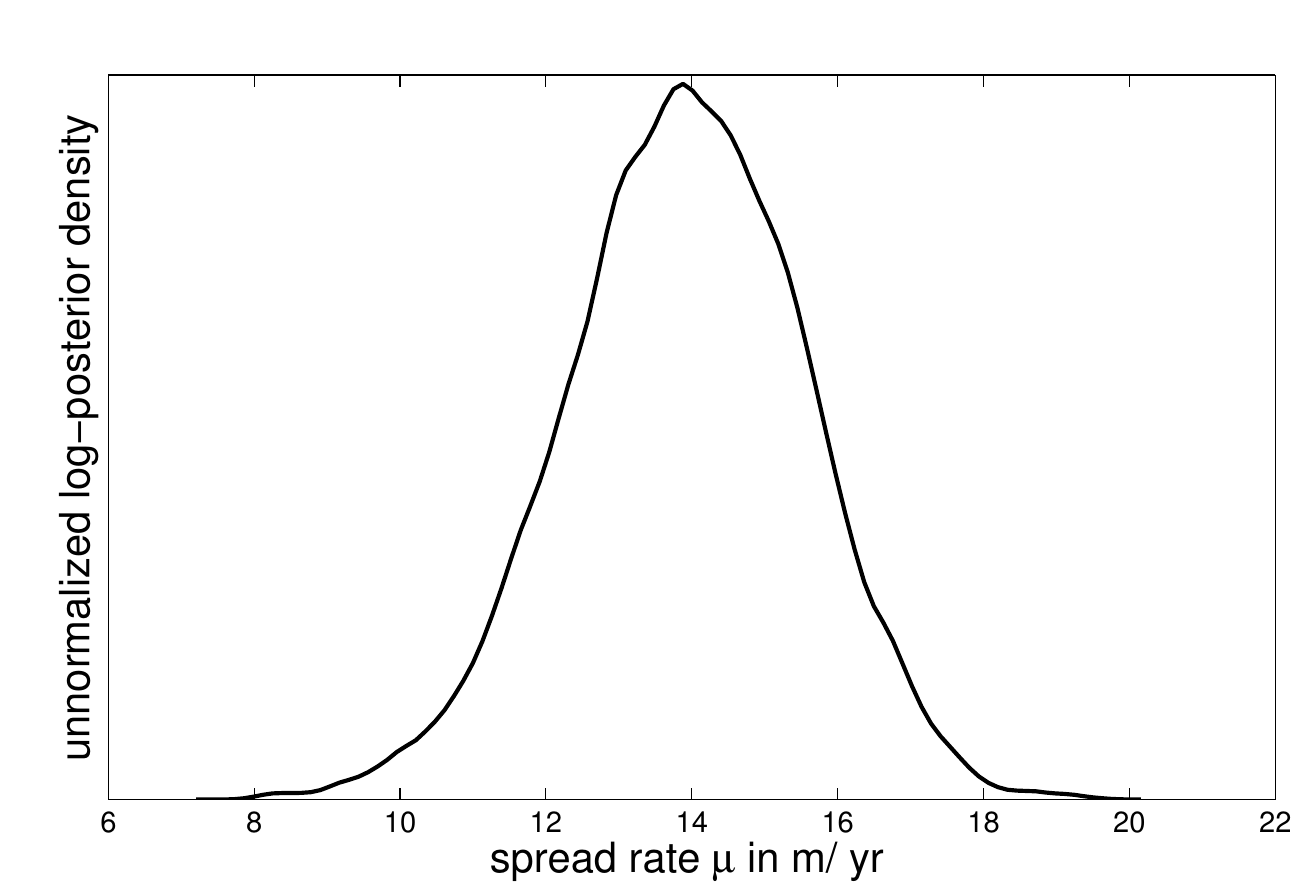}\\       
\vspace{0.25cm}
\includegraphics[width = 2.75in]{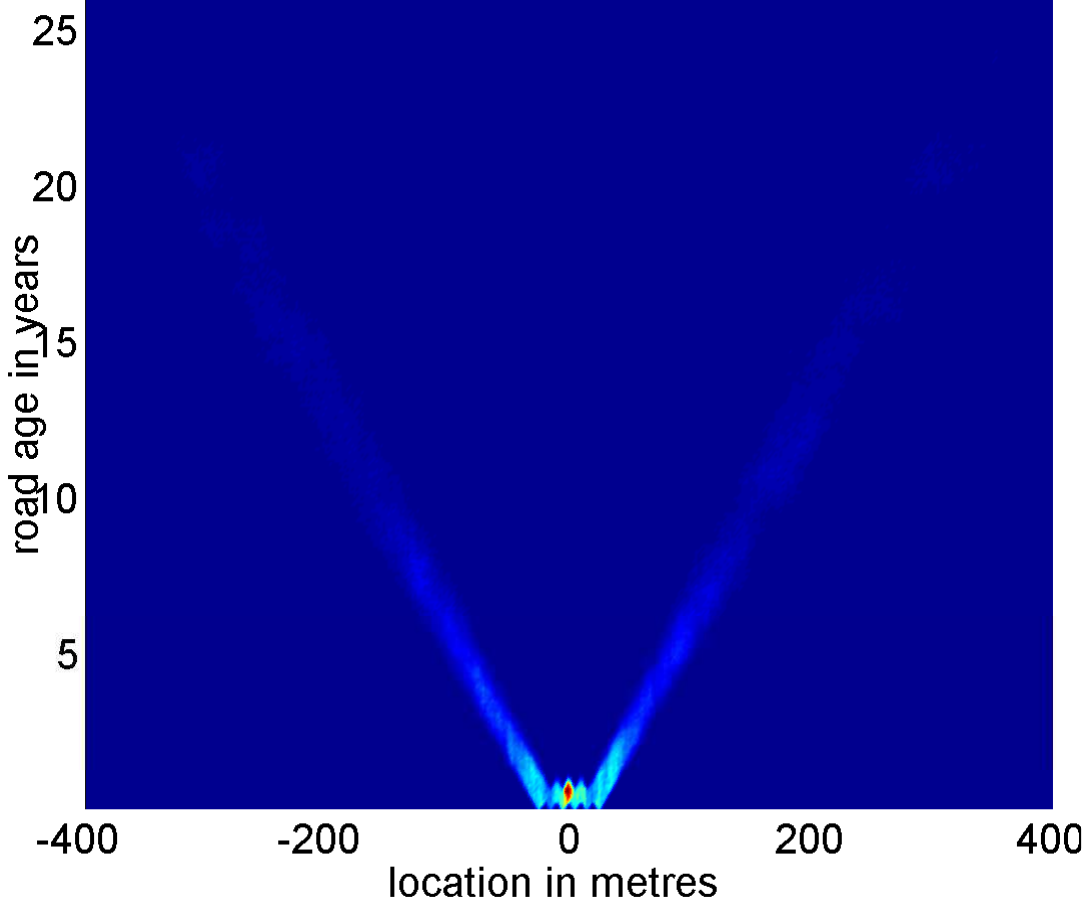}
\caption{Marginal posterior densities for the rate of introduction, $\lambda_g$ (top).  Marginal posterior density for the worm spread rate, $\nu$, (middle). Heatmap showing the marginal posterior of spatio-temporal introductions for the 26 year old road in group 1 (bottom).}
\label{fig:posterior_introductions_all}
\end{figure}

\subsection{Simulation Study}

A simulation study was conducted to examine the degree of posterior approximation resulting from summarizing the data by the statistic $s$, as well as the appropriateness of the model for describing the invasion process.  Although these effects are difficult to distinguish from a simulation study, we are able to attribute low coverage for two of the road groups studied to model specification.

We considered several options for simulation study designs. Generating data according to Model  (\ref{eqn:hierarchy}) given model parameters drawn from the full ABC posterior,  
$\pi_{ABC}\left( \lambda_1, \ldots, \lambda_8, \nu, k_{1,1},\ldots,k_{8,n_g} ,{\bf x}_{1,1},\ldots,{\bf x}_{8,n_g} \mid Y_{1,1},\ldots,Y_{8,n_g} , Z \right)$, would yield data summaries which exactly match the real observed data summaries. Because this would amount to repeating the real data analysis multiple times, we deemed this strategy to be unsuitable.  An alternative strategy consists of sampling introductions and spread rates from the marginal ABC posterior (\ref{eqn:wormsABCposterior}) and then independently sampling the number and locations of introductions from $\pi_{ABC}( k_{1,1},\ldots,k_{8,n_g} ,{\bf x}_{1,1},\ldots,{\bf x}_{8,n_g} \mid Y_{1,1},\ldots,Y_{8,n_g} , Z )$.  This strategy was also found unsuitable because of the downward bias on introduction rates resulting from the possibility of generating introduction events, ${\bf x}_{gr}$,  outside of the horizon, $H(\nu,T_{gr})$.   
In order to avoid the above issues we generated data from Model (\ref{eqn:hierarchy}) using fixed $\lambda_g, g=1\ldots,8$ and $\nu$ chosen to be the posterior median values of  group specific introduction rates and overall spread rate, as reported in the previous section. 

The data generating process and inference were repeated 150 times using the transdimensional birth-death ABC method described in Algorithm \ref{alg:birthdeathABC} implemented with 250 000 iterations, and 95\% credible intervals were computed.    A Geweke diagnostic was performed component-wise for $\lambda_g, g=1,\ldots,8$ and $\nu$, comparing the final 100 000 iterations with the initial 100 000 iterations after discarding 25 000 iterations for burn-in.  In all cases convergence was attained.  Table  \ref{tab:coverage} shows the median of posterior medians, median of posterior standard deviations (however the skew in the posterior densities implies caution should be used when interpreting this value), and the proportion of times that the 95\% credible intervals covered the introduction rate used to generate the data.  The 95\% credible regions for the simulated datasets covered the true spread rate 100\% of the time, and covered the true group specific introduction rates with reasonable frequency for all road groups, except for road groups 6 and 7.  

The large values of the spread rates $\lambda_6$ and $\lambda_7$ that were chosen for the simulation study make multiple introduction events likely, especially for older roads.   However, uniform sampling from $p({\bf x}_{gr}\mid k_{gr},\nu)$ over $H(\nu,T_{gr})$ given $k_{gr}\geq1$ introductions, generates presence in all quadrats with high probability. This is in contrast to the sets of split consecutive presence measurements separated by absences which were observed in the field data.  
Whereas conditioning on multiple consecutive presence measurements in the field data leads to high estimated introduction rates and marginal posterior densities of introduction events such as that illustrated in Figure \ref{fig:posterior_introductions_all}, a systematic lack of such multiple consecutive presence measurements in the simulated data leads to an underestimate of the number and rate of introductions.  From this we conclude that for groups 6 and 7, the model assumption that introductions occur according to a homogeneous space-time Poisson process may not be reasonable. A more likely scenario for these two groups might consist of increased frequency of introductions in recent years leading to multiple consecutive separated strings of presences.  Another possible culprit may be the assumption of a constant spread rate, which in reality may be hindered by obstacles to spread along roads. These possibilities present interesting subjects for further exploration and study of introduction and spread patterns for these individual regions.


\begin{table}
        \centering
	\begin{tabular}{|l|c|c|c|c|c|c|c|c|c|}
\hline	&$\nu$&$\lambda_1$&$\lambda_2$&$\lambda_3$&$\lambda_4$&$\lambda_5$&$\lambda_6$&$\lambda_7$&$\lambda_8$\\\hline
  &&&&&&&&&\\
	True value &14.216 &$4.03$ &    $2.69  $ &   $4.12 $  &  $ 1.31  $&     $3.5  $ &   $6.53 $ &    $7.38   $& $  4.09 $\\ 
		& &$\times10^{-5}$ &    $\times10^{-5}  $ &   $\times10^{-5} $  &  $ \times10^{-5}  $&     $\times10^{-5}  $ &   $\times10^{-5} $ &    $\times10^{-5}   $& $  \times10^{-5} $\\ \hline
	Median of &&&&&&&&&\\
	posterior &15.53&$3.14 $&  $ 2.34 $&   $3.05 $&  $  1.87 $&  $ 2.50 $&$    3.20 $& $  3.90 $&   $2.45 $\\
	 medians
	 		& &$\times10^{-5}$ &    $\times10^{-5}  $ &   $\times10^{-5} $  &  $ \times10^{-5}  $&     $\times10^{-5}  $ &   $\times10^{-5} $ &    $\times10^{-5}   $& $  \times10^{-5} $\\ \hline
    Median of  &&&&&&&&&\\
	posterior   & 1.23& 1.73  &     2.00&       1.52&       1.23&       1.24&       1.42&       1.51&       1.15  \\
	 standard & &$\times10^{-5}$ &    $\times10^{-5}  $ &   $\times10^{-5} $  &  $ \times10^{-5}  $&     $\times10^{-5}  $ &   $\times10^{-5} $ &    $\times10^{-5}   $& $  \times10^{-5} $\\
	  deviations  &&&&&&&&&
	 		\\ \hline
Percent   &&&&&&&&&\\
of Highest   &&&&&&&&&\\
Posterior   &100&100&100&90&100&89.33&60&60.67&89.33\\
Intervals   &&&&&&&&&\\ containing   &&&&&&&&&\\
 True value &&&&&&&&&\\ \hline
	\end{tabular}
\caption{Simulation study results including the median of posterior medians and coverage probability of the 95\% highest posterior credible intervals.}
\label{tab:coverage}
\end{table}

\section{Discussion and Prediction}\label{sec:map}

The highest estimated introduction rates were for groups 6 and 7, which are located in the southwestern part of Alberta's boreal forest, near the Peace River Region. In this area, human settlement and agricultural conversion occurred earlier and more extensively than in the rest of northern Alberta (Schneider 2002). The higher earthworm introduction rates may thus be related to the greater intensity and longer history of human activity. However, it is not clear why groups 2 and 4 had substantially lower introduction rates than other locations, as levels of anthropogenic disturbance are relatively similar to levels for groups 1, 3, 5, and 8. More intensive sampling at sites along a gradient of human activity would be needed to examine effects of anthropogenic disturbances on spatial variability of earthworm introduction rates.  
In all cases, the introduction rates estimated using our approach were lower than previously estimated from the same dataset \cite{CameronBayne2009}.  The analysis in \cite{CameronBayne2009} yielded an estimated introduction rate of 1.03$\times 10^{-3}$ introductions/(m$\times$yr), under the assumption of no active dispersal between sampling sites and a spread rate of $10$ m/yr as obtained from literature from other regions.  Attributing all occurrences to passive dispersal is likely to have produced an over-estimate of the introduction rate in \cite{CameronBayne2009}.


To illustrate how the introduction and spread rate estimates produced by this analysis can be used to model invasive species distributions, we modelled earthworm distribution within the Alberta Pacific Forest Industries Forest Management Area (Al-Pac FMA), a 59 054 km$^2$ area in north-eastern Alberta. We obtained road network and Alberta Vegetation Inventory (AVI) data from Alberta-Pacific Forest Industries and additional data on road ages from the Mistakiis Institute. In ArcGIS 10.1, we generated points every 10 m along the road network, randomly invaded these points using the introduction rate from our analysis, and created buffers around invaded points using the spread rate, as described in \cite{CameronBayne2009}. We produced a map of the predicted areal extent of earthworms for 2006 (the year for which we were able to obtain road age data), which was  intersected with a GIS layer containing forest habitat suitable for invasion. All forest types were considered to be suitable except stands in which black spruce \textit{Picea mariana} or tamarack \textit{Larix laricina} were dominant, as such forests have highly acidic soils and are thus less likely to be colonized by earthworms \citep{Bouche1977,EdwardsBohlen1996}.  
Because we did not have a large number of sites sampled representatively across the Al-Pac FMA, we created the maps in Figure \ref{fig:wormmap}  based on three sets of introduction and spread rates, corresponding to high, middle, and low estimates.  The extreme values (frequent introduction with fast spread vs rare introduction with slow spread) were used to heuristically illustrate the magnitude of the prediction uncertainty for the areal extent of earthworms in 2006, while the middle values were used to illustrate predictions of the extent after 50 years following data collection.
Since the  Al-Pac FMA is located between groups 1-4, a low value of $(\lambda,\nu)$ was selected from the lower end of the bivariate posterior distributions for groups 2 and 4, while a high value of $(\lambda,\nu)$ was selected from the higher end of the distributions for groups 1 and 3. The middle value was taken to be the median of the four group posteriors. All three points are shown in Figure \ref{fig:bivariate_posterior}.

\begin{figure}
\centering
   \includegraphics[trim = 0.5cm 0cm 10cm 4cm,  width=4in]{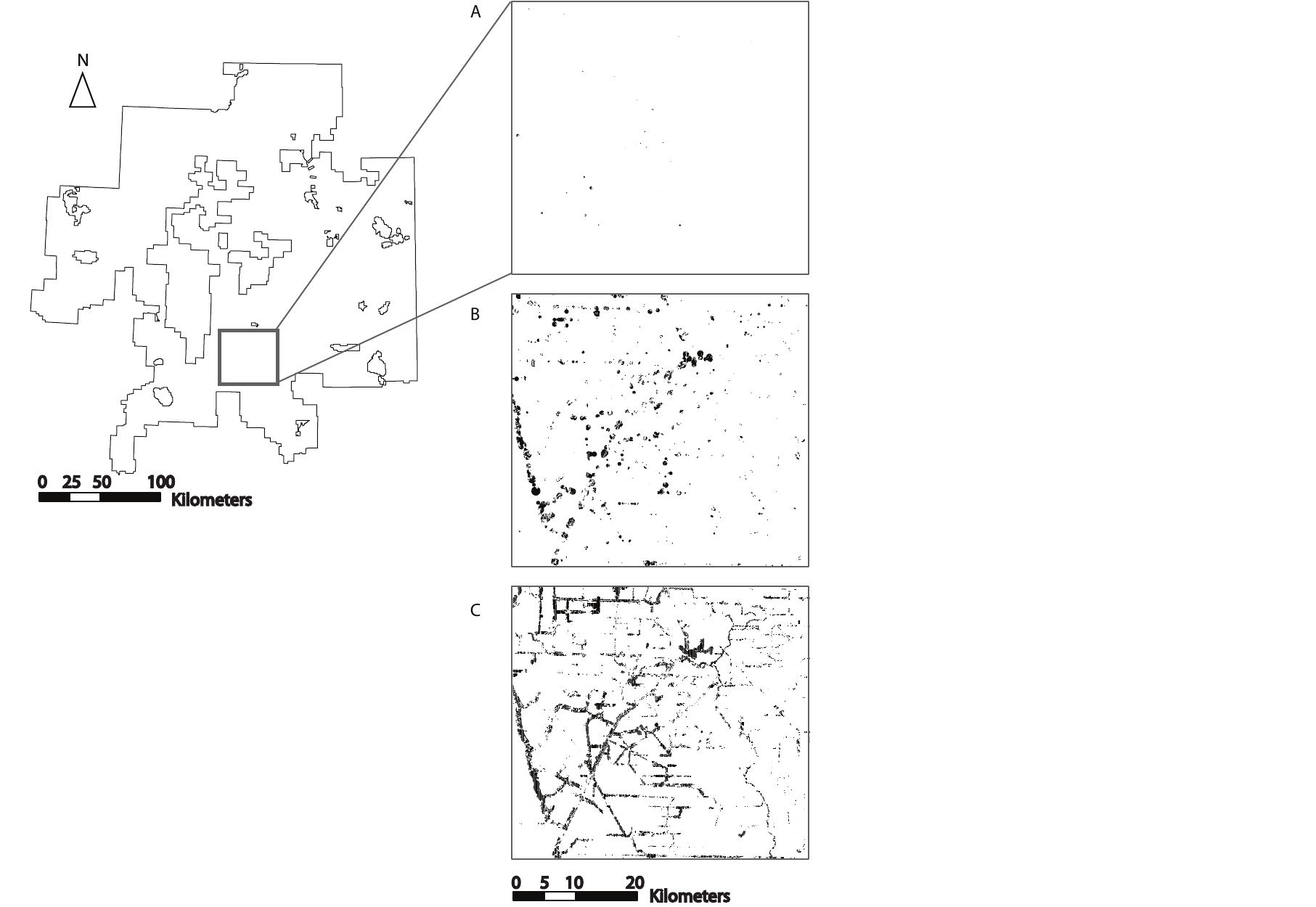}                 
  \caption{Enlarged area within the Al-Pac FMA simulation region (left).  Predicted areal extent of earthworms for 2006 (right) are indicated in black under three different combinations of introduction and spread rate: (A) $1.5915\times 10^{-6}$ introductions per meter of road per year and  11.596 m/yr; and (B) $3.7011\times10^{-5}$ introductions/(m$\times$yr) and 16.655  m/yr; (C) $1.03\times 10^{-3}$ introductions/(m$\times$yr) and 10 m/yr, as reported in \cite{CameronBayne2009}.}
\label{fig:wormmap}
\end{figure}

\begin{figure}
\centering
   \includegraphics[trim = 4cm 0cm 4cm 4cm,width=3.5in]{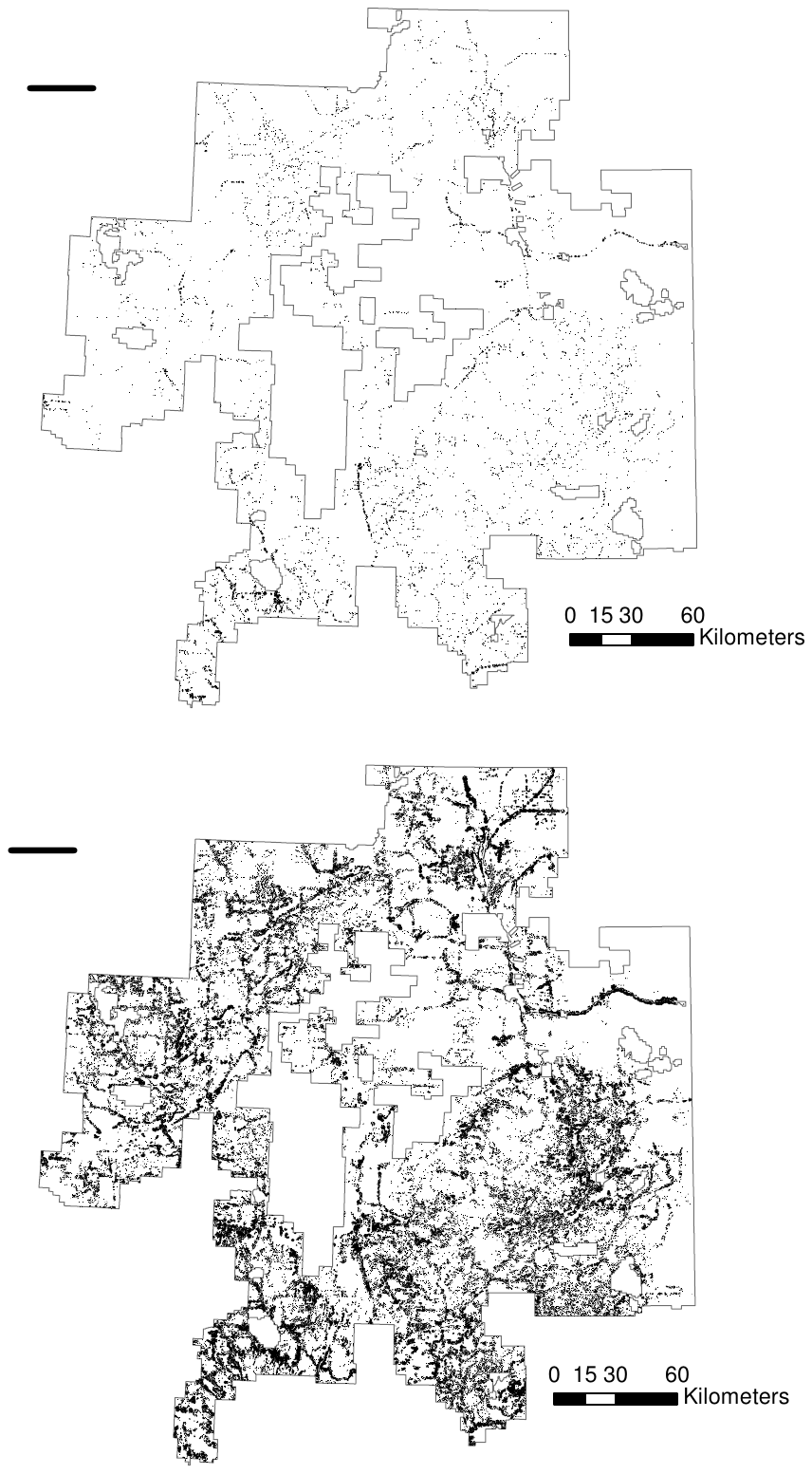}                 
  \caption{Al-Pac FMA simulation region.  Predicted areal extent of earthworms for 2006 (above) and 2056 (below) are indicated in black under an introduction rate of 4.0$\times10^{-5}$ introductions per meter of road per year and a spread rate of 14.0932 m/yr. The area invaded is 691.17 km$^2$ (2.83\% of suitable habitat) in 2006 and 9404.36 km$^2$ (38.46\% of suitable habitat) in 2056.}
\label{fig:wormmap_pred}
\end{figure}

%
%
%

The total area of suitable habitat within the Al-Pac FMA in 2006 was 24 449.3 $\mbox{km}^2$, with a total road length of 22 068 km. Using an introduction rate of $1.5915\times 10^{-6}$ introductions per meter of road per year and a spread rate of 11.596 m/yr, our model predicts that 8.02 $\mbox{km}^2$ (.03\%) of suitable habitat was invaded by 2006.  When using the higher values of $3.7011\times10^{-5}$ introductions/(m$\times$yr) and 16.655  m/yr, the model predicts that $905.65\mbox{km}^2$ was invaded (3.70\% of total suitable habitat).   Both maps are shown in Figure \ref{fig:wormmap}.
In comparison, Cameron and Bayne (2009) predicted 9.09\% of the total area within the Al-Pac FMA would be invaded by earthworms based on their higher estimated rate of introduction of 1.03$\times 10^{-3}$ introductions/(m$\times$yr), and a fixed, non-estimated spread rate of 10 m/yr.   However, our current approach to estimating the spread and introduction rates accounts for active dispersal and should produce more realistic estimates.

The increase in efficiency of the proposed transdimensional ABC algorithm over traditional fixed-dimensional formulations such as Metropolis-Hastings ABC is illustrated in Figure \ref{tab:coverage}.  A heat map of the approximate marginal density of introduction locations for one particular road is shown. Because the location vector changes dimension with the number of introductions, a fixed-dimensional ABC algorithm requires a forward-simulation step (rejection ABC) for the conditional simulation of this parameter. Instead, a transdimensional kernel allows us to make use of dependent proposals, thereby reducing Monte Carlo error when the areas of prior and posterior probabilities do not match.  In this case, the prior on introduction locations is uniform over the event horizon, but the regions of high posterior density are concentrated in a very thin band around the edges of the horizon, rendering conditional simulation extremely inefficient.


\section{Conclusion}\label{sec:conclusion}

Invasions frequently occur through a combination of long-distance jump dispersal events and diffusive spread around invaded sites, resulting in likelihoods that are often intractable. ABC methods provide an approximate inferential framework when the likelihood of the data cannot be evaluated. This paper develops a new efficient ABC sampler for a large class of models with infeasible or fully intractable likelihoods containing variable-dimensional integrals over a set of latent variables. 

This class of models is used in a variety of applications. Spatio-temporal dynamic systems often combine stochastic generating mechanisms with complex time-evolution models, so that evaluation of the likelihood requires integration over a variable number of latent point sources (e.g., \cite{vanLieshoutvanZwet2001}). Kingman's coalescent models \citep{Kingman1982} in genetics are another example where observed data is generated via a process that depends on the number and location of latent branching points defining a genealogical tree. Likelihoods of the genetic data are typically unavailable and consist of an integral over the space of all branches (e.g., \cite{TavareEtAl1997}). 

Latent variable models with intractable likelihoods pose a challenge to existing ABC-MCMC samplers defined on fixed-dimensional probability spaces.   In contrast, we were able to quickly obtain 250 000 MCMC samples per group of roads under each of the conditional values of $\nu$ using the more efficient transdimensional ABC algorithm proposed in this paper, even when requiring exact matches between summaries ($\epsilon = 0$). 

The main structural limitation of our methodology is that the resulting inference is approximate, controlled by the degree of sufficiency of our selected summary statistics and the chosen error tolerance.
However, this limitation is inherent in the problem of inferring parameters under intractable likelihoods. All simulation-based methods suffer from the problem of dimensionality, precluding exact likelihood-free inference.
Compared to established ABC methods, the sampling efficiency from our methodology allows strict error tolerances to be imposed, thereby improving the approximation.  The transdimensional ABC approach proposed in this paper can be applied to problems where the likelihood consists of intractable variable-dimension integrals.




\section*{Acknowledgements}

We would like to thank Dr. Mark Lewis, Dr. Marty Krkosek, and Stephanie Peacock for their very useful suggestions regarding our model, and to the Bamfield Marine Sciences Centre for providing the opportunity to begin this project as part of a Models in Ecology course.  We would also like to acknowledge our many field technicians for assistance with data collection and Charlene Nielsen for her assistance in creating the GIS spread maps.
The authors were funded by the Natural Sciences and Engineering Research Council of Canada (NSERC).  The project was also funded by the Alberta Biodiversity Monitoring Institute, Alberta Innovates Bio Solutions, and the Alberta Livestock and Meat Agency.  O.C. gratefully acknowledges the support of the Pacific Institute for Mathematical Sciences International Graduate Training Centre in Mathematical Biology.

\section*{References}

\bibliography{mybibfile}

\begin{thebibliography}{10}
\expandafter\ifx\csname url\endcsname\relax
  \def\url#1{\texttt{#1}}\fi
\expandafter\ifx\csname urlprefix\endcsname\relax\def\urlprefix{URL }\fi
\expandafter\ifx\csname href\endcsname\relax
  \def\href#1#2{#2} \def\path#1{#1}\fi

\bibitem{Ricciardi2007}
A.~Ricciardi, Are modern biological invasions an unprecedented form of global
  change?, Conservation Biology 21~(2) (2007) 329--336.

\bibitem{Fang2005}
W.~Fang, Spatial analysis of an invasion front of {A}cer platanoides: dynamic
  inferences from static data, Ecography 28~(3) (2005) 283--294.

\bibitem{Abbott2006}
K.~L. Abbott, Spatial dynamics of supercolonies of the invasive yellow crazy
  ant, {A}noplolepis gracilipes, on {C}hristmas {I}sland, {I}ndian {O}cean,
  Diversity and Distributions 12~(1) (2006) 101--110.

\bibitem{McIntireFajardo2009}
E.~J.~B. McIntire, A.~Fajardo, Beyond description: The active and effective way
  to infer processes from spatial patterns, Ecology 90~(1) (2009) pp. 46--56.

\bibitem{FrelichEtAl2006}
L.~E. Frelich, C.~M. Hale, S.~Scheu, A.~R. Holdsworth, L.~Heneghan, P.~J.
  Bohlen, P.~B. Reich, Earthworm invasion into previously earthworm-free
  temperate and boreal forests, Biological Invasions 8 (2006) 1235--1245.

\bibitem{HendrixEtAl2008}
P.~F. Hendrix, M.~A. Callaham, J.~M. Drake, C.-Y. Huang, S.~W. James, B.~A.
  Snyder, W.~Zhang, Pandora{\textquoteright}s box contained bait - the global
  problem of introduced earthworms, Annual Review of Ecology, Evolution, and
  Systematics 39 (2008) 593–613.

\bibitem{MarinissenBosch1992}
J.~C.~Y. Marinissen, F.~van~den Bosch, Colonization of new habitats by
  earthworms, Oecologia 91~(3) (1992) pp. 371--376.

\bibitem{GundaleEtAl2005}
M.~J. Gundale, W.~M. Jolly, T.~H. Deluca, Susceptibility of a northern hardwood
  forest to exotic earthworm invasion, Conservation Biology 19~(4) (2005)
  1075--1083.

\bibitem{CameronEtAl2007}
E.~K. Cameron, E.~M. Bayne, M.~Clapperton, Human-facilitated invasion of exotic
  earthworms into northern boreal forests, Ecoscience 14 (2007) 482--490.

\bibitem{BohlenEtAl2004}
P.~J. Bohlen, P.~M. Groffman, T.~J. Fahey, M.~C. Fisk, E.~Suarez, D.~M.
  Pelletier, R.~T. Fahey, Ecosystem consequences of exotic earthworm invasion
  of north temperate forests, Ecosystems 7~(1) (2004) 1--12.

\bibitem{HaleEtAl2006}
C.~M. Hale, L.~E. Frelich, P.~B. Reich, Changes in hardwood forest understory
  plant communities in response to {E}uropean earthworm invasions, Ecology
  87~(7) (2006) 1637--1649.

\bibitem{LossBlair2011}
S.~R. Loss, R.~B. Blair, Reduced density and nest survival of ground-nesting
  songbirds relative to earthworm invasions in northern hardwood forests,
  Conservation Biology 25~(5) (2011) 983--992.

\bibitem{ScheuParkinson1994}
S.~Scheu, D.~Parkinson, Effects of invasion of an aspen forest ({C}anada) by
  {D}endrobaena octaedra ({L}umbricidae) on plant growth, Ecology 75~(8) (1994)
  pp. 2348--2361.

\bibitem{NuzzoEtAl2009}
V.~A. Nuzzo, J.~C. Maerz, B.~Blossey, Earthworm invasion as the driving force
  behind plant invasion and community change in northeastern {N}orth {A}merican
  forests, Conservation Biology 23~(4) (2009) 966--974.

\bibitem{ShigesadaEtAl1995}
N.~Shigesada, K.~Kawasaki, Y.~Takeda, Modeling stratified diffusion in
  biological invasions, The American Naturalist 146~(2) (1995) pp. 229--251.

\bibitem{IllianEtAl2009}
J.~B. Illian, J.~Møller, R.~P. Waagepetersen, Hierarchical spatial point
  process analysis for a plant community with high biodiversity, Environmental
  and Ecological Statistics 16 (2009) 389--405.

\bibitem{DiggleGratton1984}
P.~J. Diggle, R.~J. Gratton, {M}onte {C}arlo methods of inference for implicit
  statistical models, Journal of the Royal Statistical Society. Series B
  (Methodological) 46~(2) (1984) pp. 193--227.

\bibitem{BeaumontEtAl2002}
M.~A. Beaumont, W.~Zhang, D.~J. Balding, Approximate {B}ayesian computation in
  population genetics, Genetics 162 (2002) 2025--2035.

\bibitem{MarjoramEtAl2003}
P.~Marjoram, J.~Molitor, V.~Plagnol, S.~Tavar\'e, {M}arkov chain {M}onte
  {C}arlo without likelihoods, Proceedings of the National Academy of Sciences
  100~(26) (2003) 15324--15328.

\bibitem{MillerEtAl2005}
N.~Miller, A.~Estoup, S.~Toepfer, D.~Bourguet, L.~Lapchin, S.~Derridj, K.~S.
  Kim, P.~Reynaud, L.~Furlan, T.~Guillemaud, Multiple transatlantic
  introductions of the western corn rootworm, Science 310~(5750) (2005) 992.

\bibitem{LombaertEtAl2010}
E.~Lombaert, T.~Guillemaud, J.-M. Cornuet, T.~Malausa, B.~Facon, A.~Estoup,
  Bridgehead effect in the worldwide invasion of the biocontrol harlequin
  ladybird, PLoS ONE 5 (2010) e9743.

\bibitem{FearnheadPrangle2012}
P.~Fearnhead, D.~Prangle, Constructing summary statistics for approximate
  {B}ayesian computation: semi-automatic approximate {B}ayesian computation,
  Journal of the Royal Statistical Society: Series B (Statistical Methodology)
  74~(3) (2012) 419--474.

\bibitem{Green1995}
P.~J. Green, Reversible jump {M}arkov chain {M}onte {C}arlo computation and
  {B}ayesian model determination, Biometrika 82~(4) (1995) pp. 711--732.

\bibitem{RichardsonGreen1997}
S.~Richardson, P.~J. Green, On {B}ayesian analysis of mixtures with an unknown
  number of components, Journal of the Royal Statistical Society. Series B
  (Methodological) 59~(4) (1997) pp. 731--792.

\bibitem{SissonFan2010}
S.~A. Sisson, Y.~Fan, Likelihood-free {M}arkov chain {M}onte {C}arlo, in: G.~J.
  S.~P.~Brooks, A.~Gelman, X.-L. Meng (Eds.), Handbook of {M}arkov Chain
  {M}onte {C}arlo, Chapman and Hall/CRC Press, 2010.

\bibitem{congdon2006}
P.~Congdon, {B}ayesian model choice based on {M}onte {C}arlo estimates of
  posterior model probabilities, Computational statistics \& data analysis
  50~(2) (2006) 346--357.

\bibitem{CameronBayne2009}
E.~K. Cameron, E.~M. Bayne, Road age and its importance in earthworm invasion
  of northern boreal forests, Journal of Applied Ecology 46~(1) (2009) 28--36.

\bibitem{TiunovEtAl2006}
A.~V. Tiunov, C.~M. Hale, A.~R. Holdsworth, T.~S. Vsevolodova-Perel, Invasion
  patterns of {L}umbricidae into the previously earthworm-free areas of
  northeastern {E}urope and the western {G}reat {L}akes region of {N}orth
  {A}merica, Biological Invasions 8 (2006) 1223--1234.

\bibitem{Hale2008}
C.~M. Hale, Evidence for human-mediated dispersal of exotic earthworms: support
  for exploring strategies to limit further spread, Molecular Ecology 17 (2008)
  1165--1169.

\bibitem{Bouche1977}
M.~Bouch{\'e}, Strategies lombriciennes. {S}oil organisms as components of
  ecosystems, in: U.~Lohm, T.~Persson (Eds.), Proceedings of the VI
  International Soil Zoology Colloquium of the International Society of Soil
  Science, Ecological Bulletin 25, Stockholm, Sweden, 1977, pp. 122--132.

\bibitem{EdwardsBohlen1996}
C.~Edwards, P.~Bohlen, Biology and Ecology of Earthworms, Chapman and Hall,
  Lodon, UK, 1996.

\bibitem{vanLieshoutvanZwet2001}
M.~N.~M. van Lieshout, E.~W. van Zwet, Exact sampling from conditional
  {B}oolean models with applications to maximum likelihood inference, Advances
  in Applied Probability 33~(2) (2001) 339--353.

\bibitem{Kingman1982}
J.~F.~C. Kingman, {On the Genealogy of Large Populations}, Journal of Applied
  Probability 19 (1982) 27--43.

\bibitem{TavareEtAl1997}
S.~Tavar\'{e}, D.~J. Balding, R.~C. Griffiths, P.~Donnelly, Inferring
  coalescence times from {DNA} sequence data, Genetics 145~(2) (1997) pp.
  505–518.

\bibitem{Tierney1998}
L.~Tierney, A note on {M}etropolis-{H}astings kernels for general state spaces,
  The Annals of Applied Probability 8~(1) (1998) pp. 1--9.

\end{thebibliography}


\appendix

\section{}

\subsection{Transdimensional ABC algorithm}

\begin{theorem}
The Markov chain generated via Algorithm 2 has invariant distribution:
\begin{align}
 \pi_{ABC} \left(\boldsymbol\theta  \mid  Y \right)  
&  \propto
\pi\left(\boldsymbol\theta\right)\,
{\sum_{k=0}^\infty
\int_{\mathcal{Y}}   \int_{\mathbb{R}^{n_k}}   
p\left(D \mid {\bf x}_k,\boldsymbol\theta \right) \,
p \left({\bf x}_k \mid k, \boldsymbol\theta \right) \,
p\left(k  \mid  \boldsymbol\theta\right) \,
\mbox{d} {\bf x}_k}\,
p\left( s(Y) \mid D\right) \mbox{d}D.
\label{eqn:ABCposteriorAppendix}
\end{align}
\end{theorem}

\begin{proof} 
Define the augmented spaces
$\mathcal{M}_k =  \Theta \times {k} \times \mathbb{R}^{n_k}  \times \mathbb{R}^{m_k} \times \mathcal{Y}\times \mathcal{Y}$, for $k\in\mathcal{K}$, and the variable-dimensional parameter space $\mathcal{M} = \cup_{k \in \mathcal{K}} \mathcal{M}_k$. 
Our goal is to construct a Markov chain with invariant distribution,
\begin{align*}
\pi(\mbox{d}z) = \pi(z) \mbox{d}z 
 \propto 
p\left( s(Y) \mid D_k\right)\,
p\left(D_k \mid {\bf x}_k, \boldsymbol\theta \right) \,
p\left({\bf x}_k \mid  k, \boldsymbol\theta \right) \,
p\left(k  \mid  \boldsymbol\theta \right) \,
\pi(\boldsymbol\theta)\;
\pi(\boldsymbol{u}_k)
\mbox{d}\lambda, \; z\in \mathcal{M}_k,\, k\in\mathcal{K}.
\end{align*}

\noindent on the measurable space $\big(\mathcal{M}, \sigma(\mathcal{M})\big)$, where $\sigma(\mathcal{M})$ is the sigma algebra generated by subsets of $\mathcal{M}$ and $\lambda$ is the Lebesgue measure on $\sigma(\mathcal{M})$.  

We use an argument similar to Green (1995) 
to define transitions on this parameter space.  First, we restrict attention to moves between any two model spaces $\mathcal{M}_i$ and $\mathcal{M}_j$. Consider elements $v = (\boldsymbol\theta_i, i,{\bf x}_i,{\bf u}_i,D_i) \in \mathcal{M}_i$ and $w = (\boldsymbol\theta_j, j,{\bf x}_j,{\bf u}_j,D_j) \in \mathcal{M}_j$, where $\boldsymbol{x}_k\in\mathbb{R}^{n_k}$, $\boldsymbol{u}_k\in\mathbb{R}^{m_k}$, $k=i, j$.  The constraint $m_i + n_i = m_j + n_j$ ensures that the dimension of $v$ and $w$ match.  Next, define the following proposal distributions for a move from $v$ to $w$ and back:
\begin{align*}
 & Q\left( v , \mbox{d}w \right) = 
q(\boldsymbol\theta_j \mid \boldsymbol\theta_i)\,
q(j \mid i)\,
q({\bf x}_j,{\bf u}_j \mid {\bf x}_i,{\bf u}_i)\,
p(D_j  \mid  {\bf x}_j)\;
\mbox{d}\lambda,\\
& Q\left( w , \mbox{d}v \right) = 
q(\boldsymbol\theta_i \mid \boldsymbol\theta_j)\,
q(i \mid j)\,
q({\bf x}_i,{\bf u}_i \mid {\bf x}_j, {\bf u}_j )\,
p(D_i  \mid  {\bf x}_i )\;
\mbox{d}\lambda,
\end{align*}
where $\lambda$ is the Lebesque measure on $\mathcal{B}=\sigma(\mathcal{M}_k)$, $k =1,2$.  Now let $\alpha: \mathcal{M}\times\mathcal{M} \to [0,1]$ be an acceptance probability, and define $\delta_{a}$ to be the Dirac delta measure on $\mathcal{B}$ centered at $a$.  Then, the corresponding transition kernels are: 
\begin{align*}
& P\left(v , \mbox{d}w\right) 
= 
\alpha\left(v , w \right) 
Q\left( v , \mbox{d}w \right) 
+
\delta_{v}\left(\mbox{d}w\right) 
\int_{\mathcal{M}_j}  
\left [ 1 -  \alpha\left( v , w \right) \right ]
 Q\left( v , \mbox{d}m \right),\\
& P\left(w , \mbox{d}v\right) 
= 
\alpha\left(w , v \right) 
Q\left( w , \mbox{d}v \right) 
+
\delta_{w}\left(\mbox{d}v\right) 
\int_{\mathcal{M}_i}  
\left [ 1 -  \alpha\left( w , v \right) \right ]
 Q\left( w , \mbox{d}m \right),
\end{align*}
on the measurable space $\left( \mathcal{M}_i \times \mathcal{M}_j, \mathcal{B}\times\mathcal{B}\right)$, and $\left( \mathcal{M}_j \times \mathcal{M}_i, \mathcal{B}\times\mathcal{B}\right)$, respectively.  We assume that the Markov chain associated with this transition kernel is aperiodic and irreducible (true if proposal distribution $Q\left( v , \mbox{d}w \right)$ generates an aperiodic and irreducible chain).  
The detailed balance condition,
\begin{align*}
& \pi \left( \mbox{d}v \right)P\left( v , \mbox{d}w \right)  
 = \pi \left( \mbox{d}w \right)P\left( w , \mbox{d}v \right),
\end{align*} 
is satisfied iff for $A,B \in \mathcal{B}$,
\begin{align*}
& \pi \left(\mbox{d}v \right)Q\left( v , \mbox{d}w \right)
\alpha\left( v , w \right) 
 = \pi \left( \mbox{d} w \right)Q\left( w , \mbox{d}v \right)
\alpha\left( w , v \right), \quad v\in A,\, w \in B.
\end{align*}

Next, define a diffeomorphic transformation, $\phi_{ij}:\mathbb{R}^{n_i}\times\mathbb{R}^{m_i}\to\mathbb{R}^{n_j}\times\mathbb{R}^{m_j}$.  
Using Lebesgue measure $\lambda$ on $\left( \mathcal{M}_i \times \mathcal{M}_j, \mathcal{B}\times\mathcal{B}\right)$ define the new measure,
\begin{align*}
\mu\left( A \times B \right) & \equiv 
\lambda\big[ \{ v \in A \cap \mathcal{M}_i \, , \, \left(\boldsymbol\theta, j, \phi_{ij}(\boldsymbol{x}_i, \boldsymbol{u}_i), D\right) \in B \cap \mathcal{M}_{j} \} \\
& \quad\quad \cup \{  v \in A \cap \mathcal{M}_{j} \, , \, \left(\boldsymbol\theta, i, \phi_{ji}(\boldsymbol{x}_j,\boldsymbol{u}_j), D\right) \in B \cap \mathcal{M}_{i} \}  \big]\\
&=\lambda\big[ \{ v \in A \cap \mathcal{M}_i \, , \, \left(\boldsymbol\theta, j, \phi_{ij}(\boldsymbol{x}_i, \boldsymbol{u}_i), D\right) \in B \cap \mathcal{M}_{j} \} \\
& \quad\quad + \{  v \in A \cap \mathcal{M}_{j} \, , \, \left(\boldsymbol\theta, i, \phi_{ji}(\boldsymbol{x}_j,\boldsymbol{u}_j), D\right) \in B \cap \mathcal{M}_{i} \}  \big]\\
&= \lambda\big[ \{w \in B \cap \mathcal{M}_j \, , \, \left(\boldsymbol\theta, i, \phi_{ji}(\boldsymbol{x}_j,\boldsymbol{u}_j), D\right) \in A \cap \mathcal{M}_{i} \}\big] \\
& \quad\quad+ \lambda\big[ \{  w \in B \cap \mathcal{M}_i 
\, , \, \left(\boldsymbol\theta, j, \phi_{ij}(\boldsymbol{x}_i,\boldsymbol{u}_i), D\right) \in A \cap \mathcal{M}_j \}  \big]\\
&=\mu\left( B \times A \right) 
\end{align*} 
which is symmetric on $\left( \mathcal{M}_j \times \mathcal{M}_i, \mathcal{B}\times\mathcal{B}\right)$.  Then $\pi \left(\mbox{d}v \right)Q\left( v , \mbox{d}w \right)$ and $\pi \left(\mbox{d}w \right)Q\left( w , \mbox{d}v \right)$ have densities with respect to $\mu$.  These are given by,
\begin{align*}
f(v,w) &= \pi\left(\mbox{d}v \right)  p(D_j  \mid  {\bf x}_j ) q(\boldsymbol\theta_j \mid \boldsymbol\theta_i) q(j \mid i) g_{ij}({\bf u}_j \mid {\bf u}_i) |J_{ij}|, \\
f(w,v) &= \pi\left(\mbox{d}w \right)  p(D_i  \mid  {\bf x}_i )  q(\boldsymbol\theta_i \mid \boldsymbol\theta_j) q(i \mid j) g_{ji}({\bf u}_i \mid {\bf u}_j),
\end{align*}

\noindent respectively.
Then reversibility across all moves is guaranteed 
\cite{Green1995,Tierney1998}  
if $\alpha(v,w) = \min \{1, f(v,w) / f(w,v) \}$ for all $i,j \in \mathcal{K}$, and so detailed balance is satisfied under
\begin{align*}
\alpha\left( v , w \right) &= 
\min\left\{ 1, 
\frac{ p\left( s(Y) \mid D_j \right)}{ p\left( s(Y) \mid D_i \right)}
\frac{ p({\bf x}_j \mid j)p(j \mid \boldsymbol\theta_j)\pi(\boldsymbol\theta_j) }{p({\bf x}_i \mid i)p(i \mid \boldsymbol\theta_i)\pi(\boldsymbol\theta_i)} 
 \frac{q(\boldsymbol\theta_j \mid \boldsymbol\theta_i)q(j \mid i) g_{ij}({\bf u}_j \mid {\bf u}_i)}
{q(\boldsymbol\theta_i \mid \boldsymbol\theta_j)q(i \mid j) g_{ji}({\bf u}_i \mid {\bf u}_j) |J_{ij}| } \right\},
\end{align*}
for all move types.  As Algorithm 2 does not return the values of the auxiliary vectors, ${\bf u}_k$, the resulting marginal invariant distribution is (\ref{eqn:ABCposteriorAppendix}).
\end{proof}

\end{document}


\begin{frontmatter}

\title{Supplementary Materials for \\
\textbf{Transdimensional Approximate Bayesian Computation for Inference on Invasive Species Models with Latent Variables of Unknown Dimension}}

\author{Oksana A. Chkrebtii\fnref{myemail}}
\address{Department of Statistics, The Ohio State University, Columbus, OH, USA}
\fntext[myemail]{email: oksana@stat.osu.edu}

\author{Erin K. Cameron\fnref{currentaddress}}
\address{Department of Biological Sciences, University of Alberta, Edmonton, AB, Canada}
\fntext[currentaddress]{Current address: Metapopulation Research Group, Department of Biological and Environmental Sciences, 
PO Box 65 (Viikinkaari 1), 00014 University of Helsinki, Finland}

\author{David A. Campbell}
\address{Department of Statistics \& Actuarial Science, Simon Fraser University, Surrey, Canada}

\author{Erin M. Bayne}
\address{Department of Biological Sciences, University of Alberta, Edmonton, AB, Canada}


\end{frontmatter}



\section{Projected spread}

Figure \ref{fig:wormmap_small} shows predicted areal extent of earthworm presence for 2006 (the year of data collection) and after a further 50 years, under the average introduction and spread rate obtained over all groups in our analysis for an enlarged area within the Al-Pac FMA simulation region.

\begin{figure}
\centering
   \includegraphics[trim = 3cm 10cm 3cm 0cm, width=4in]{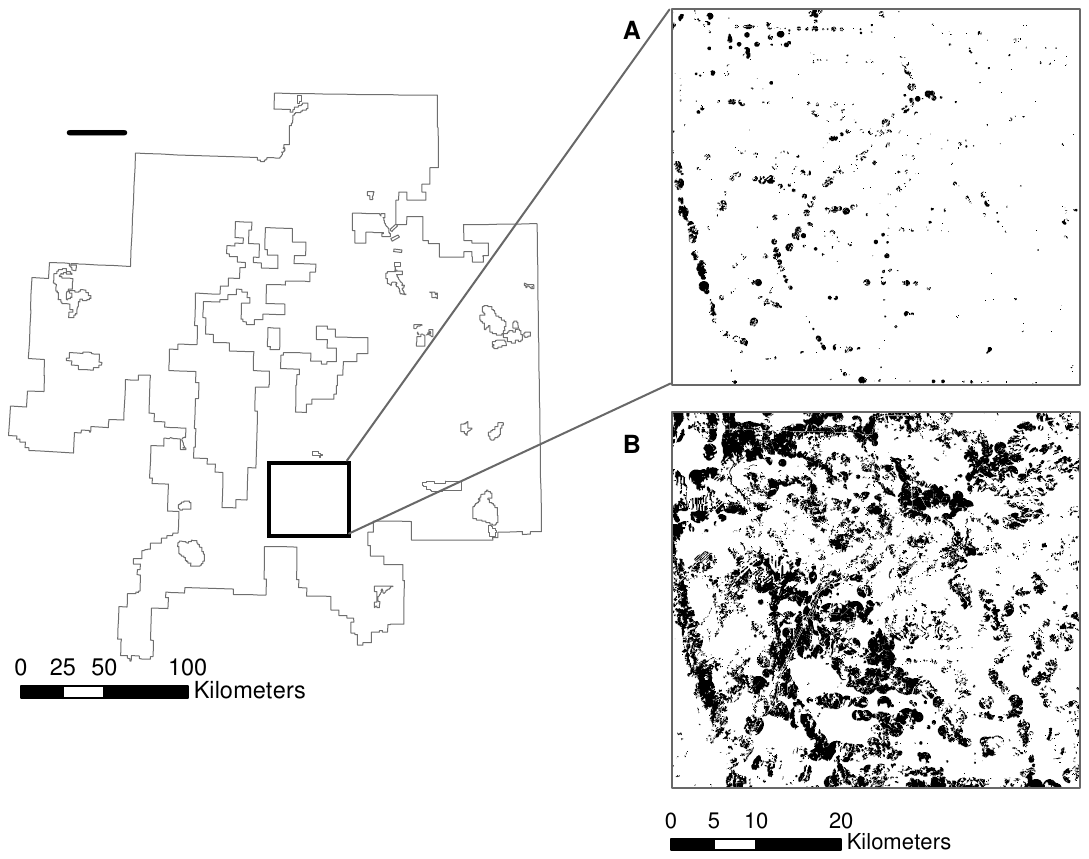}       
  \caption{Enlarged area within the Al-Pac FMA simulation region (left).  Predicted areal extent of earthworms for 2006 (right, above) and 2056 (right, below) are indicated in black under an introduction rate of 4.0$\times10^{-5}$ introductions per meter of road per year and a spread rate of 14.0932 m/yr.}
\label{fig:wormmap_small}
\end{figure}

\section{Likelihood calculations}

Exact inference on introduction and spread requires a closed-form representation of $p\left( Y_{rg} \mid {\bf x}_{rg}, \nu, T_{rg}\right)$ in order to compute the data likelihood,
\begin{align}
& p\left( Y_{rg} \mid \lambda_{g}, \nu, T_{rg}\right)
 = \quad 
\\
& \quad
  \sum_{k_{rg}=0}^\infty
\int_{H(\nu,T_{rg})}   
p\left( Y_{rg} \mid {\bf x}_{rg}, \nu, T_{rg}\right) 
 p \left({\bf x}_{rg} \mid k_{rg}, \nu, T_{rg} \right)
 p\left(k_{rg}  \mid  \lambda_g, \nu\right)
\mbox{d}{\bf x}_{rg} , \nonumber
\end{align}

\noindent 
For our simple model of earthworm introduction and spread, this likelihood can indeed be obtained analytically via a lengthly geometric and combinatorial computation.  Our motivating example serves to highlight the difficulty of obtaining an analogous closed form for more complex models, e.g. in the case of nonlinear time-dependent spread, or the presence of obstacles space and over time.


Figure \ref{fig:polygons} shows a partition of the space-time horizon $\mathcal{H}_{gr}$, obtained by identifying and grouping possible introduction locations by their effect on the data. 
The shapes of the partitions depend on the model, the spread rate $\nu$, and the road age $T_{rg}$.  Note that a more complicated spread mechanism or horizon topography will necessarily change the shape of these regions and their effect on the data, making the combinatorial argument presented below invalid for the general case.

\begin{figure}
\centering
  \includegraphics[trim = 6cm 5cm 6cm 8cm, width=0.8\textwidth]{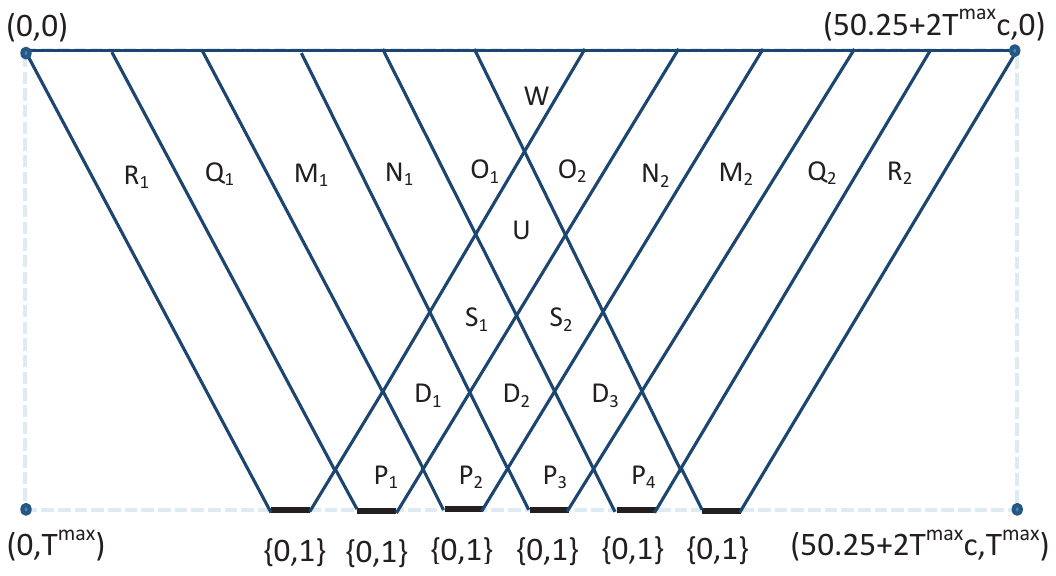}
  \caption{Spatio-temporal horizon partitioned by grouping possible introduction locations by their effect on the simulated data.}
\label{fig:polygons}
\end{figure}

Each case considered below corresponds to a possible combination of introductions that could have led to the observed data under our deterministic spread model.  For example, we can generate the data $Y_{rg} = (0,1,1,1,1,1)$ with at least one introduction occurring in each of regions $U$ and $R_2$ with any number of remaining introductions in $D_1,D_2,D_3,P_1,P_2,P_3,P_4,S_1,S_2$.  Alternatively we can generate the same data with at least one introduction occurring in each of $R_2, P_1,P_2,P_3,P_4$.  In other words, we can generate the data by any one of a large number of mutually exclusive ways.

The following subsections are grouped by equal posterior probability. 
Denote by $\mathcal{I}$ the set $\{1,\ldots,k\}$.  Let the notation $\mbox{x}_k^{(j)} \in S$, $j\in\{1,\ldots,k\}$ describes the case when the $j$th introduction occurs within the spatio-temporal region $S \subset \mathcal{H}_{gr}$.  


\subsection{000000}

\begin{align*}
p\left(000000 \mid {\bf x}_k,\nu, T\right) &= \mathbb{I}\{k=0\}
 \end{align*}
%
%


\subsection{100000,000001}

\begin{align*}
p\left(100000 \mid {\bf x}_k,\nu, T\right)
&= \mathbb{I}\{k>0; \;\forall i\in\mathcal{I} ,\; \mbox{x}^{(i)} \in R_1 \}\line
p\left(000001 \mid {\bf x}_k,\nu, T\right)
&= \mathbb{I}\{k>0; \;\forall i\in\mathcal{I} ,\; \mbox{x}^{(i)} \in R_2 \} 
\end{align*}



\subsection{010000,001000,000100,000010}

\begin{align*}
p\left(010000 \mid {\bf x}_k,\nu, T\right)
&= \mathbb{I}\{k>0;\;\forall i\in\mathcal{I},\; \mbox{x}^{(i)} \in P_1 \}\line
p\left(001000 \mid {\bf x}_k,\nu, T\right)
&= \mathbb{I}\{k>0;\;\forall i\in\mathcal{I},\; \mbox{x}^{(i)} \in P_2 \}\line
p\left(000100 \mid {\bf x}_k,\nu, T\right)
&= \mathbb{I}\{k>0;\;\forall i\in\mathcal{I},\; \mbox{x}^{(i)} \in P_3 \}\line
p\left(000010 \mid {\bf x}_k,\nu, T\right)
&= \mathbb{I}\{k>0;\;\forall i\in\mathcal{I},\; \mbox{x}^{(i)} \in P_4 \}
\end{align*}



\subsection{010100,010010,001010}

\begin{align*}
p\left(010100 \mid {\bf x}_k,\nu, T\right)
&= \mathbb{I}\{k>1; \;\exists \; j_1, j_2\in\mathcal{I} :  \mbox{x}^{(j_1)} \in P_1 , \mbox{x}^{(j_2)} \in P_3 ;\\
& \hspace{1cm}  \;\forall i\neq j_1,j_2\in\mathcal{I},\;  \mbox{x}^{(i)} \in P_1 \cup P_3\}\line
p\left(010010 \mid {\bf x}_k,\nu, T\right)
&= \mathbb{I}\{k>1; \;\exists \; j_1, j_2\in\mathcal{I} :  \mbox{x}^{(j_1)} \in P_1 , \mbox{x}^{(j_2)} \in P_4 ;\\
& \hspace{1cm}  \;\forall i\neq j_1,j_2\in\mathcal{I},\;  \mbox{x}^{(i)} \in P_1 \cup P_4\}\line
p\left(001010 \mid {\bf x}_k,\nu, T\right)
&= \mathbb{I}\{k>1; \;\exists \; j_1, j_2\in\mathcal{I} :  \mbox{x}^{(j_1)} \in P_2 , \mbox{x}^{(j_2)} \in P_4 ;\\
& \hspace{1cm}  \;\forall i\neq j_1,j_2\in\mathcal{I},\;  \mbox{x}^{(i)} \in P_2 \cup P_4\}
\end{align*}

%


\subsection{100001}
\begin{align*}
p\left(100001 \mid {\bf x}_k,\nu, T\right)
&= \mathbb{I}\{k>1; \;\exists \; j_1, j_2\in\mathcal{I} :  \mbox{x}^{(j_1)} \in R_1 , \mbox{x}^{(j_2)} \in R_2 ;\\
& \hspace{1cm}  \;\forall i\neq j_1,j_2\in\mathcal{I},\;  \mbox{x}^{(i)} \in R_1 \cup R_2\}
\end{align*}


\subsection{101001,100101}
\begin{align*}
&p\left(101001 \mid {\bf x}_k,\nu, T\right)\\
&= \mathbb{I}\{k>2; 
 \exists \; j_1, j_2 , j_3\in\mathcal{I} :  \mbox{x}^{(j_1)} \in R_1 , \mbox{x}^{(j_2)} \in R_2 , \mbox{x}^{(j_3)} \in P_3 ;\\
& \hspace{1cm}  \;\forall i\neq j_1,j_2,j_3\in\mathcal{I},\;  \mbox{x}^{(i)} \in R_1 \cup R_2 \cup P_3\}\line
&p\left(100101 \mid {\bf x}_k,\nu, T\right)\\
&= \mathbb{I}\{k>2; 
 \exists \; j_1, j_2 , j_3\in\mathcal{I} :  \mbox{x}^{(j_1)} \in R_1 , \mbox{x}^{(j_2)} \in R_2 , \mbox{x}^{(j_3)} \in P_4 ;\\
& \hspace{1cm}  \;\forall i\neq j_1,j_2,j_3\in\mathcal{I},\;  \mbox{x}^{(i)} \in R_1 \cup R_2 \cup P_4\}
\end{align*}


\subsection{101000,100100,100010,000101,0001001,010001}
\begin{align*}
p\left(101000 \mid {\bf x}_k,\nu, T\right)
&= \mathbb{I}\{k>1; \;\exists \; j_1, j_2\in\mathcal{I} :  \mbox{x}^{(j_1)} \in R_1 , \mbox{x}^{(j_2)} \in P_2 ;\\
& \hspace{1cm}  \;\forall i\neq j_1,j_2\in\mathcal{I},\;  \mbox{x}^{(i)} \in R_1 \cup P_2\}\line
p\left(100100 \mid {\bf x}_k,\nu, T\right)
&= \mathbb{I}\{k>1; \;\exists \; j_1, j_2\in\mathcal{I} :  \mbox{x}^{(j_1)} \in R_1 , \mbox{x}^{(j_2)} \in P_3 ;\\
& \hspace{1cm}  \;\forall i\neq j_1,j_2\in\mathcal{I},\;  \mbox{x}^{(i)} \in R_1 \cup P_3\}\line
p\left(100010 \mid {\bf x}_k,\nu, T\right)
&= \mathbb{I}\{k>1; \;\exists \; j_1, j_2\in\mathcal{I} :  \mbox{x}^{(j_1)} \in R_1 , \mbox{x}^{(j_2)} \in P_4 ;\\
& \hspace{1cm}  \;\forall i\neq j_1,j_2\in\mathcal{I},\;  \mbox{x}^{(i)} \in R_1 \cup P_4\}\line
p\left(010001 \mid {\bf x}_k,\nu, T\right)
&= \mathbb{I}\{k>1; \;\exists \; j_1, j_2\in\mathcal{I} :  \mbox{x}^{(j_1)} \in R_2 , \mbox{x}^{(j_2)} \in P_1 ;\\
& \hspace{1cm}  \;\forall i\neq j_1,j_2\in\mathcal{I},\;  \mbox{x}^{(i)} \in R_2 \cup P_1\}\line
p\left(001001 \mid {\bf x}_k,\nu, T\right)
&= \mathbb{I}\{k>1; \;\exists \; j_1, j_2\in\mathcal{I} :  \mbox{x}^{(j_1)} \in R_2 , \mbox{x}^{(j_2)} \in P_2 ;\\
& \hspace{1cm}  \;\forall i\neq j_1,j_2\in\mathcal{I},\;  \mbox{x}^{(i)} \in R_2 \cup P_2\}\line
p\left(000101 \mid {\bf x}_k,\nu, T\right)
&= \mathbb{I}\{k>1; \;\exists \; j_1, j_2\in\mathcal{I} :  \mbox{x}^{(j_1)} \in R_2 , \mbox{x}^{(j_2)} \in P_3 ;\\
& \hspace{1cm}  \;\forall i\neq j_1,j_2\in\mathcal{I},\;  \mbox{x}^{(i)} \in R_2 \cup P_3\}
\end{align*}


\subsection{101010,010101}
\begin{align*}
&p\left(101010 \mid {\bf x}_k,\nu, T\right)\\
&= \mathbb{I}\{k>2; 
 \exists \; j_1, j_2 , j_3\in\mathcal{I} :  \mbox{x}^{(j_1)} \in R_1 , \mbox{x}^{(j_2)} \in P_3 , \mbox{x}^{(j_3)} \in P_4 ;\\
& \hspace{1cm}  \;\forall i\neq j_1,j_2,j_3\in\mathcal{I},\;  \mbox{x}^{(i)} \in R_1 \cup P_3 \cup P_4\}\line
&p\left(010101 \mid {\bf x}_k,\nu, T\right)\\
&= \mathbb{I}\{k>2; 
 \exists \; j_1, j_2 , j_3\in\mathcal{I} :  \mbox{x}^{(j_1)} \in P_1 , \mbox{x}^{(j_2)} \in P_3 , \mbox{x}^{(j_3)} \in R_2 ;\\
& \hspace{1cm}  \;\forall i\neq j_1,j_2,j_3\in\mathcal{I},\;  \mbox{x}^{(i)} \in P_1 \cup P_3 \cup R_2\}
\end{align*}


\subsection{110000,000011}
\begin{align*}
&p\left(110000 \mid {\bf x}_k,\nu, T\right)\\
&= \mathbb{I}\{k>0; \;\exists \; j_1\in\mathcal{I} :  \mbox{x}^{(j_1)} \in Q_1;
\;\forall i\neq j_1\in\mathcal{I},\;  \mbox{x}^{(i)} \in P_1 \cup Q_1 \cup R_1\}\liney
& + \mathbb{I}\{k>1; \;\exists \; j_1, j_2\in\mathcal{I} :  \mbox{x}^{(j_1)} \in P_1, \mbox{x}^{(j_2)} \in R_1;
 \;\forall i\neq j_1,j_2\in\mathcal{I},\;  \mbox{x}^{(i)} \in P_1 \cup R_1\}\line
&p\left(000011 \mid {\bf x}_k,\nu, T\right)\\
&= \mathbb{I}\{k>0; \;\exists \; j_1\in\mathcal{I} :  \mbox{x}^{(j_1)} \in Q_2;
  \;\forall i\neq j_1\in\mathcal{I},\;  \mbox{x}^{(i)} \in P_4 \cup Q_2 \cup R_2\}\liney
& + \mathbb{I}\{k>1; \;\exists \; j_1, j_2\in\mathcal{I} :  \mbox{x}^{(j_1)} \in P_4, \mbox{x}^{(j_2)} \in R_2;
 \;\forall i\neq j_1,j_2\in\mathcal{I},\;  \mbox{x}^{(i)} \in P_4 \cup R_2\}
\end{align*}


\subsection{110001,100011}
\begin{align*}
&p\left(110001 \mid {\bf x}_k,\nu, T\right)\\
&= \mathbb{I}\{k>1; \;\exists \; j_1, j_2\in\mathcal{I} :  \mbox{x}^{(j_1)} \in Q_1,\mbox{x}^{(j_2)} \in R_2;
 \;\forall i\neq j_1,j_2\in\mathcal{I},\;  \mbox{x}^{(i)} \in P_1 \cup Q_1 \cup R_1 \cup R_2\} \liney
& + \mathbb{I}\{k>2; \;\exists \; j_1, j_2, j_3\in\mathcal{I} :  \mbox{x}^{(j_1)} \in P_1, \mbox{x}^{(j_2)} \in R_1, \mbox{x}^{(j_3)} \in R_2;
 \;\forall i\neq j_1,j_2,j_3\in\mathcal{I},\;  \mbox{x}^{(i)} \in P_1 \cup R_1 \cup R_2\}\line
&p\left(100011 \mid {\bf x}_k,\nu, T\right)\\
&= \mathbb{I}\{k>1; \;\exists \; j_1, j_2\in\mathcal{I} :  \mbox{x}^{(j_1)} \in Q_2, \mbox{x}^{(j_2)} \in R_1;
  \;\forall i\neq j_1,j_2\in\mathcal{I},\;  \mbox{x}^{(i)} \in P_4 \cup Q_2 \cup R_1 \cup R_2\}\liney
& + \mathbb{I}\{k>2; \;\exists \; j_1, j_2 , j_3\in\mathcal{I} :  \mbox{x}^{(j_1)} \in P_4, \mbox{x}^{(j_2)} \in R_1,  \mbox{x}^{(j_3)} \in R_2;
 \;\forall i\neq j_1,j_2,j_3\in\mathcal{I},\;  \mbox{x}^{(i)} \in P_4 \cup R_1 \cup R_2\}
\end{align*}


\subsection{110010,110100,010011,001011}
\begin{align*}
&p\left(110010 \mid {\bf x}_k,\nu, T\right)\\
&= \mathbb{I}\{k>1; \;\exists \; j_1, j_2\in\mathcal{I} :  \mbox{x}^{(j_1)} \in Q_1,\mbox{x}^{(j_2)} \in P_4;
 \;\forall i\neq j_1,j_2\in\mathcal{I},\;  \mbox{x}^{(i)} \in P_1  \cup  P_4 \cup Q_1 \cup R_1\}\liney
& + \mathbb{I}\{k>2; \;\exists \; j_1, j_2, j_3\in\mathcal{I} :  \mbox{x}^{(j_1)} \in P_1, \mbox{x}^{(j_2)} \in R_1, \mbox{x}^{(j_3)} \in P_4;
\;\forall i\neq j_1,j_2,j_3\in\mathcal{I},\;  \mbox{x}^{(i)} \in P_1 \cup P_4 \cup R_1\}\line
&p\left(010011 \mid {\bf x}_k,\nu, T\right)\\
&= \mathbb{I}\{k>1; \;\exists \; j_1, j_2\in\mathcal{I} :  \mbox{x}^{(j_1)} \in Q_2, \mbox{x}^{(j_2)} \in P_1;
 \;\forall i\neq j_1,j_2\in\mathcal{I},\;  \mbox{x}^{(i)} \in P_1 \cup P_4 \cup Q_2 \cup R_2\}\liney
& + \mathbb{I}\{k>2; \;\exists \; j_1, j_2 , j_3\in\mathcal{I} :  \mbox{x}^{(j_1)} \in P_4, \mbox{x}^{(j_2)} \in R_1,  \mbox{x}^{(j_3)} \in P_1;
  \;\forall i\neq j_1,j_2,j_3\in\mathcal{I},\;  \mbox{x}^{(i)} \in P_1 \cup P_4 \cup R_2\}\line
&p\left(010011 \mid {\bf x}_k,\nu, T\right)\\
&= \mathbb{I}\{k>1; \;\exists \; j_1, j_2\in\mathcal{I} :  \mbox{x}^{(j_1)} \in Q_2, \mbox{x}^{(j_2)} \in P_1;
 \;\forall i\neq j_1,j_2\in\mathcal{I},\;  \mbox{x}^{(i)} \in P_1 \cup P_4 \cup Q_2 \cup R_2\}\liney
& + \mathbb{I}\{k>2; \;\exists \; j_1, j_2 , j_3\in\mathcal{I} :  \mbox{x}^{(j_1)} \in P_1, \mbox{x}^{(j_2)} \in P_4,  \mbox{x}^{(j_3)} \in R_2;
 \;\forall i\neq j_1,j_2,j_3\in\mathcal{I},\;  \mbox{x}^{(i)} \in P_1 \cup P_4 \cup R_2\}\line
&p\left(001011 \mid {\bf x}_k,\nu, T\right)\\
&= \mathbb{I}\{k>1; \;\exists \; j_1, j_2\in\mathcal{I} :  \mbox{x}^{(j_1)} \in Q_2, \mbox{x}^{(j_2)} \in P_2;
 \;\forall i\neq j_1,j_2\in\mathcal{I},\;  \mbox{x}^{(i)} \in P_2 \cup P_4 \cup Q_2 \cup R_2\}\liney
& + \mathbb{I}\{k>2; \;\exists \; j_1, j_2 , j_3\in\mathcal{I} :  \mbox{x}^{(j_1)} \in P_2, \mbox{x}^{(j_2)} \in P_4,  \mbox{x}^{(j_3)} \in R_2;
\;\forall i \neq j_1,j_2,j_3\in\mathcal{I},\;  \mbox{x}^{(i)} \in P_2 \cup P_4 \cup R_2\}
\end{align*}


\subsection{011000,001100,000110}
\begin{align*}
&p\left(011000 \mid {\bf x}_k,\nu, T\right)\\
&= \mathbb{I}\{k>0; \;\exists \; j_1\in\mathcal{I} :  \mbox{x}^{(j_1)} \in D_1;
  \;\forall i\neq j_1\in\mathcal{I},\;  \mbox{x}^{(i)} \in D_1 \cup P_1 \cup P_2\}\liney
& + \mathbb{I}\{k>1; \;\exists \; j_1, j_2\in\mathcal{I} :  \mbox{x}^{(j_1)} \in P_1 , \mbox{x}^{(j_2)} \in P_2 ;
\;\forall i\neq j_1,j_2\in\mathcal{I},\;  \mbox{x}^{(i)} \in P_1 \cup P_2\}\line
&p\left(001100 \mid {\bf x}_k,\nu, T\right)\\
&= \mathbb{I}\{k>0; \;\exists \; j_1\in\mathcal{I} :  \mbox{x}^{(j_1)} \in D_2;
\;\forall i\neq j_1\in\mathcal{I},\;  \mbox{x}^{(i)} \in D_2 \cup P_2 \cup P_3\}\liney
& + \mathbb{I}\{k>1; \;\exists \; j_1, j_2\in\mathcal{I} :  \mbox{x}^{(j_1)} \in P_2 , \mbox{x}^{(j_2)} \in P_3 ;
 \;\forall i\neq j_1,j_2\in\mathcal{I},\;  \mbox{x}^{(i)} \in P_2 \cup P_3\}\line
&p\left(000110 \mid {\bf x}_k,\nu, T\right)\\
&= \mathbb{I}\{k>0; \;\exists \; j_1\in\mathcal{I} :  \mbox{x}^{(j_1)} \in D_3;
 \;\forall i\neq j_1\in\mathcal{I},\;  \mbox{x}^{(i)} \in D_3 \cup P_3 \cup P_4\}\liney
& + \mathbb{I}\{k>1; \;\exists \; j_1, j_2\in\mathcal{I} :  \mbox{x}^{(j_1)} \in P_3 , \mbox{x}^{(j_2)} \in P_4 ;
 \;\forall i\neq j_1,j_2\in\mathcal{I},\;  \mbox{x}^{(i)} \in P_3 \cup P_4\}
\end{align*}


\subsection{101100,100110,001101,011001}
\begin{align*}
&p\left(101100 \mid {\bf x}_k,\nu, T\right)\\
&= \mathbb{I}\{k>1; \;\exists \; j_1,j_2\in\mathcal{I} :  \mbox{x}^{(j_1)} \in D_2, \mbox{x}^{(j_2)} \in R_1;
\;\forall i\neq j_1,j_2\in\mathcal{I},\;  \mbox{x}^{(i)} \in D_2 \cup P_2 \cup P_3 \cup R_1\}\liney
& + \mathbb{I}\{k>2; \;\exists \; j_1, j_2 , j_3\in\mathcal{I} :  \mbox{x}^{(j_1)} \in P_2 , \mbox{x}^{(j_2)} \in P_3,\mbox{x}^{(j_3)} \in R_1 ;
  \;\forall i\neq j_1,j_2,j_3\in\mathcal{I},\;  \mbox{x}^{(i)} \in P_2 \cup P_3 \cup R_1\}\line
&p\left(100110 \mid {\bf x}_k,\nu, T\right)\\
&= \mathbb{I}\{k>1; \;\exists \; j_1, j_2\in\mathcal{I} :  \mbox{x}^{(j_1)} \in D_3, \mbox{x}^{(j_2)} \in R_1 ;
  \;\forall i\neq j_1,j_2\in\mathcal{I},\;  \mbox{x}^{(i)} \in D_3 \cup P_3 \cup P_4 \cup R_1\}\liney
& + \mathbb{I}\{k>2; \;\exists \; j_1, j_2 , j_3\in\mathcal{I} :  \mbox{x}^{(j_1)} \in P_3 , \mbox{x}^{(j_2)} \in P_4, \mbox{x}^{(j_3)} \in R_1 ;
\;\forall i\neq j_1,j_2,j_3\in\mathcal{I},\;  \mbox{x}^{(i)} \in P_3 \cup P_4 \cup R_1\}\line
&p\left(001101 \mid {\bf x}_k,\nu, T\right)\\
&= \mathbb{I}\{k>1; \;\exists \; j_1,j_2\in\mathcal{I} :  \mbox{x}^{(j_1)} \in D_2, \mbox{x}^{(j_2)} \in R_2;
 \;\forall i\neq j_1,j_2\in\mathcal{I},\;  \mbox{x}^{(i)} \in D_2 \cup P_2 \cup P_3 \cup R_2\}\liney
& + \mathbb{I}\{k>2; \;\exists \; j_1, j_2 , j_3\in\mathcal{I} :  \mbox{x}^{(j_1)} \in P_2 , \mbox{x}^{(j_2)} \in P_3,\mbox{x}^{(j_3)} \in R_2 ;
 \;\forall i\neq j_1,j_2,j_3\in\mathcal{I},\;  \mbox{x}^{(i)} \in P_2 \cup P_3 \cup R_2\}\line
&p\left(011001 \mid {\bf x}_k,\nu, T\right)\\
&= \mathbb{I}\{k>1; \;\exists \; j_1,j_2\in\mathcal{I} :  \mbox{x}^{(j_1)} \in D_1, \mbox{x}^{(j_2)} \in R_2;
 \;\forall i\neq j_1,j_2\in\mathcal{I},\;  \mbox{x}^{(i)} \in D_1 \cup P_1 \cup P_2 \cup R_2\}\liney
& + \mathbb{I}\{k>2; \;\exists \; j_1, j_2,j_3\in\mathcal{I} :  \mbox{x}^{(j_1)} \in P_1 , \mbox{x}^{(j_2)} \in P_2 , \mbox{x}^{(j_3)} \in R_2 ;
\;\forall i\neq j_1,j_2,j_3\in\mathcal{I},\;  \mbox{x}^{(i)} \in P_1 \cup P_2 \cup R_2\}
\end{align*}


\subsection{011010,010110}
\begin{align*}
&p\left(011010 \mid {\bf x}_k,\nu, T\right)\\
&= \mathbb{I}\{k>1; \;\exists \; j_1,j_2\in\mathcal{I} :  \mbox{x}^{(j_1)} \in D_1, \mbox{x}^{(j_2)} \in P_4;
\;\forall i\neq j_1,j_2\in\mathcal{I},\;  \mbox{x}^{(i)} \in D_1 \cup P_1 \cup P_2 \cup P_4\}\liney
& + \mathbb{I}\{k>2; \;\exists \; j_1, j_2,j_3\in\mathcal{I} :  \mbox{x}^{(j_1)} \in P_1 , \mbox{x}^{(j_2)} \in P_2, \mbox{x}^{(j_3)} \in P_4 ;
  \;\forall i\neq j_1,j_2,j_3\in\mathcal{I},\;  \mbox{x}^{(i)} \in P_1 \cup P_2 \cup P_4\}\line
&p\left(010110 \mid {\bf x}_k,\nu, T\right)\\
&= \mathbb{I}\{k>1; \;\exists \; j_1,j_2\in\mathcal{I} :  \mbox{x}^{(j_1)} \in D_3,\mbox{x}^{(j_2)} \in P_1;
  \;\forall i\neq j_1,j_2\in\mathcal{I},\;  \mbox{x}^{(i)} \in D_3 \cup P_1 \cup P_3 \cup P_4\}\liney
& + \mathbb{I}\{k>2; \;\exists \; j_1, j_2,j_3\in\mathcal{I} :  \mbox{x}^{(j_1)} \in P_3 , \mbox{x}^{(j_2)} \in P_1, \mbox{x}^{(j_3)} \in P_4 ;
  \;\forall i\neq j_1,j_2,j_3\in\mathcal{I},\;  \mbox{x}^{(i)} \in P_1 \cup P_3 \cup P_4\}
\end{align*}


\subsection{011011,110110}
\begin{align*}
&p\left(011011 \mid {\bf x}_k,\nu, T\right)\\
&= \mathbb{I}\{k>1; \;\exists \; j_1,j_2\in\mathcal{I} :  \mbox{x}^{(j_1)} \in D_1, \mbox{x}^{(j_2)} \in Q_2;\\
& \hspace{1cm}  \;\forall i\neq j_1,j_2\in\mathcal{I},\;  \mbox{x}^{(i)} \in D_1 \cup P_1 \cup P_2 \cup P_4 \cup Q_2 \cup R_2\}\liney
& + \mathbb{I}\{k>2; \;\exists \; j_1,j_2,j_3\in\mathcal{I} :  \mbox{x}^{(j_1)} \in D_1, \mbox{x}^{(j_2)} \in P_4, \mbox{x}^{(j_3)} \in R_2;\\
& \hspace{1cm}  \;\forall i\neq j_1,j_2,j_3\in\mathcal{I},\;  \mbox{x}^{(i)} \in D_1 \cup P_1 \cup P_2 \cup P_4 \cup R_2\}\liney
& + \mathbb{I}\{k>2; \;\exists \; j_1,j_2,j_3\in\mathcal{I} :  \mbox{x}^{(j_1)} \in P_1, \mbox{x}^{(j_2)} \in P_2, \mbox{x}^{(j_3)} \in Q_2;\\
& \hspace{1cm}  \;\forall i\neq j_1,j_2,j_3\in\mathcal{I},\;  \mbox{x}^{(i)} \in  P_1 \cup P_2 \cup P_4 \cup Q_2 \cup R_2\}\liney
& + \mathbb{I}\{k>3; \;\exists \; j_1,j_2,j_3,j_4\in\mathcal{I} :  \mbox{x}^{(j_1)} \in P_1, \mbox{x}^{(j_2)} \in P_2, \mbox{x}^{(j_3)} \in P_4, \mbox{x}^{(j_4)} \in R_2;\\
& \hspace{1cm}  \;\forall i\neq j_1,j_2,j_3,j_4\in\mathcal{I},\;  \mbox{x}^{(i)} \in P_1 \cup P_2 \cup P_4 \cup R_2\}\line
&p\left(110110 \mid {\bf x}_k,\nu, T\right)\\
&= \mathbb{I}\{k>1; \;\exists \; j_1,j_2\in\mathcal{I} :  \mbox{x}^{(j_1)} \in D_3, \mbox{x}^{(j_2)} \in Q_1;\\
& \hspace{1cm}  \;\forall i\neq j_1,j_2\in\mathcal{I},\;  \mbox{x}^{(i)} \in D3 \cup P_1 \cup P_3 \cup P_4 \cup Q_1 \cup R_1\}\liney
& + \mathbb{I}\{k>2; \;\exists \; j_1,j_2,j_3\in\mathcal{I} :  \mbox{x}^{(j_1)} \in D_3, \mbox{x}^{(j_2)} \in P_1, \mbox{x}^{(j_3)} \in R_1;\\
& \hspace{1cm}  \;\forall i\neq j_1,j_2,j_3\in\mathcal{I},\;  \mbox{x}^{(i)} \in D_3 \cup P_1 \cup P_3 \cup P_4 \cup R_1\}\liney
& + \mathbb{I}\{k>2; \;\exists \; j_1,j_2,j_3\in\mathcal{I} :  \mbox{x}^{(j_1)} \in P_3, \mbox{x}^{(j_2)} \in P_4, \mbox{x}^{(j_3)} \in Q_1;\\
& \hspace{1cm}  \;\forall i\neq j_1,j_2,j_3\in\mathcal{I},\;  \mbox{x}^{(i)} \in  P_1 \cup P_3 \cup P_4 \cup Q_1 \cup R_1\}\liney
& + \mathbb{I}\{k>3; \;\exists \; j_1,j_2,j_3,j_4\in\mathcal{I} :  \mbox{x}^{(j_1)} \in P1, \mbox{x}^{(j_2)} \in P_3, \mbox{x}^{(j_3)} \in P_4, \mbox{x}^{(j_4)} \in R_1;\\
& \hspace{1cm}  \;\forall i\neq j_1,j_2,j_3,j_4\in\mathcal{I},\;  \mbox{x}^{(i)} \in P_1 \cup P_3 \cup P_4 \cup R_1\}
\end{align*}


\subsection{110101,101011}
\begin{align*}
&p\left(110101 \mid {\bf x}_k,\nu, T\right)\\
&= \mathbb{I}\{k>2; \;\exists \; j_1,j_2,j_3\in\mathcal{I} :  \mbox{x}^{(j_1)} \in P_3,\mbox{x}^{(j_2)} \in Q_1, \mbox{x}^{(j_3)} \in R_2;\\
& \hspace{1cm}  \;\forall i\neq j_1,j_2,j_3\in\mathcal{I},\;  \mbox{x}^{(i)} \in P_1 \cup P_3 \cup Q_1 \cup R_1 \cup R_2\}\liney
& + \mathbb{I}\{k>3; \;\exists \; j_1, j_2,j_3,j_4\in\mathcal{I} :  \mbox{x}^{(j_1)} \in P_1, \mbox{x}^{(j_2)} \in P_3, \mbox{x}^{(j_3)} \in R_1, \mbox{x}^{(j_4)} \in R_2;\\
& \hspace{1cm}  \;\forall i\neq j_1,j_2,j_3,j_4\in\mathcal{I},\;  \mbox{x}^{(i)} \in P_1 \cup P_3 \cup R_1 \cup R_2\}\line
&p\left(101011 \mid {\bf x}_k,\nu, T\right)\\
&= \mathbb{I}\{k>2; \;\exists \; j_1,j_2,j_3\in\mathcal{I} :  \mbox{x}^{(j_1)} \in P_2,\mbox{x}^{(j_1)} \in Q_2,\mbox{x}^{(j_3)} \in R_1;\\
& \hspace{1cm}  \;\forall i\neq j_1,j_2,j_3\in\mathcal{I},\;  \mbox{x}^{(i)} \in P_2 \cup P_4 \cup Q_2 \cup R_1 \cup  R_2\}\liney
& + \mathbb{I}\{k>3; \;\exists \; j_1, j_2,j_3,j_4\in\mathcal{I} :  \mbox{x}^{(j_1)} \in P_2, \mbox{x}^{(j_2)} \in P_4, \mbox{x}^{(j_3)} \in R_1, \mbox{x}^{(j_4)} \in R_2;\\
& \hspace{1cm}  \;\forall i\neq j_1,j_2,j_3,j_4\in\mathcal{I},\;  \mbox{x}^{(i)} \in P_2 \cup P_4 \cup R_1\cup R_2\}
\end{align*}


\subsection{101101}
\begin{align*}
&p\left(101101 \mid {\bf x}_k,\nu, T\right)\\
&= \mathbb{I}\{k>2; \;\exists \; j_1,j_2,j_3\in\mathcal{I} :  \mbox{x}^{(j_1)} \in D_2,\mbox{x}^{(j_2)} \in R_1,\mbox{x}^{(j_3)} \in R_2;\\
& \hspace{1cm}  \;\forall i\neq j_1,j_2,j_3\in\mathcal{I},\;  \mbox{x}^{(i)} \in D_2 \cup P_2 \cup P_3 \cup R_1 \cup R_3\}\liney
& + \mathbb{I}\{k>3; \;\exists \; j_1, j_2,j_3,j_4\in\mathcal{I} :  \mbox{x}^{(j_1)} \in P_2 , \mbox{x}^{(j_2)} \in P_3, \mbox{x}^{(j_3)} \in R_1, \mbox{x}^{(j_4)} \in R_2 ;\\
& \hspace{1cm}  \;\forall i\neq j_1,j_2,j_3,j_4\in\mathcal{I},\;  \mbox{x}^{(i)} \in P_2 \cup P_3 \cup R_1 \cup R_2\}\line
\end{align*}


\subsection{011100,001110}
\begin{align*}
&p\left(011100 \mid {\bf x}_k,\nu, T\right)\\
&= \mathbb{I}\{k>0; \;\exists \; j_1\in\mathcal{I} :  \mbox{x}^{(j_1)} \in S_1;\\
& \hspace{1cm}  \;\forall i\neq j_1\in\mathcal{I},\;  \mbox{x}^{(i)} \in D_1 \cup D_2 \cup P_1 \cup P_2 \cup P_3 \cup S_1 \}\liney
& + \mathbb{I}\{k>1; \;\exists \; j_1,j_2\in\mathcal{I} :  \mbox{x}^{(j_1)} \in D_1,\mbox{x}^{(j_2)} \in D_2;\\
& \hspace{1cm}  \;\forall i\neq j_1,j_2\in\mathcal{I},\;  \mbox{x}^{(i)} \in D_1 \cup D_2 \cup P_1 \cup P_2 \cup P_3  \}\liney
& + \mathbb{I}\{k>1; \;\exists \; j_1,j_2\in\mathcal{I} :  \mbox{x}^{(j_1)} \in D_1,\mbox{x}^{(j_2)} \in P_3;\\
& \hspace{1cm}  \;\forall i\neq j_1,j_2\in\mathcal{I},\;  \mbox{x}^{(i)} \in D_1 \cup P_1 \cup P_2 \cup P_3 \}\liney
& + \mathbb{I}\{k>1; \;\exists \; j_1,j_2\in\mathcal{I} :  \mbox{x}^{(j_1)} \in D_2, \mbox{x}^{(j_2)} \in P_1;\\
& \hspace{1cm}  \;\forall i\neq j_1,j_2\in\mathcal{I},\;  \mbox{x}^{(i)} \in  D_2 \cup P_1 \cup P_2 \cup P_3 \}\liney
& + \mathbb{I}\{k>2; \;\exists \; j_1,j_2,j_3\in\mathcal{I} :  \mbox{x}^{(j_1)} \in P_1,\mbox{x}^{(j_2)} \in P_2,\mbox{x}^{(j_3)} \in P_3;\\
& \hspace{1cm}  \;\forall i\neq j_1,j_2,j_3\in\mathcal{I},\;  \mbox{x}^{(i)} \in  P_1 \cup P_2 \cup P_3\}\line
&p\left(001110 \mid {\bf x}_k,\nu, T\right)\\
&= \mathbb{I}\{k>0; \;\exists \; j_1\in\mathcal{I} :  \mbox{x}^{(j_1)} \in S_2;\\
& \hspace{1cm}  \;\forall i\neq j_1\in\mathcal{I},\;  \mbox{x}^{(i)} \in D_2 \cup D_3 \cup P_2 \cup P_3 \cup P_4 \cup S_2 \}\liney
& + \mathbb{I}\{k>1; \;\exists \; j_1,j_2\in\mathcal{I} :  \mbox{x}^{(j_1)} \in D_2,\mbox{x}^{(j_2)} \in D_3;\\
& \hspace{1cm}  \;\forall i\neq j_1,j_2\in\mathcal{I},\;  \mbox{x}^{(i)} \in D_2 \cup D_3 \cup P_3 \cup P_3 \cup P_4  \}\liney
& + \mathbb{I}\{k>1; \;\exists \; j_1,j_2\in\mathcal{I} :  \mbox{x}^{(j_1)} \in D_2,\mbox{x}^{(j_2)} \in P_4;\\
& \hspace{1cm}  \;\forall i\neq j_1,j_2\in\mathcal{I},\;  \mbox{x}^{(i)} \in D_2 \cup P_2 \cup P_3 \cup P_4 \}\liney
& + \mathbb{I}\{k>1; \;\exists \; j_1,j_2\in\mathcal{I} :  \mbox{x}^{(j_1)} \in D_3, \mbox{x}^{(j_2)} \in P_2;\\
& \hspace{1cm}  \;\forall i\neq j_1,j_2\in\mathcal{I},\;  \mbox{x}^{(i)} \in  D_3 \cup P_2 \cup P_3 \cup P_4 \}\liney
& + \mathbb{I}\{k>2; \;\exists \; j_1,j_2,j_3\in\mathcal{I} :  \mbox{x}^{(j_1)} \in P_2,\mbox{x}^{(j_2)} \in P_3,\mbox{x}^{(j_3)} \in P_4;\\
& \hspace{1cm}  \;\forall i\neq j_1,j_2,j_3\in\mathcal{I},\;  \mbox{x}^{(i)} \in  P_2 \cup P_3 \cup P_4\}\line
\end{align*}


\subsection{011101,101110}
\begin{align*}
&p\left(011101 \mid {\bf x}_k,\nu, T\right)\\
&= \mathbb{I}\{k>1; \;\exists \; j_1,j_2\in\mathcal{I} :  \mbox{x}^{(j_1)} \in S_1, \mbox{x}^{(j_2)} \in R_2;\\
& \hspace{1cm}  \;\forall i\neq j_1,j_2\in\mathcal{I},\;  \mbox{x}^{(i)} \in D_1 \cup D_2 \cup P_1 \cup P_2 \cup P_3 \cup R_2 \cup S_1 \}\liney
& + \mathbb{I}\{k>2; \;\exists \; j_1,j_2,j_3\in\mathcal{I} :  \mbox{x}^{(j_1)} \in D_1,\mbox{x}^{(j_2)} \in D_2,\mbox{x}^{(j_3)} \in R_2;\\
& \hspace{1cm}  \;\forall i\neq j_1,j_2,j_3\in\mathcal{I},\;  \mbox{x}^{(i)} \in D_1 \cup D_2 \cup P_1 \cup P_2 \cup P_3 \cup R_2  \}\liney
& + \mathbb{I}\{k>2; \;\exists \; j_1,j_2,j_3\in\mathcal{I} :  \mbox{x}^{(j_1)} \in D_1,\mbox{x}^{(j_2)} \in P_3,\mbox{x}^{(j_3)} \in R_2;\\
& \hspace{1cm}  \;\forall i\neq j_1,j_2,j_3\in\mathcal{I},\;  \mbox{x}^{(i)} \in D_1 \cup P_1 \cup P_2 \cup P_3 \cup R_2 \}\liney
& + \mathbb{I}\{k>2; \;\exists \; j_1,j_2,j_3\in\mathcal{I} :  \mbox{x}^{(j_1)} \in D_2, \mbox{x}^{(j_2)} \in P_1, \mbox{x}^{(j_3)} \in R_2;\\
& \hspace{1cm}  \;\forall i\neq j_1,j_2,j_3\in\mathcal{I},\;  \mbox{x}^{(i)} \in  D_2 \cup P_1 \cup P_2 \cup P_3 \cup R_2 \}\liney
& + \mathbb{I}\{k>3; \;\exists \; j_1,j_2,j_3,j_4\in\mathcal{I} :  \mbox{x}^{(j_1)} \in P_1,\mbox{x}^{(j_2)} \in P_2,\mbox{x}^{(j_3)} \in P_3,\mbox{x}^{(j_4)} \in R_2;\\
& \hspace{1cm}  \;\forall i\neq j_1,j_2,j_3,j_4\in\mathcal{I},\;  \mbox{x}^{(i)} \in  P_1 \cup P_2 \cup P_3 \cup R_2\}\line
&p\left(101110 \mid {\bf x}_k,\nu, T\right)\\
&= \mathbb{I}\{k>1; \;\exists \; j_1,j_2\in\mathcal{I} :  \mbox{x}^{(j_1)} \in S_2,\mbox{x}^{(j_2)} \in R_1;\\
& \hspace{1cm}  \;\forall i\neq j_1,j_2\in\mathcal{I},\;  \mbox{x}^{(i)} \in D_2 \cup D_3 \cup P_2 \cup P_3 \cup P_4 \cup R_1 \cup S_2 \}\liney
& + \mathbb{I}\{k>2; \;\exists \; j_1,j_2,j_3\in\mathcal{I} :  \mbox{x}^{(j_1)} \in D_2,\mbox{x}^{(j_2)} \in D_3,\mbox{x}^{(j_3)} \in R_1;\\
& \hspace{1cm}  \;\forall i\neq j_1,j_2,j_3\in\mathcal{I},\;  \mbox{x}^{(i)} \in D_2 \cup D_3 \cup P_3 \cup P_3 \cup P_4 \cup R_1  \}\liney
& + \mathbb{I}\{k>2; \;\exists \; j_1,j_2,j_3\in\mathcal{I} :  \mbox{x}^{(j_1)} \in D_2,\mbox{x}^{(j_2)} \in P_4,\mbox{x}^{(j_3)} \in R_1;\\
& \hspace{1cm}  \;\forall i\neq j_1,j_2,j_3\in\mathcal{I},\;  \mbox{x}^{(i)} \in D_2 \cup P_2 \cup P_3 \cup P_4 \cup R_1\}\liney
& + \mathbb{I}\{k>2; \;\exists \; j_1,j_2,j_3\in\mathcal{I} :  \mbox{x}^{(j_1)} \in D_3, \mbox{x}^{(j_2)} \in P_2, \mbox{x}^{(j_3)} \in R_1;\\
& \hspace{1cm}  \;\forall i\neq j_1,j_2,j_3\in\mathcal{I},\;  \mbox{x}^{(i)} \in  D_3 \cup P_2 \cup P_3 \cup P_4 \cup R_1 \}\liney
& + \mathbb{I}\{k>3; \;\exists \; j_1,j_2,j_3,j_4\in\mathcal{I} :  \mbox{x}^{(j_1)} \in P_2,\mbox{x}^{(j_2)} \in P_3,\mbox{x}^{(j_3)} \in P_4,\mbox{x}^{(j_4)} \in R_1;\\
& \hspace{1cm}  \;\forall i\neq j_1,j_2,j_3,j_4\in\mathcal{I},\;  \mbox{x}^{(i)} \in  P_2 \cup P_3 \cup P_4 \cup  R_1\}\line
\end{align*}


\subsection{111000,000111}
\begin{align*}
&p\left(111000 \mid {\bf x}_k,\nu, T\right)\\
&= \mathbb{I}\{k>0; \;\exists \; j_1\in\mathcal{I} :  \mbox{x}^{(j_1)} \in M_1;\\
& \hspace{1cm}  \;\forall i\neq j_1\in\mathcal{I},\;  \mbox{x}^{(i)} \in D_1 \cup M_1 \cup P_1 \cup P_2  \cup Q_1 \cup R_1 \}\liney
& + \mathbb{I}\{k>1; \;\exists \; j_1,j_2\in\mathcal{I} :  \mbox{x}^{(j_1)} \in D_1,\mbox{x}^{(j_2)} \in Q_1;\\
& \hspace{1cm}  \;\forall i\neq j_1,j_2\in\mathcal{I},\;  \mbox{x}^{(i)} \in D_1 \cup P_1 \cup P_2  \cup Q_1 \cup R_1 \}\liney  
& + \mathbb{I}\{k>1; \;\exists \; j_1,j_2\in\mathcal{I} :  \mbox{x}^{(j_1)} \in P_2, \mbox{x}^{(j_2)} \in Q_1;\\
& \hspace{1cm}  \;\forall i\neq j_1,j_2\in\mathcal{I},\;  \mbox{x}^{(i)} \in P_1 \cup P_2  \cup Q_1 \cup R_1 \}\liney
& + \mathbb{I}\{k>1; \;\exists \; j_1,j_2\in\mathcal{I} :  \mbox{x}^{(j_1)} \in D_1,\mbox{x}^{(j_2)} \in R_1;\\
& \hspace{1cm}  \;\forall i\neq j_1,j_2\in\mathcal{I},\;  \mbox{x}^{(i)} \in P_1 \cup P_2  \cup R_1 \}\liney
& + \mathbb{I}\{k>2; \;\exists \; j_1,j_2,j_3\in\mathcal{I} :  \mbox{x}^{(j_1)} \in P_1,\mbox{x}^{(j_2)} \in P_2,\mbox{x}^{(j_3)} \in R_1;\\
& \hspace{1cm}  \;\forall i\neq j_1,j_2,j_3\in\mathcal{I},\;  \mbox{x}^{(i)} \in P_1 \cup P_2  \cup R_1 \}\line
&p\left(000111 \mid {\bf x}_k,\nu, T\right)\\
&= \mathbb{I}\{k>0; \;\exists \; j_1\in\mathcal{I} :  \mbox{x}^{(j_1)} \in M_2;\\
& \hspace{1cm}  \;\forall i\neq j_1\in\mathcal{I},\;  \mbox{x}^{(i)} \in D_3 \cup M_2 \cup P_3 \cup P_4  \cup Q_2 \cup R_2 \}\liney
& + \mathbb{I}\{k>1; \;\exists \; j_1,j_2\in\mathcal{I} :  \mbox{x}^{(j_1)} \in D_3,\mbox{x}^{(j_2)} \in Q_2;\\
& \hspace{1cm}  \;\forall i\neq j_1,j_2\in\mathcal{I},\;  \mbox{x}^{(i)} \in D_3 \cup P_3 \cup P_4  \cup Q_2 \cup R_2 \}\liney  
& + \mathbb{I}\{k>1; \;\exists \; j_1,j_2\in\mathcal{I} :  \mbox{x}^{(j_1)} \in P_3, \mbox{x}^{(j_2)} \in Q_2;\\
& \hspace{1cm}  \;\forall i\neq j_1,j_2\in\mathcal{I},\;  \mbox{x}^{(i)} \in P_3 \cup P_4  \cup Q_2 \cup R_2 \}\liney
& + \mathbb{I}\{k>1; \;\exists \; j_1,j_2\in\mathcal{I} :  \mbox{x}^{(j_1)} \in D_3,\mbox{x}^{(j_2)} \in R_2;\\
& \hspace{1cm}  \;\forall i\neq j_1,j_2\in\mathcal{I},\;  \mbox{x}^{(i)} \in P_3 \cup P_4  \cup R_2 \}\liney
& + \mathbb{I}\{k>2; \;\exists \; j_1,j_2,j_3\in\mathcal{I} :  \mbox{x}^{(j_1)} \in P_3,\mbox{x}^{(j_2)} \in P_4,\mbox{x}^{(j_3)} \in R_2;\\
& \hspace{1cm}  \;\forall i\neq j_1,j_2,j_3\in\mathcal{I},\;  \mbox{x}^{(i)} \in P_3 \cup P_4  \cup R_2 \}
\end{align*}


\subsection{111010,010111}
\begin{align*}
&p\left(111010 \mid {\bf x}_k,\nu, T\right)\\
&= \mathbb{I}\{k>1; \;\exists \; j_1,j_2\in\mathcal{I} :  \mbox{x}^{(j_1)} \in M_1,\mbox{x}^{(j_2)} \in P_4;\\
& \hspace{1cm}  \;\forall i\neq j_1,j_2\in\mathcal{I},\;  \mbox{x}^{(i)} \in D_1 \cup M_1 \cup P_1 \cup P_2  \cup P_4 \cup Q_1 \cup R_1 \}\liney
& + \mathbb{I}\{k>2; \;\exists \; j_1,j_2,j_3\in\mathcal{I} :  \mbox{x}^{(j_1)} \in D_1,\mbox{x}^{(j_2)} \in P_4,\mbox{x}^{(j_3)} \in Q_1;\\
& \hspace{1cm}  \;\forall i\neq j_1,j_2,j_3\in\mathcal{I},\;  \mbox{x}^{(i)} \in D_1 \cup P_1 \cup P_2  \cup P_4 \cup Q_1 \cup R_1 \}\liney  
& + \mathbb{I}\{k>2; \;\exists \; j_1,j_2,j_3\in\mathcal{I} :  \mbox{x}^{(j_1)} \in P_2, \mbox{x}^{(j_2)} \in P_4, \mbox{x}^{(j_3)} \in Q_1;\\
& \hspace{1cm}  \;\forall i\neq j_1,j_2,j_3\in\mathcal{I},\;  \mbox{x}^{(i)} \in P_1 \cup P_2  \cup P_4 \cup Q_1 \cup R_1 \}\liney
& + \mathbb{I}\{k>2; \;\exists \; j_1,j_2,j_3\in\mathcal{I} :  \mbox{x}^{(j_1)} \in D_1,\mbox{x}^{(j_2)} \in P_4,\mbox{x}^{(j_3)} \in R_1;\\
& \hspace{1cm}  \;\forall i\neq j_1,j_2,j_3\in\mathcal{I},\;  \mbox{x}^{(i)} \in P_1 \cup P_2  \cup P_4 \cup R_1 \}\liney
& + \mathbb{I}\{k>3; \;\exists \; j_1,j_2,j_3,j_4\in\mathcal{I} :  \mbox{x}^{(j_1)} \in P_1,\mbox{x}^{(j_2)} \in P_2,\mbox{x}^{(j_3)} \in P_4,\mbox{x}^{(j_4)} \in R_1;\\
& \hspace{1cm}  \;\forall i\neq j_1,j_2,j_3,j_4\in\mathcal{I},\;  \mbox{x}^{(i)} \in P_1 \cup P_2 \cup P_4 \cup R_1 \}\line
&p\left(010111 \mid {\bf x}_k,\nu, T\right)\\
&= \mathbb{I}\{k>1; \;\exists \; j_1,j_2\in\mathcal{I} :  \mbox{x}^{(j_1)} \in M_2, \mbox{x}^{(j_2)} \in P_1;\\
& \hspace{1cm}  \;\forall i\neq j_1,j_2\in\mathcal{I},\;  \mbox{x}^{(i)} \in D_3 \cup M_2 \cup P_1 \cup P_3 \cup P_4  \cup Q_2 \cup R_2 \}\liney
& + \mathbb{I}\{k>2; \;\exists \; j_1,j_2,j_3\in\mathcal{I} :  \mbox{x}^{(j_1)} \in D_3,\mbox{x}^{(j_2)} \in P_1,\mbox{x}^{(j_3)} \in Q_2;\\
& \hspace{1cm}  \;\forall i\neq j_1,j_2,j_3\in\mathcal{I},\;  \mbox{x}^{(i)} \in D_3 \cup P_1\cup P_3 \cup P_4  \cup Q_2 \cup R_2 \}\liney  
& + \mathbb{I}\{k>2; \;\exists \; j_1,j_2,j_3\in\mathcal{I} :  \mbox{x}^{(j_1)} \in P_3, \mbox{x}^{(j_2)} \in P_1, \mbox{x}^{(j_3)} \in Q_2;\\
& \hspace{1cm}  \;\forall i\neq j_1,j_2,j_3\in\mathcal{I},\;  \mbox{x}^{(i)} \in P_1 \cup P_3 \cup P_4  \cup Q_2 \cup R_2 \}\liney
& + \mathbb{I}\{k>2; \;\exists \; j_1,j_2,j_3\in\mathcal{I} :  \mbox{x}^{(j_1)} \in D_3,\mbox{x}^{(j_2)} \in P_1,\mbox{x}^{(j_3)} \in R_2;\\
& \hspace{1cm}  \;\forall i\neq j_1,j_2,j_3\in\mathcal{I},\;  \mbox{x}^{(i)} \in P_1\cup P_3 \cup P_4  \cup R_2 \}\liney
& + \mathbb{I}\{k>3; \;\exists \; j_1,j_2,j_3,j_4\in\mathcal{I} :  \mbox{x}^{(j_1)} \in P_1,\mbox{x}^{(j_2)} \in P_3,\mbox{x}^{(j_3)} \in P_4,\mbox{x}^{(j_4)} \in R_2;\\
& \hspace{1cm}  \;\forall i\neq j_1,j_2,j_3,j_4\in\mathcal{I},\;  \mbox{x}^{(i)} \in P_1 \cup P_3 \cup P_4  \cup R_2 \}
\end{align*}

\subsection{ 111001,100111}
\begin{align*}
&p\left(111001 \mid {\bf x}_k,\nu, T\right)\\
&= \mathbb{I}\{k>1; \;\exists \; j_1,j_2\in\mathcal{I} :  \mbox{x}^{(j_1)} \in M_1,\mbox{x}^{(j_2)} \in R_2;\\
& \hspace{1cm}  \;\forall i\neq j_1,j_2\in\mathcal{I},\;  \mbox{x}^{(i)} \in D_1 \cup M_1 \cup P_1 \cup P_2  \cup Q_1 \cup R_1 \cup R_2 \}\liney
& + \mathbb{I}\{k>2; \;\exists \; j_1,j_2,j_3\in\mathcal{I} :  \mbox{x}^{(j_1)} \in D_1,\mbox{x}^{(j_2)} \in Q_1,\mbox{x}^{(j_3)} \in R_2;\\
& \hspace{1cm}  \;\forall i\neq j_1,j_2,j_3\in\mathcal{I},\;  \mbox{x}^{(i)} \in D_1 \cup P_1 \cup P_2  \cup Q_1 \cup R_1  \cup R_2 \}\liney  
& + \mathbb{I}\{k>2; \;\exists \; j_1,j_2,j_3\in\mathcal{I} :  \mbox{x}^{(j_1)} \in P_2, \mbox{x}^{(j_2)} \in Q_1, \mbox{x}^{(j_3)} \in R_2;\\
& \hspace{1cm}  \;\forall i\neq j_1,j_2,j_3\in\mathcal{I},\;  \mbox{x}^{(i)} \in P_1 \cup P_2  \cup Q_1 \cup R_1 \cup R_2 \}\liney
& + \mathbb{I}\{k>2; \;\exists \; j_1,j_2,j_3\in\mathcal{I} :  \mbox{x}^{(j_1)} \in D_1,\mbox{x}^{(j_2)} \in R_1,\mbox{x}^{(j_3)} \in R_2;\\
& \hspace{1cm}  \;\forall i\neq j_1,j_2,j_3\in\mathcal{I},\;  \mbox{x}^{(i)} \in P_1 \cup P_2  \cup R_1 \cup R_2 \}\liney
& + \mathbb{I}\{k>3; \;\exists \; j_1,j_2,j_3,j_4\in\mathcal{I} :  \mbox{x}^{(j_1)} \in P_1,\mbox{x}^{(j_2)} \in P_2,\mbox{x}^{(j_3)} \in R_1,\mbox{x}^{(j_4)} \in R_2;\\
& \hspace{1cm}  \;\forall i\neq j_1,j_2,j_3,j_4\in\mathcal{I},\;  \mbox{x}^{(i)} \in P_1 \cup P_2 \cup R_1 \cup R_2 \}\line
&p\left(100111 \mid {\bf x}_k,\nu, T\right)\\
&= \mathbb{I}\{k>1; \;\exists \; j_1,j_2\in\mathcal{I} :  \mbox{x}^{(j_1)} \in M_2, \mbox{x}^{(j_2)} \in R_1;\\
& \hspace{1cm}  \;\forall i\neq j_1,j_2\in\mathcal{I},\;  \mbox{x}^{(i)} \in D_3 \cup M_2 \cup P_3 \cup P_4  \cup Q_2 \cup R_1 \cup R_2 \}\liney
& + \mathbb{I}\{k>2; \;\exists \; j_1,j_2,j_3\in\mathcal{I} :  \mbox{x}^{(j_1)} \in D_3,\mbox{x}^{(j_2)} \in Q_2,\mbox{x}^{(j_3)} \in R_1;\\
& \hspace{1cm}  \;\forall i\neq j_1,j_2,j_3\in\mathcal{I},\;  \mbox{x}^{(i)} \in D_3\cup P_3 \cup P_4  \cup Q_2 \cup R_1 \cup R_2 \}\liney  
& + \mathbb{I}\{k>2; \;\exists \; j_1,j_2,j_3\in\mathcal{I} :  \mbox{x}^{(j_1)} \in P_3, \mbox{x}^{(j_2)} \in Q_2, \mbox{x}^{(j_3)} \in R_1;\\
& \hspace{1cm}  \;\forall i\neq j_1,j_2,j_3\in\mathcal{I},\;  \mbox{x}^{(i)} \in P_3 \cup P_4  \cup Q_2 \cup R_1 \cup  R_2 \}\liney
& + \mathbb{I}\{k>2; \;\exists \; j_1,j_2,j_3\in\mathcal{I} :  \mbox{x}^{(j_1)} \in D_3,\mbox{x}^{(j_2)} \in R_1,\mbox{x}^{(j_3)} \in R_2;\\
& \hspace{1cm}  \;\forall i\neq j_1,j_2,j_3\in\mathcal{I},\;  \mbox{x}^{(i)} \in P_3 \cup P_4  \cup R_1\cup R_2 \}\liney
& + \mathbb{I}\{k>3; \;\exists \; j_1,j_2,j_3,j_4\in\mathcal{I} :  \mbox{x}^{(j_1)} \in P_3,\mbox{x}^{(j_2)} \in P_4,\mbox{x}^{(j_3)} \in R_1,\mbox{x}^{(j_4)} \in R_2;\\
& \hspace{1cm}  \;\forall i\neq j_1,j_2,j_3,j_4\in\mathcal{I},\;  \mbox{x}^{(i)} \in P_3 \cup P_4  \cup R_1 \cup R_2 \}
\end{align*}


\subsection{ 111011,110111}
\begin{align*}
&p\left(111011 \mid {\bf x}_k,\nu, T\right)\\
&= \mathbb{I}\{k>1; \;\exists \; j_1,j_2\in\mathcal{I} :  \mbox{x}^{(j_1)} \in M_1,\mbox{x}^{(j_2)} \in Q_2;\\
& \hspace{1cm}  \;\forall i\neq j_1,j_2\in\mathcal{I},\;  \mbox{x}^{(i)} \in D_1 \cup M_1 \cup P_1 \cup P_2 \cup P_4   \cup Q_1\cup Q_2  \cup R_1 \cup R_2 \}\liney
& + \mathbb{I}\{k>2; \;\exists \; j_1,j_2,j_3\in\mathcal{I} :  \mbox{x}^{(j_1)} \in M_1,\mbox{x}^{(j_2)} \in P_4,\mbox{x}^{(j_3)} \in R_2;\\
& \hspace{1cm}  \;\forall i\neq j_1,j_2,j_3\in\mathcal{I},\;  \mbox{x}^{(i)} \in D_1 \cup M_1 \cup P_1 \cup P_2 \cup P_4   \cup Q_1  \cup R_1 \cup R_2 \}\liney
& + \mathbb{I}\{k>2; \;\exists \; j_1,j_2,j_3\in\mathcal{I} :  \mbox{x}^{(j_1)} \in D_1,\mbox{x}^{(j_2)} \in Q_1,\mbox{x}^{(j_3)} \in Q_2;\\
& \hspace{1cm}  \;\forall i\neq j_1,j_2,j_3\in\mathcal{I},\;  \mbox{x}^{(i)} \in D_1 \cup P_1 \cup P_2 \cup P_4   \cup Q_1\cup Q_2  \cup R_1 \cup R_2 \}\liney
& + \mathbb{I}\{k>3; \;\exists \; j_1,j_2,j_3,j_4\in\mathcal{I} :  \mbox{x}^{(j_1)} \in D_1,\mbox{x}^{(j_2)} \in P_4,\mbox{x}^{(j_3)} \in Q_1,\mbox{x}^{(j_4)} \in R_2;\\
& \hspace{1cm}  \;\forall i\neq j_1,j_2,j_3,j_4\in\mathcal{I},\;  \mbox{x}^{(i)} \in D_1 \cup P_1 \cup P_2 \cup P_4   \cup Q_1 \cup R_1 \cup R_2 \}\liney
& + \mathbb{I}\{k>2; \;\exists \; j_1,j_2,j_3\in\mathcal{I} :  \mbox{x}^{(j_1)} \in P_2,\mbox{x}^{(j_2)} \in Q_1,\mbox{x}^{(j_3)} \in Q_2;\\
& \hspace{1cm}  \;\forall i\neq j_1,j_2,j_3\in\mathcal{I},\;  \mbox{x}^{(i)} \in  P_1 \cup P_2 \cup P_4   \cup Q_1\cup Q_2  \cup R_1 \cup R_2 \}\liney
& + \mathbb{I}\{k>3; \;\exists \; j_1,j_2,j_3,j_4\in\mathcal{I} :  \mbox{x}^{(j_1)} \in P_2,\mbox{x}^{(j_2)} \in P_4,\mbox{x}^{(j_3)} \in Q_1,\mbox{x}^{(j_4)} \in R_2;\\
& \hspace{1cm}  \;\forall i\neq j_1,j_2,j_3,j_4\in\mathcal{I},\;  \mbox{x}^{(i)} \in P_1 \cup P_2 \cup P_4   \cup Q_1 \cup R_1 \cup R_2 \}\liney
& + \mathbb{I}\{k>2; \;\exists \; j_1,j_2,j_3\in\mathcal{I} :  \mbox{x}^{(j_1)} \in D_1,\mbox{x}^{(j_2)} \in Q_2,\mbox{x}^{(j_3)} \in R_1;\\
& \hspace{1cm}  \;\forall i\neq j_1,j_2,j_3\in\mathcal{I},\;  \mbox{x}^{(i)} \in D_1 \cup P_1 \cup P_2 \cup P_4  \cup Q_2  \cup R_1 \cup R_2 \}\liney
& + \mathbb{I}\{k>3; \;\exists \; j_1,j_2,j_3,j_4\in\mathcal{I} :  \mbox{x}^{(j_1)} \in D_1,\mbox{x}^{(j_2)} \in P_4, \mbox{x}^{(j_3)} \in R_1, \mbox{x}^{(j_4)} \in R_2;\\
& \hspace{1cm}  \;\forall i\neq j_1,j_2,j_3,j_4\in\mathcal{I},\;  \mbox{x}^{(i)} \in D_1 \cup P_1 \cup P_2 \cup P_4   \cup R_1 \cup R_2 \}\liney
& + \mathbb{I}\{k>3; \;\exists \; j_1,j_2,j_3,j_4\in\mathcal{I} :  \mbox{x}^{(j_1)} \in P_1,\mbox{x}^{(j_2)} \in P_2,\mbox{x}^{(j_3)} \in Q_2,\mbox{x}^{(j_4)} \in R_1;\\
& \hspace{1cm}  \;\forall i\neq j_1,j_2,j_3,j_4\in\mathcal{I},\;  \mbox{x}^{(i)} \in  P_1 \cup P_2 \cup P_4  \cup Q_2  \cup R_1 \cup R_2 \}\liney
& + \mathbb{I}\{k>4; \;\exists \; j_1,j_2,j_3,j_4,j_5\in\mathcal{I} :  \mbox{x}^{(j_1)} \in P_1,\mbox{x}^{(j_2)} \in P_2, \\
& \hspace{2cm}\mbox{x}^{(j_3)} \in P_4,\mbox{x}^{(j_4)} \in R_1,\mbox{x}^{(j_5)} \in R_2;\\
& \hspace{1cm}  \;\forall i\neq j_1,j_2\in\mathcal{I},\;  \mbox{x}^{(i)} \in P_1 \cup P_2 \cup P_4  \cup R_1 \cup R_2 \}\line
&p\left(110111 \mid {\bf x}_k,\nu, T\right)\\
&= \mathbb{I}\{k>1; \;\exists \; j_1,j_2\in\mathcal{I} :  \mbox{x}^{(j_1)} \in M_2,\mbox{x}^{(j_2)} \in Q_1;\\
& \hspace{1cm}  \;\forall i\neq j_1,j_2\in\mathcal{I},\;  \mbox{x}^{(i)} \in D_3 \cup M_2 \cup P_1 \cup P_3 \cup P_4   \cup Q_1\cup Q_2  \cup R_1 \cup R_2 \}\liney
& + \mathbb{I}\{k>2; \;\exists \; j_1,j_2,j_3\in\mathcal{I} :  \mbox{x}^{(j_1)} \in M_2,\mbox{x}^{(j_2)} \in P_1,\mbox{x}^{(j_3)} \in R_1;\\
& \hspace{1cm}  \;\forall i\neq j_1,j_2,j_3\in\mathcal{I},\;  \mbox{x}^{(i)} \in  \{D_3 \cup M_2 \cup P_1 \cup P_3 \cup P_4  \cup Q_2  \cup R_1 \cup R_2 \}\liney
& + \mathbb{I}\{k>2; \;\exists \; j_1,j_2,j_3\in\mathcal{I} :  \mbox{x}^{(j_1)} \in D_3,\mbox{x}^{(j_2)} \in Q_1,\mbox{x}^{(j_3)} \in Q_2;\\
& \hspace{1cm}  \;\forall i\neq j_1,j_2,j_3\in\mathcal{I},\;  \mbox{x}^{(i)} \in  \{D_3  \cup P_1 \cup P_3 \cup P_4   \cup Q_1\cup Q_2  \cup R_1 \cup R_2 \}\liney
& + \mathbb{I}\{k>3; \;\exists \; j_1,j_2,j_3,j_4\in\mathcal{I} :  \mbox{x}^{(j_1)} \in D_3,\mbox{x}^{(j_2)} \in P_1,\mbox{x}^{(j_3)} \in Q_2,\mbox{x}^{(j_4)} \in R_1;\\
& \hspace{1cm}  \;\forall i\neq j_1,j_2,j_3,j_4\in\mathcal{I},\;  \mbox{x}^{(i)} \in  \{D_3 \cup P_1 \cup P_3 \cup P_4 \cup Q_2  \cup R_1 \cup R_2 \}\liney
& + \mathbb{I}\{k>2; \;\exists \; j_1,j_2,j_3\in\mathcal{I} :  \mbox{x}^{(j_1)} \in P_3,\mbox{x}^{(j_2)} \in Q_1,\mbox{x}^{(j_3)} \in Q_2;\\
& \hspace{1cm}  \;\forall i\neq j_1,j_2,j_3\in\mathcal{I},\;  \mbox{x}^{(i)} \in  \{P_1 \cup P_3 \cup P_4   \cup Q_1\cup Q_2  \cup R_1 \cup R_2 \}\liney
& + \mathbb{I}\{k>3; \;\exists \; j_1,j_2,j_3,j_4\in\mathcal{I} :  \mbox{x}^{(j_1)} \in P_1,\mbox{x}^{(j_2)} \in P_3,\mbox{x}^{(j_3)} \in Q_2,\mbox{x}^{(j_4)} \in R_1;\\
& \hspace{1cm}  \;\forall i\neq j_1,j_2,j_3,j_4\in\mathcal{I},\;  \mbox{x}^{(i)} \in  \{ P_1 \cup P_3 \cup P_4  \cup Q_2  \cup R_1 \cup R_2 \}\liney
& + \mathbb{I}\{k>2; \;\exists \; j_1,j_2,j_3\in\mathcal{I} :  \mbox{x}^{(j_1)} \in D_3,\mbox{x}^{(j_2)} \in Q_1,\mbox{x}^{(j_3)} \in R_2;\\
& \hspace{1cm}  \;\forall i\neq j_1,j_2,j_3\in\mathcal{I},\;  \mbox{x}^{(i)} \in  \{D_3 \cup P_1 \cup P_3 \cup P_4   \cup Q_1 \cup R_1 \cup R_2 \}\liney
& + \mathbb{I}\{k>3; \;\exists \; j_1,j_2,j_3,j_4\in\mathcal{I} :  \mbox{x}^{(j_1)} \in D_3,\mbox{x}^{(j_2)} \in P_1, \mbox{x}^{(j_3)} \in R_1, \mbox{x}^{(j_4)} \in R_2;\\
& \hspace{1cm}  \;\forall i\neq j_1,j_2,j_3,j_4\in\mathcal{I},\;  \mbox{x}^{(i)} \in  \{D_3 \cup P_1 \cup P_3 \cup P_4  \cup R_1 \cup R_2 \}\liney
& + \mathbb{I}\{k>3; \;\exists \; j_1,j_2,j_3,j_4\in\mathcal{I} :  \mbox{x}^{(j_1)} \in P_3,\mbox{x}^{(j_2)} \in P_4,\mbox{x}^{(j_3)} \in Q_1,\mbox{x}^{(j_4)} \in R_2;\\
& \hspace{1cm}  \;\forall i\neq j_1,j_2,j_3,j_4\in\mathcal{I},\;  \mbox{x}^{(i)} \in  \{ P_1 \cup P_3 \cup P_4   \cup Q_1\cup R_1 \cup R_2 \}\liney
& + \mathbb{I}\{k>4; \;\exists \; j_1,j_2,j_3,j_4,j_5\in\mathcal{I} :  \mbox{x}^{(j_1)} \in P_1,\mbox{x}^{(j_2)} \in P_3, \\
& \hspace{2cm}\mbox{x}^{(j_3)} \in P_4,\mbox{x}^{(j_4)} \in R_1,\mbox{x}^{(j_5)} \in R_2;\\
& \hspace{1cm}  \;\forall i\neq j_1,j_2\in\mathcal{I},\;  \mbox{x}^{(i)} \in  \{ P_1 \cup P_3 \cup P_4  \cup R_1 \cup R_2 \}
\end{align*}


\subsection{011110}
\begin{align*}
&p\left(011110 \mid {\bf x}_k,\nu, T\right)\\
&= \mathbb{I}\{k>0; \;\exists \; j_1\in\mathcal{I} :  \mbox{x}^{(j_1)} \in U;\\
& \hspace{1cm}  \;\forall i\neq j_1\in\mathcal{I},\;  \mbox{x}^{(i)} \in D_1 \cup D_2 \cup D_3 \cup P_1 \cup P_2 \cup P_3 \cup P_4 \cup S_1\cup S_2  \cup U \}\liney
& + \mathbb{I}\{k>1; \;\exists \; j_1,j_2\in\mathcal{I} :  \mbox{x}^{(j_1)} \in S_1,\mbox{x}^{(j_2)} \in S_2;\\
& \hspace{1cm}  \;\forall i\neq j_1,j_2\in\mathcal{I},\;  \mbox{x}^{(i)} \in D_1 \cup D_2 \cup D_3 \cup P_1 \cup P_2 \cup P_3 \cup P_4 \cup S_1\cup S_2  \}\liney
& + \mathbb{I}\{k>1; \;\exists \; j_1,j_2\in\mathcal{I} : \mbox{x}^{(j_1)} \in D_1, \mbox{x}^{(j_2)} \in S_1;\\
& \hspace{1cm}  \;\forall i\neq j_1,j_2\in\mathcal{I},\;  \mbox{x}^{(i)} \in D_1 \cup D_2 \cup D_3 \cup P_1 \cup P_2 \cup P_3 \cup P_4 \cup S_1 \}\liney
& + \mathbb{I}\{k>1; \;\exists \; j_1,j_2\in\mathcal{I} :  \mbox{x}^{(j_1)} \in P_4,\mbox{x}^{(j_2)} \in S_1;\\
& \hspace{1cm}  \;\forall i\neq j_1,j_2\in\mathcal{I},\;  \mbox{x}^{(i)} \in D_2 \cup D_3 \cup P_1 \cup P_2 \cup P_3 \cup P_4 \cup S_1 \}\liney
& + \mathbb{I}\{k>1; \;\exists \; j_1,j_2\in\mathcal{I} :  \mbox{x}^{(j_1)} \in D_1, \mbox{x}^{(j_2)} \in S_2;\\
& \hspace{1cm}  \;\forall i\neq j_1,j_2\in\mathcal{I},\;  \mbox{x}^{(i)} \in D_1 \cup D_2 \cup D_3 \cup P_1 \cup P_2 \cup P_3 \cup P_4\cup S_2  \}\liney
& + \mathbb{I}\{k>1; \;\exists \; j_1,j_2\in\mathcal{I} :  \mbox{x}^{(j_1)} \in P_1,\mbox{x}^{(j_2)} \in S_2;\\
& \hspace{1cm}  \;\forall i\neq j_1,j_2\in\mathcal{I},\;  \mbox{x}^{(i)} \in D_2 \cup D_3 \cup P_1 \cup P_2 \cup P_3 \cup P_4 \cup S_2  \}\liney
& + \mathbb{I}\{k>1; \;\exists \; j_1,j_2\in\mathcal{I} :   \mbox{x}^{(j_1)} \in D1, \mbox{x}^{(j_2)} \in D3;\\
& \hspace{1cm}  \;\forall i\neq j_1,j_2\in\mathcal{I},\;  \mbox{x}^{(i)} \in D_1 \cup D_2 \cup D_3 \cup P_1 \cup P_2 \cup P_3 \cup P_4 \}\liney
& + \mathbb{I}\{k>2; \;\exists \; j_1,j_2,j_3\in\mathcal{I} :  \mbox{x}^{(j_1)} \in D_1,\mbox{x}^{(j_2)} \in D_2, \mbox{x}^{(j_3)} \in P_4;\\
& \hspace{1cm}  \;\forall i\neq j_1,j_2,j_3\in\mathcal{I},\;  \mbox{x}^{(i)} \in D_1 \cup D_2 \cup P_1 \cup P_2 \cup P_3 \cup P_4\}\liney
& + \mathbb{I}\{k>2; \;\exists \; j_1,j_2,j_3\in\mathcal{I} :  \mbox{x}^{(j_1)} \in D_2,\mbox{x}^{(j_2)} \in D_3,\mbox{x}^{(j_3)} \in P_1;\\
& \hspace{1cm}  \;\forall i\neq j_1,j_2,j_3\in\mathcal{I},\;  \mbox{x}^{(i)} \in D_2 \cup D_3 \cup P_1 \cup P_2 \cup P_3 \cup P_4 \}\liney
& + \mathbb{I}\{k>2; \;\exists \; j_1,j_2,j_3\in\mathcal{I} :  \mbox{x}^{(j_1)} \in D_1,\mbox{x}^{(j_2)} \in P_3,\mbox{x}^{(j_3)} \in P_4;\\
& \hspace{1cm}  \;\forall i\neq j_1,j_2,j_3\in\mathcal{I},\;  \mbox{x}^{(i)} \in D_1 \cup P_1 \cup P_2 \cup P_3 \cup P_4\}\liney
& + \mathbb{I}\{k>2; \;\exists \; j_1,j_2,j_3\in\mathcal{I} :  \mbox{x}^{(j_1)} \in D_2,\mbox{x}^{(j_2)} \in P_1,\mbox{x}^{(j_3)} \in P_4;\\
& \hspace{1cm}  \;\forall i\neq j_1,j_2,j_3\in\mathcal{I},\;  \mbox{x}^{(i)} \in D_2 \cup P_1 \cup P_2 \cup P_3 \cup P_4 \}\liney
& + \mathbb{I}\{k>2; \;\exists \; j_1,j_2,j_3\in\mathcal{I} :  \mbox{x}^{(j_1)} \in D_3,\mbox{x}^{(j_2)} \in P_1,\mbox{x}^{(j_3)} \in P_2;\\
& \hspace{1cm}  \;\forall i\neq j_1,j_2,j_3\in\mathcal{I},\;  \mbox{x}^{(i)} \in D_3 \cup P_1 \cup P_2 \cup P_3 \cup P_4 \}\liney
& + \mathbb{I}\{k>3; \;\exists \; j_1,j_2,j_3,j_4\in\mathcal{I} :  \mbox{x}^{(j_1)} \in P_1,\mbox{x}^{(j_2)} \in P_2,\mbox{x}^{(j_3)} \in P_3,\mbox{x}^{(j_4)} \in P_4;\\
& \hspace{1cm}  \;\forall i\neq j_1,j_3,j_3,j_4\in\mathcal{I},\;  \mbox{x}^{(i)} \in P_1 \cup P_2 \cup P_3 \cup P_4 \}
\end{align*}


\subsection{111100,001111}
\begin{align*}
&p\left(111100 \mid {\bf x}_k,\nu, T\right)\\
&= \mathbb{I}\{k>0; \;\exists \; j_1\in\mathcal{I} :  \mbox{x}^{(j_1)} \in k_1;\\
& \hspace{1cm}  \;\forall i\neq j_1\in\mathcal{I},\;  \mbox{x}^{(i)} \in D_1 \cup D_2 \cup M_1 \cup k_1 \cup P_1 \cup P_2 \cup P_3 \cup Q_1 \cup R_1\cup S_1 \}\liney
& + \mathbb{I}\{k>1; \;\exists \; j_1,j_2\in\mathcal{I} :  \mbox{x}^{(j_1)} \in M_1, \mbox{x}^{(j_2)} \in S_1;\\
& \hspace{1cm}  \;\forall i\neq j_1,j_2\in\mathcal{I},\;  \mbox{x}^{(i)} \in D_1 \cup D_2 \cup M_1 \cup P_1 \cup P_2 \cup P_3 \cup Q_1 \cup R_1\cup S_1 \}\liney
& + \mathbb{I}\{k>1; \;\exists \; j_1,j_2\in\mathcal{I} :  \mbox{x}^{(j_1)} \in D_2, \mbox{x}^{(j_2)} \in M_1;\\
& \hspace{1cm}  \;\forall i\neq j_1,j_2\in\mathcal{I},\;  \mbox{x}^{(i)} \in D_1 \cup D_2 \cup M_1 \cup P_1 \cup P_2 \cup P_3 \cup Q_1 \cup R_1 \}\liney
& + \mathbb{I}\{k>1; \;\exists \; j_1,j_2\in\mathcal{I} :  \mbox{x}^{(j_1)} \in M_1, \mbox{x}^{(j_2)} \in P_3;\\
& \hspace{1cm}  \;\forall i\neq j_1,j_2\in\mathcal{I},\;  \mbox{x}^{(i)} \in D_1  \cup M_1 \cup P_1 \cup P_2 \cup P_3 \cup Q_1 \cup R_1\}\liney
& + \mathbb{I}\{k>1; \;\exists \; j_1,j_2\in\mathcal{I} :  \mbox{x}^{(j_1)} \in Q_1, \mbox{x}^{(j_2)} \in S_1;\\
& \hspace{1cm}  \;\forall i\neq j_1,j_2\in\mathcal{I},\;  \mbox{x}^{(i)} \in D_1 \cup D_2 \cup P_1 \cup P_2 \cup P_3 \cup Q_1 \cup R_1\cup S_1 \}\liney
& + \mathbb{I}\{k>1; \;\exists \; j_1,j_2\in\mathcal{I} :  \mbox{x}^{(j_1)} \in D_2, \mbox{x}^{(j_2)} \in Q_1;\\
& \hspace{1cm}  \;\forall i\neq j_1,j_2\in\mathcal{I},\;  \mbox{x}^{(i)} \in D_1 \cup D_2 \cup P_1 \cup P_2 \cup P_3 \cup Q_1 \cup R_1 \}\liney
& + \mathbb{I}\{k>2; \;\exists \; j_1,j_2,j_3\in\mathcal{I} :  \mbox{x}^{(j_1)} \in P_2,\mbox{x}^{(j_2)} \in P_3, \mbox{x}^{(j_3)} \in Q_1;\\
& \hspace{1cm}  \;\forall i\neq j_1,j_2,j_3\in\mathcal{I},\;  \mbox{x}^{(i)} \in D_1 \cup P_1 \cup P_2 \cup P_3 \cup Q_1 \cup R_1 \}\liney
& + \mathbb{I}\{k>1; \;\exists \; j_1,j_2\in\mathcal{I} :  \mbox{x}^{(j_1)} \in R_1, \mbox{x}^{(j_2)} \in S_1;\\
& \hspace{1cm}  \;\forall i\neq j_1,j_2\in\mathcal{I},\;  \mbox{x}^{(i)} \in D_1 \cup D_2 \cup P_1 \cup P_2 \cup P_3 \cup R_1\cup S_1 \}\liney
& + \mathbb{I}\{k>2; \;\exists \; j_1,j_2,j_3\in\mathcal{I} :  \mbox{x}^{(j_1)} \in D_1,\mbox{x}^{(j_2)} \in D_2,\mbox{x}^{(j_3)} \in R_1;\\
& \hspace{1cm}  \;\forall i\neq j_1,j_2,j_3\in\mathcal{I},\;  \mbox{x}^{(i)} \in D_1 \cup D_2 \cup P_1 \cup P_2 \cup P_3  \cup R_1 \}\liney
& + \mathbb{I}\{k>2; \;\exists \; j_1,j_2,j_3\in\mathcal{I} :  \mbox{x}^{(j_1)} \in D_1,\mbox{x}^{(j_2)} \in P_3,\mbox{x}^{(j_3)} \in R_1;\\
& \hspace{1cm}  \;\forall i\neq j_1,j_2,j_3\in\mathcal{I},\;  \mbox{x}^{(i)} \in D_1 \cup P_1 \cup P_2 \cup P_3 \cup R_1 \}\liney
& + \mathbb{I}\{k>2; \;\exists \; j_1,j_2,j_3\in\mathcal{I} :  \mbox{x}^{(j_1)} \in D_2,\mbox{x}^{(j_2)} \in P_1,\mbox{x}^{(j_3)} \in R_1;\\
& \hspace{1cm}  \;\forall i\neq j_1,j_2,j_3\in\mathcal{I},\;  \mbox{x}^{(i)} \in D_2 \cup P_1 \cup P_2 \cup P_3 \cup R_1 \}\liney
& + \mathbb{I}\{k>3; \;\exists \; j_1,j_2,j_3,j_4\in\mathcal{I} :  \mbox{x}^{(j_1)} \in P_1,\mbox{x}^{(j_2)} \in P_2,\mbox{x}^{(j_3)} \in P_3,\mbox{x}^{(j_4)} \in R_1;\\
& \hspace{1cm}  \;\forall i\neq j_1,j_2,j_3,j_4\in\mathcal{I},\;  \mbox{x}^{(i)} \in P_1 \cup P_2 \cup P_3 \cup R_1 \}\line
&p\left(001111 \mid {\bf x}_k,\nu, T\right)\\
&= \mathbb{I}\{k>0; \;\exists \; j_1\in\mathcal{I} :  \mbox{x}^{(j_1)} \in k_2;\\
& \hspace{1cm}  \;\forall i\neq j_1\in\mathcal{I},\;  \mbox{x}^{(i)} \in D_2 \cup D_3 \cup M_2 \cup k_2 \cup P_2 \cup P_3 \cup P_4 \cup Q_2 \cup R_2\cup S_2 \}\liney
& + \mathbb{I}\{k>1; \;\exists \; j_1,j_2\in\mathcal{I} :  \mbox{x}^{(j_1)} \in M_2, \mbox{x}^{(j_2)} \in S_2;\\
& \hspace{1cm}  \;\forall i\neq j_1,j_2\in\mathcal{I},\;  \mbox{x}^{(i)} \in D_2 \cup D_3 \cup M_2 \cup P_2 \cup P_3 \cup P_4 \cup Q_2 \cup R_2\cup S_2 \}\liney
& + \mathbb{I}\{k>1; \;\exists \; j_1,j_2\in\mathcal{I} :  \mbox{x}^{(j_1)} \in D_2, \mbox{x}^{(j_2)} \in M_2;\\
& \hspace{1cm}  \;\forall i\neq j_1,j_2\in\mathcal{I},\;  \mbox{x}^{(i)} \in D_2 \cup D_3 \cup M_2 \cup P_2 \cup P_3 \cup P_4 \cup Q_2 \cup R_2 \}\liney
& + \mathbb{I}\{k>1; \;\exists \; j_1,j_2\in\mathcal{I} :  \mbox{x}^{(j_1)} \in M_2, \mbox{x}^{(j_2)} \in P_2;\\
& \hspace{1cm}  \;\forall i\neq j_1,j_2\in\mathcal{I},\;  \mbox{x}^{(i)} \in D_3  \cup M_2 \cup P_2 \cup P_3 \cup P_4 \cup Q_2 \cup R_2\}\liney
& + \mathbb{I}\{k>1; \;\exists \; j_1,j_2\in\mathcal{I} :  \mbox{x}^{(j_1)} \in Q_2, \mbox{x}^{(j_2)} \in S_2;\\
& \hspace{1cm}  \;\forall i\neq j_1,j_2\in\mathcal{I},\;  \mbox{x}^{(i)} \in D_2 \cup D_3 \cup P_2 \cup P_3 \cup P_4 \cup Q_2 \cup R_2\cup S_2 \}\liney
& + \mathbb{I}\{k>1; \;\exists \; j_1,j_2\in\mathcal{I} :  \mbox{x}^{(j_1)} \in D_2, \mbox{x}^{(j_2)} \in Q_2;\\
& \hspace{1cm}  \;\forall i\neq j_1,j_2\in\mathcal{I},\;  \mbox{x}^{(i)} \in D_2 \cup D_3 \cup P_2 \cup P_3 \cup P_4 \cup Q_2 \cup R_2 \}\liney
& + \mathbb{I}\{k>2; \;\exists \; j_1,j_2,j_3\in\mathcal{I} :  \mbox{x}^{(j_1)} \in P_2,\mbox{x}^{(j_2)} \in P_3, \mbox{x}^{(j_3)} \in Q_2;\\
& \hspace{1cm}  \;\forall i\neq j_1,j_2,j_3\in\mathcal{I},\;  \mbox{x}^{(i)} \in D_3 \cup P_2 \cup P_3 \cup P_4 \cup Q_2 \cup R_2 \}\liney
& + \mathbb{I}\{k>1; \;\exists \; j_1,j_2\in\mathcal{I} :  \mbox{x}^{(j_1)} \in R_2, \mbox{x}^{(j_2)} \in S_2;\\
& \hspace{1cm}  \;\forall i\neq j_1,j_2\in\mathcal{I},\;  \mbox{x}^{(i)} \in D_2 \cup D_3 \cup P_2 \cup P_3 \cup P_4 \cup R_2\cup S_2 \}\liney
& + \mathbb{I}\{k>2; \;\exists \; j_1,j_2,j_3\in\mathcal{I} :  \mbox{x}^{(j_1)} \in D_2,\mbox{x}^{(j_2)} \in D_3,\mbox{x}^{(j_3)} \in R_2;\\
& \hspace{1cm}  \;\forall i\neq j_1,j_2,j_3\in\mathcal{I},\;  \mbox{x}^{(i)} \in D_2 \cup D_3 \cup P_2 \cup P_3 \cup P_4  \cup R_2 \}\liney
& + \mathbb{I}\{k>2; \;\exists \; j_1,j_2,j_3\in\mathcal{I} :  \mbox{x}^{(j_1)} \in D_3,\mbox{x}^{(j_2)} \in P_2,\mbox{x}^{(j_3)} \in R_2;\\
& \hspace{1cm}  \;\forall i\neq j_1,j_2,j_3\in\mathcal{I},\;  \mbox{x}^{(i)} \in D_3 \cup P_2 \cup P_3 \cup P_4 \cup R_2 \}\liney
& + \mathbb{I}\{k>2; \;\exists \; j_1,j_2,j_3\in\mathcal{I} :  \mbox{x}^{(j_1)} \in D_2,\mbox{x}^{(j_2)} \in P_4,\mbox{x}^{(j_3)} \in R_2;\\
& \hspace{1cm}  \;\forall i\neq j_1,j_2,j_3\in\mathcal{I},\;  \mbox{x}^{(i)} \in D_2 \cup P_2 \cup P_3 \cup P_4 \cup R_2 \}\liney
& + \mathbb{I}\{k>3; \;\exists \; j_1,j_2,j_3,j_4\in\mathcal{I} :  \mbox{x}^{(j_1)} \in P_2,\mbox{x}^{(j_2)} \in P_3,\mbox{x}^{(j_3)} \in P_4,\mbox{x}^{(j_4)} \in R_2;\\
& \hspace{1cm}  \;\forall i\neq j_1,j_2,j_3,j_4\in\mathcal{I},\;  \mbox{x}^{(i)} \in P_2 \cup P_3 \cup P_4 \cup R_2 \}
\end{align*}


\subsection{ 111101,101111}
\begin{align*}
&p\left(111101 \mid {\bf x}_k,\nu, T\right)\\
&= \mathbb{I}\{k>1; \;\exists \; j_1,j_2\in\mathcal{I} :  \mbox{x}^{(j_1)} \in k_1, \mbox{x}^{(j_1)} \in R_2;\\
& \hspace{1cm}  \;\forall i\neq j_1,j_2\in\mathcal{I},\;  \mbox{x}^{(i)} \in D_1 \cup D_2 \cup M_1 \cup k_1 \cup P_1 \cup P_2 \cup P_3 \cup Q_1 \cup R_1 \cup R_2\cup S_1 \}\liney
& + \mathbb{I}\{k>2; \;\exists \; j_1,j_2,j_3\in\mathcal{I} :  \mbox{x}^{(j_1)} \in M_1,\mbox{x}^{(j_2)} \in R_2, \mbox{x}^{(j_3)} \in S_1;\\
& \hspace{1cm}  \;\forall i\neq j_1,j_2,j_3\in\mathcal{I},\;  \mbox{x}^{(i)} \in D_1 \cup D_2 \cup M_1 \cup P_1 \cup P_2 \cup P_3 \cup Q_1 \cup R_1 \cup R_2\cup S_1 \}\liney
& + \mathbb{I}\{k>2; \;\exists \; j_1,j_2,j_3\in\mathcal{I} :  \mbox{x}^{(j_1)} \in D_2, \mbox{x}^{(j_2)} \in M_1,\mbox{x}^{(j_3)} \in R_2;\\
& \hspace{1cm}  \;\forall i\neq j_1,j_2,j_3\in\mathcal{I},\;  \mbox{x}^{(i)} \in D_1 \cup D_2 \cup M_1 \cup P_1 \cup P_2 \cup P_3 \cup Q_1 \cup R_1 \cup R_2 \}\liney
& + \mathbb{I}\{k>2; \;\exists \; j_1,j_2,j_3\in\mathcal{I} :  \mbox{x}^{(j_1)} \in M_1, \mbox{x}^{(j_2)} \in P_3,\mbox{x}^{(j_3)} \in R_2;\\
& \hspace{1cm}  \;\forall i\neq j_1,j_2,j_3\in\mathcal{I},\;  \mbox{x}^{(i)} \in D_1  \cup M_1 \cup P_1 \cup P_2 \cup P_3 \cup Q_1 \cup R_1 \cup R_2\}\liney
& + \mathbb{I}\{k>2; \;\exists \; j_1,j_2,j_3\in\mathcal{I} :  \mbox{x}^{(j_1)} \in Q_1,\mbox{x}^{(j_2)} \in R_2, \mbox{x}^{(j_3)} \in S_1;\\
& \hspace{1cm}  \;\forall i\neq j_1,j_2,j_3\in\mathcal{I},\;  \mbox{x}^{(i)} \in D_1 \cup D_2 \cup P_1 \cup P_2 \cup P_3 \cup Q_1 \cup R_1 \cup R_2 \cup S_1 \}\liney
& + \mathbb{I}\{k>2; \;\exists \; j_1,j_2,j_3\in\mathcal{I} :  \mbox{x}^{(j_1)} \in D_2, \mbox{x}^{(j_2)} \in Q_1,\mbox{x}^{(j_3)} \in R_2;\\
& \hspace{1cm}  \;\forall i\neq j_1,j_2,j_3\in\mathcal{I},\;  \mbox{x}^{(i)} \in D_1 \cup D_2 \cup P_1 \cup P_2 \cup P_3 \cup Q_1 \cup R_1 \cup R_2 \}\liney
& + \mathbb{I}\{k>3; \;\exists \; j_1,j_2,j_3,j_4\in\mathcal{I} :  \mbox{x}^{(j_1)} \in P_2,\mbox{x}^{(j_2)} \in P_3, \mbox{x}^{(j_3)} \in Q_1, \mbox{x}^{(j_4)} \in R_2;\\
& \hspace{1cm}  \;\forall i\neq j_1,j_2,j_3,j_4\in\mathcal{I},\;  \mbox{x}^{(i)} \in D_1 \cup P_1 \cup P_2 \cup P_3 \cup Q_1 \cup R_1 \cup R_2 \}\liney
& + \mathbb{I}\{k>2; \;\exists \; j_1,j_2,j_3\in\mathcal{I} :  \mbox{x}^{(j_1)} \in R_1,\mbox{x}^{(j_2)} \in R_2, \mbox{x}^{(j_3)} \in S_1;\\
& \hspace{1cm}  \;\forall i\neq j_1,j_2,j_3\in\mathcal{I},\;  \mbox{x}^{(i)} \in D_1 \cup D_2 \cup P_1 \cup P_2 \cup P_3 \cup R_1 \cup R_2 \cup S_1 \}\liney
& + \mathbb{I}\{k>3; \;\exists \; j_1,j_2,j_3,j_4\in\mathcal{I} :  \mbox{x}^{(j_1)} \in D_1,\mbox{x}^{(j_2)} \in D_2,\mbox{x}^{(j_3)} \in R_1,\mbox{x}^{(j_4)} \in R_2;\\
& \hspace{1cm}  \;\forall i\neq j_1,j_2,j_3,j_4\in\mathcal{I},\;  \mbox{x}^{(i)} \in D_1 \cup D_2 \cup P_1 \cup P_2 \cup P_3  \cup R_1 \cup R_2\}\liney
& + \mathbb{I}\{k>3; \;\exists \; j_1,j_2,j_3,j_4\in\mathcal{I} :  \mbox{x}^{(j_1)} \in D_1,\mbox{x}^{(j_2)} \in P_3,\mbox{x}^{(j_3)} \in R_1, \mbox{x}^{(j_4)} \in R_2;\\
& \hspace{1cm}  \;\forall i\neq j_1,j_2,j_3,j_4\in\mathcal{I},\;  \mbox{x}^{(i)} \in D_1 \cup P_1 \cup P_2 \cup P_3 \cup R_1 \cup R_2 \}\liney
& + \mathbb{I}\{k>3; \;\exists \; j_1,j_2,j_3,j_4\in\mathcal{I} :  \mbox{x}^{(j_1)} \in D_2,\mbox{x}^{(j_2)} \in P_1, \mbox{x}^{(j_3)} \in R_1,\mbox{x}^{(j_4)} \in R_2;\\
& \hspace{1cm}  \;\forall i\neq j_1,j_2,j_3,j_4\in\mathcal{I},\;  \mbox{x}^{(i)} \in D_2 \cup P_1 \cup P_2 \cup P_3 \cup R_1 \cup R_2\}\liney
& + \mathbb{I}\{k>4; \;\exists \; j_1,j_2,j_3,j_4,j_5\in\mathcal{I} :  \mbox{x}^{(j_1)} \in P_1,\mbox{x}^{(j_2)} \in P_2,\mbox{x}^{(j_3)} \in P_3,\mbox{x}^{(j_4)} \in R_1, \mbox{x}^{(j_5)} \in R_2;\\
& \hspace{1cm}  \;\forall i\neq j_1,j_2,j_3,j_4,j_5\in\mathcal{I},\;  \mbox{x}^{(i)} \in P_1 \cup P_2 \cup P_3 \cup R_1 \cup R_2\}\line
&p\left(101111 \mid {\bf x}_k,\nu, T\right)\\
&= \mathbb{I}\{k>1; \;\exists \; j_1,j_2\in\mathcal{I} :  \mbox{x}^{(j_1)} \in k_2, \mbox{x}^{(j_2)} \in R_1;\\
& \hspace{1cm}  \;\forall i\neq j_1,j_2\in\mathcal{I},\;  \mbox{x}^{(i)} \in D_2 \cup D_3 \cup M_2 \cup k_2 \cup P_2 \cup P_3 \cup P_4 \cup Q_2 \cup R_1 \cup R_2\cup S_2 \}\liney
& + \mathbb{I}\{k>2; \;\exists \; j_1,j_2,j_3\in\mathcal{I} :  \mbox{x}^{(j_1)} \in M_2,\mbox{x}^{(j_2)} \in R_1,  \mbox{x}^{(j_3)} \in S_2;\\
& \hspace{1cm}  \;\forall i\neq j_1,j_2,j_3\in\mathcal{I},\;  \mbox{x}^{(i)} \in D_2 \cup D_3 \cup M_2 \cup P_2 \cup P_3 \cup P_4 \cup Q_2 \cup R_1 \cup R_2\cup S_2 \}\liney
& + \mathbb{I}\{k>2; \;\exists \; j_1,j_2,j_3\in\mathcal{I} :  \mbox{x}^{(j_1)} \in D_2, \mbox{x}^{(j_2)} \in M_2,\mbox{x}^{(j_3)} \in R_1;\\
& \hspace{1cm}  \;\forall i\neq j_1,j_2,j_3\in\mathcal{I},\;  \mbox{x}^{(i)} \in D_2 \cup D_3 \cup M_2 \cup P_2 \cup P_3 \cup P_4 \cup Q_2 \cup R_1 \cup R_2 \}\liney
& + \mathbb{I}\{k>2; \;\exists \; j_1,j_2,j_3\in\mathcal{I} :  \mbox{x}^{(j_1)} \in M_2, \mbox{x}^{(j_2)} \in P_2, \mbox{x}^{(j_3)} \in R_1;\\
& \hspace{1cm}  \;\forall i\neq j_1,j_2,j_3\in\mathcal{I},\;  \mbox{x}^{(i)} \in D_3  \cup M_2 \cup P_2 \cup P_3 \cup P_4 \cup Q_2 \cup R_1 \cup R_2\}\liney
& + \mathbb{I}\{k>2; \;\exists \; j_1,j_2,j_3\in\mathcal{I} :  \mbox{x}^{(j_1)} \in Q_2, \mbox{x}^{(j_2)} \in R_1,  \mbox{x}^{(j_3)} \in S_2;\\
& \hspace{1cm}  \;\forall i\neq j_1,j_2,j_3\in\mathcal{I},\;  \mbox{x}^{(i)} \in D_2 \cup D_3 \cup P_2 \cup P_3 \cup P_4 \cup Q_2 \cup R_1 \cup R_2\cup S_2 \}\liney
& + \mathbb{I}\{k>2; \;\exists \; j_1,j_2,j_3\in\mathcal{I} :  \mbox{x}^{(j_1)} \in D_2, \mbox{x}^{(j_2)} \in Q_2, \mbox{x}^{(j_3)} \in R_1;\\
& \hspace{1cm}  \;\forall i\neq j_1,j_2,j_3\in\mathcal{I},\;  \mbox{x}^{(i)} \in D_2 \cup D_3 \cup P_2 \cup P_3 \cup P_4 \cup Q_2 \cup R_1 \cup R_2 \}\liney
& + \mathbb{I}\{k>3; \;\exists \; j_1,j_2,j_3,j_4\in\mathcal{I} :  \mbox{x}^{(j_1)} \in P_2,\mbox{x}^{(j_2)} \in P_3, \mbox{x}^{(j_3)} \in Q_2, \mbox{x}^{(j_4)} \in R_1;\\
& \hspace{1cm}  \;\forall i\neq j_1,j_2,j_3,j_4\in\mathcal{I},\;  \mbox{x}^{(i)} \in D_3 \cup P_2 \cup P_3 \cup P_4 \cup Q_2 \cup R_1 \cup R_2 \}\liney
& + \mathbb{I}\{k>2; \;\exists \; j_1,j_2,j_3\in\mathcal{I} :  \mbox{x}^{(j_1)} \in R_1, \mbox{x}^{(j_2)} \in R_2, \mbox{x}^{(j_3)} \in S_2;\\
& \hspace{1cm}  \;\forall i\neq j_1,j_2,j_3\in\mathcal{I},\;  \mbox{x}^{(i)} \in D_2 \cup D_3 \cup P_2 \cup P_3 \cup P_4 \cup R_1 \cup R_2\cup S_2 \}\liney
& + \mathbb{I}\{k>3; \;\exists \; j_1,j_2,j_3,j_4\in\mathcal{I} :  \mbox{x}^{(j_1)} \in D_2,\mbox{x}^{(j_2)} \in D_3, \mbox{x}^{(j_3)} \in R_1, \mbox{x}^{(j_4)} \in R_2;\\
& \hspace{1cm}  \;\forall i\neq j_1,j_2,j_3,j_4\in\mathcal{I},\;  \mbox{x}^{(i)} \in D_2 \cup D_3 \cup P_2 \cup P_3 \cup P_4 \cup R_1 \cup R_2 \}\liney
& + \mathbb{I}\{k>3; \;\exists \; j_1,j_2,j_3,j_4\in\mathcal{I} :  \mbox{x}^{(j_1)} \in D_3,\mbox{x}^{(j_2)} \in P_2, \mbox{x}^{(j_3)} \in R_1, \mbox{x}^{(j_4)} \in R_2;\\
& \hspace{1cm}  \;\forall i\neq j_1,j_2,j_3,j_4\in\mathcal{I},\;  \mbox{x}^{(i)} \in D_3 \cup P_2 \cup P_3 \cup P_4 \cup R_1 \cup R_2 \}\liney
& + \mathbb{I}\{k>3; \;\exists \; j_1,j_2,j_3,j_4\in\mathcal{I} :  \mbox{x}^{(j_1)} \in D_2,\mbox{x}^{(j_2)} \in P_4, \mbox{x}^{(j_3)} \in R_1, \mbox{x}^{(j_4)} \in R_2;\\
& \hspace{1cm}  \;\forall i\neq j_1,j_2,j_3,j_4\in\mathcal{I},\;  \mbox{x}^{(i)} \in D_2 \cup P_2 \cup P_3 \cup P_4 \cup R_1 \cup R_2 \}\liney
& + \mathbb{I}\{k>4; \;\exists \; j_1,j_2,j_3,j_4,j_5\in\mathcal{I} :  \mbox{x}^{(j_1)} \in P_2,\mbox{x}^{(j_2)} \in P_3,\mbox{x}^{(j_3)} \in P_4, \mbox{x}^{(j_4)} \in R_1, \mbox{x}^{(j_5)} \in R_2;\\
& \hspace{1cm}  \;\forall i\neq j_1,j_2,j_3,j_4,j_5\in\mathcal{I},\;  \mbox{x}^{(i)} \in P_2 \cup P_3 \cup P_4 \cup R_1 \cup R_2 \}
\end{align*}


\subsection{ 111110,011111}



\subsection{ 111111}
All the combinations of introductions that have not been considered above will generate the data $(111111)$.

This example illustrates the difficulty in obtaining a closed-form expression of the likelihood even for our simple introduction and spread model.  It is important to note that the above strategy for computing the likelihood is a special case and may not be applicable under different model assumptions.

\section{Data}

The following table contains measurements  of presence of worms or cocoons, $Y^{(i)}_{gr}$, for quadrat $1\leq i \leq 6$ along road $1\leq r\leq n_g$ in group $1\leq g \leq 8$, and associated road age $T_{gr}$.  This data is described in the main body of the paper.
\\

\singlespacing{
\begin{longtable}{c c c c c }
\hline 
\multicolumn{1}{c}{\textbf{group ($g$)}} 
& 
\multicolumn{1}{c}{\textbf{ transect $r$}} 
& 
\multicolumn{1}{c}{\textbf{plot $i$}}  
&
\multicolumn{1}{c}{\textbf{ road age $T_{gr}$}} 
& 
\multicolumn{1}{c}{\textbf{worm or cocoon presence $Y^{(i)}_{gr}$}} \\ 
\hline 
\endfirsthead

\multicolumn{5}{c}%
{{\bfseries \tablename\ \thetable{} -- continued from previous page}} \\
\hline 
\multicolumn{1}{c}{\textbf{group $g$}} 
& 
\multicolumn{1}{c}{\textbf{ transect $r$}} 
& 
\multicolumn{1}{c}{\textbf{plot $i$}}  
&
\multicolumn{1}{c}{\textbf{ road age $T_{gr}$}} 
& 
\multicolumn{1}{c}{\textbf{worm or cocoon presence $Y^{(i)}_{gr}$}} \\ 
\hline 
\endhead

\hline \multicolumn{5}{c}{{Continued on next page}} \\
\endfoot

\hline \hline
\endlastfoot
1       & 1        & 1    & 19      & 0                  \\
1       & 1        & 2    & 19      & 0                  \\
1       & 1        & 3    & 19      & 0                  \\
1       & 1        & 4    & 19      & 0                  \\
1       & 1        & 5    & 19      & 0                  \\
1       & 1        & 6    & 19      & 0                  \\
1       & 2        & 1    & 26      & 0                  \\
1       & 2        & 2    & 26      & 0                  \\
\end{longtable}
}

%
%